\documentclass[11pt]{article}

\usepackage{amsmath,amscd} 
\usepackage{amssymb}
\usepackage{amsthm}
\usepackage{fullpage}

\newtheorem{thm}{Theorem}[section]
\newtheorem{proposition}[thm]{Proposition}
\newtheorem{prop}[thm]{Proposition}
\newtheorem{defn}[thm]{Definition}

\newtheorem{lem}[thm]{Lemma}
\newtheorem{claim}[thm]{Claim}
\newtheorem{cor}[thm]{Corollary}

\newcommand{\abs}[1]{\mbox{$ | #1 | $}}

\newcommand{\PosSLP}{{\mathrm PosSLP}}

\def\nat{{\mathbb N}}
 \def\real{{\mathbb R}}
 \def\rat{{\mathbb Q}}

\newcommand{\norm}[1]{ \mbox{$\| #1 \|$}}
\newcommand{\norminf}[1]{{\mbox{$\norm{#1}_\infty$}}}

\DeclareMathOperator{\sep}{sep}

\DeclareMathOperator{\cond}{cond}

\begin{document}

\title{Polynomial Time Algorithms for\\
Multi-Type Branching Processes\\ 
and Stochastic Context-Free Grammars}
\author{Kousha Etessami\\U. of Edinburgh\\{\tt kousha@inf.ed.ac.uk}
\and 
Alistair Stewart\\U. of Edinburgh\\{\tt stewart.al@gmail.com} 
\and
Mihalis Yannakakis\\Columbia U.\\{\tt mihalis@cs.columbia.edu}}

\date{}

\setcounter{page}{0}

\maketitle

\thispagestyle{empty}

\begin{abstract}
  We show that one can approximate 
 the least fixed point solution
for a multivariate system of monotone probabilistic polynomial equations
in time 
polynomial in both the encoding size of the system of
equations and in 
$\log(1/\epsilon)$, where $\epsilon > 0$ is the desired additive error 
bound of the solution.  (The model of 
computation is the standard Turing machine model.)

We use this result to resolve several open problems regarding the computational
complexity of computing key quantities associated with some
classic and heavily studied stochastic processes, including multi-type branching
processes and stochastic context-free grammars.

\end{abstract}

\section{Introduction}

  Some of the central computational problems associated
with a number of classic stochastic processes can be rephrased
as a problem of computing 
the non-negative {\em least fixed point}  solution of an
associated multivariate system of monotone polynomial equations.

   In particular, this is the case for computing the 
{\em extinction probabilities}
(also called {\em final probabilities}) for {\em 
multi-type branching processes} (BPs), 
a problem which was first studied in the 1940s by Kolmogorov
and Sevastyanov \cite{KolSev47}.  Branching processes are a 
basic stochastic model in probability theory, with
applications in diverse areas ranging from
population biology to the physics of nuclear chain reactions 
(see \cite{Harris63} for the classic theoretical text on BPs, and
 \cite{KA02,HJV05,PazPal08} for some of the
more recent applied textbooks on  BPs). 
BPs describe the stochastic
evolution of a population of objects of distinct types. 
In each generation, every object $a$ of each type $T$ 
gives rise 
to a (possibly empty) multiset of objects of distinct types in
the next generation,
according to a given probability distribution on such multisets
associated with the type $T$.
The extinction probability, $q_T$, associated with type $T$ is
the probability that, starting with exactly one object of type $T$,
the population will eventually become extinct.
Computing these probabilities is fundamental to
many other kinds of analyses for BPs (see, e.g., \cite{Harris63}).
Such probabilities are in general irrational, even when
the finite data describing the BP (namely, the probability distributions associated 
with each of the finitely many types $T$) are
given by rational values (as is assumed usually for computational purposes).
Thus, we would like to compute the probabilities approximately to desired precision.

Another essentially equivalent
problem is that of computing the 
probability of the language generated by a {\em stochastic 
context-free grammar} (SCFG), and more generally its
{\em termination probabilities}
(also called the {\em partition function}).
SCFGs are a fundamental model
in statistical natural language processing
and in biological sequence analysis (see, e.g., \cite{ManSch99,DEKM99,NedSat08}).  
A SCFG provides a probabilistic model for the generation of strings
in a language, by associating probabilities to the rules of a CFG.
The termination probability of a nonterminal $A$ is the probability that a random derivation of the
SCFG starting from $A$ eventually terminates and generates a finite-string; 
the total probability of the language of a SCFG is simply
the termination probability for the start symbol of the SCFG.
Computing these termination probabilities 
is again a key computational problem for the analysis of SCFGs,
and is required for computing other quantities, for example
the probability of generating a given string.

Despite decades of applied work on BPs and SCFGs, as well as theoretical
work on their computational problems, no polynomial
time algorithm was known for computing extinction probabilities
for BPs, nor for termination probabilities for SCFGs,
nor even for approximating them within any nontrivial constant:
prior to this work it was not even known whether one can distinguish in P-time
the case where the probability is close to 0 from the case where it is 
close to 1.

We now describe the kinds of nonlinear equations that have to be solved
in order to compute the above mentioned probabilities.
Consider systems of multi-variate polynomial fixed point equations
in $n$ variables, with $n$ equations, of the form
$x_i = P_i(x), \ i= 1,\ldots,n$
where $x = (x_1, \ldots,x_n)$ denotes the vector of variables,
and $P_i(x)$ is a multivariate polynomial in the variables $x$. 
We denote the entire system of equations by $x = P(x)$.
The system is {\em monotone} if all the coefficients of the polynomials are
nonnegative. It is a {\em probabilistic polynomial system} (PPS) if in addition the
coefficients of each polynomial sum to at most 1.

It is easy to see that a system of 
probabilistic polynomials $P(x)$
always maps any vector in $[0,1]^n$ to 
another vector in $[0,1]^n$.
It thus follows, by  Brouwer's fixed point 
theorem, that a PPS
$x = P(x)$ always has a solution in $[0,1]^n$.
In fact, it always has a unique {\em least} solution, 
$q^* \in [0,1]^n$, which is 
coordinate-wise smaller than any other non-negative
solution, and which is the {\em least fixed point} (LFP) of the monotone
operator $P: [0,1]^n \rightarrow [0,1]^n$ on 
$[0,1]^n$. The existence of the LFP, $q^*$,
is guaranteed by Tarski's fixed point theorem.
From a BP or a SCFG we can construct easily a
probabilistic polynomial system $x=P(x)$ whose
LFP $q^*$ yields precisely the
vector of extinction probabilities for the BP, 
or termination probabilities for the SCFG.
Indeed, the converse also holds: computing the
extinction probabilities for a BP (termination probabilities of an SCFG) 
and computing the LFP of  a system of
probabilistic polynomial equations are equivalent problems.
As we discuss below, some other stochastic models also lead to equivalent
problems or to special cases.

\medskip

{\bf Previous Work.}
As already stated, the polynomial-time computability of
these basic probabilities
for multi-type branching processes and SCFGs have been longstanding open problems.
In \cite{rmc}, we studied these problems as special sub-cases of a 
more general class of stochastic processes called {\em recursive Markov
chains} (RMCs),  which form a natural model for analysis of 
probabilistic procedural
programs with recursion,
and we showed that these problems are equivalent 
to computing termination probabilities for the special subclass of {\em 1-exit} RMCs (1-RMC).
General RMCs are expressively equivalent to the model of
{\em probabilistic pushdown systems} studied in \cite{EKM,BEK}.
We showed that for BPs, SCFGs, and 1-RMCs, the {\em qualitative} problem of determining which probabilities
are exactly 1 (or 0) can be solved in P-time, by exploiting basic
results from the theory of branching processes.
We proved however that the {\em decision} problem of determining whether the
extinction probability of a BP (or termination probability of a SCFG or a 1-RMC)  is $\geq 1/2$
is at least as hard as some longstanding
open problems in the complexity of numerical computation, namely, 
the {\em square-root sum problem}, and a much more general decision
problem (called PosSLP)  which captures the power
of unit-cost exact rational arithmetic \cite{ABKM06}, and hence it  is
very unlikely that the decision problem can be solved in P-time.
For general RMCs we showed that in fact this hardness holds for computing
{\em any} nontrivial {\em approximation} of the termination probabilities.
No such lower bound was shown for the approximation problem
for the subclass of 1-RMCs (and BPs and SCFGs). In terms of upper bounds, the best
we knew so far, even for any nontrivial approximation, was 
that the problem is in FIXP (which is in PSPACE), 
i.e., it can be reduced to approximating a Nash equilibrium of a 3-player game \cite{fixp}.
We improve drastically on this in this paper, resolving the problem completely,
by showing we can compute these probabilities in
P-time  to any desired accuracy.

An equivalent way to formulate the problem of computing
the LFP, $q^*$, of a PPS, $x=P(x)$, is as a mathematical optimization problem:
{\em minimize:}  $\sum^n_{i=1} x_i$;  {\em subject to:} 
\{$P(x)-x \leq 0 ; ~x\geq 0 $\}.
This program has a unique optimal solution, which is the LFP $q^*$.
If the constraints were convex, the solution could be computed approximately
using convex optimization methods.
In general,  the PPS constraints are {\em not} convex (e.g., 
$x_2 x_3 - x_1 \leq 0$ is not a convex constraint), however for certain 
restricted subclasses of PPSs they are.
This is so for {\em backbutton processes} which
were introduced and studied by Fagin et. al. in \cite{FaginKKRRRST00}
and used there
to analyze random walks on the web.
Backbutton processes constitute a 
restricted subclass of SCFGs (see \cite{rmc}).  Fagin et. al. applied
semidefinite programming to approximate
the corresponding termination probabilities for backbutton processes, and
used this as a basis for approximating other important quantities associated with them.

There are a number of natural iterative methods that one
can try to use  (and which indeed are used in practice) 
in order to solve the equations arising from BPs and SCFGs.
The simplest such method is {\em value iteration}:  starting
with the vector $x^0 = 0$,
iteratively compute the sequence $x^{i+1} := P(x^{i})$, $i = 1,\ldots$.   
The sequence always converges monotonically to the LFP $q^*$.
Unfortunately, it can be very slow to converge:  even for
the simple univariate polynomial system $x = (1/2) x^2 + 1/2$,
for which $q^* =1$, one requires $2^{i-3}$ iterations to 
exceed  $1- 1/2^{i-1}$, i.e. to get $i$ bits of precision \cite{rmc}.

In \cite{rmc} we provided a much better 
method that always converges monotonically to $q^*$.
Namely, we showed that a decomposed variant of
Newton's method can be applied to  
such systems of equations (and in fact, much more generally,
to any monotone system of polynomial equations)  $x=P(x)$, 
and always converges monotonically to the LFP solution $q^*$
(if a solution exists).
Optimized variants of this decomposed Newton's method have
by now been implemented in several tools \cite{WojEte07,NedSat08},
and they perform quite well in practice on many instances.

The theoretical speed of convergence of Newton's method on such monotone 
(probabilistic) polynomial systems was subsequently studied in much greater detail by Esparza, Kiefer,
and Luttenberger in \cite{lfppoly}.  They showed that, even for
Newton's method on PPSs,
there are instances where exponentially
many iterations of Newton's method (even with {\em exact arithmetic} in each
iteration) are required, as a function of the
encoding size of the system, in order to converge
to within just one bit of precision of the solution $q^*$.
On the upper bound side, they showed that after 
{\em some} number
of iterations in an initial phase, 
thereafter Newton obtains an additional bit of precision per iteration
(this is called {\em linear convergence}
in the numerical analysis literature).
In the special case where 
the input system of equations is {\em strongly connected},
meaning roughly that all variables depend (directly or indirectly)
on each other in the system of equations $x=P(x)$,
they proved an exponential upper bound on the number
of iterations required in the initial phase as a function
of input size. 
For the general case
where the input system of equations is not strongly connected,
they did not provide any upper bound as a function of the input size.
In more recent work, Esparza et al \cite{EGK10}
further studied probabilistic polynomial systems.
They did not provide any new worst-case upper bounds on the behavior
of Newton's method in this case, but they
studied a modified method which 
is in practice more robust numerically,
and they also showed that the qualitative problem of determining
 whether the LFP $q^* = \mathbf{1}$ is decidable
in {\em strongly} polynomial time.

\medskip
{\bf Our Results.}
In this paper we provide the first polynomial time 
algorithm for computing, to any desired accuracy, the least fixed point 
solution, $q^*$, of probabilistic
polynomial  systems, and thus also provide the first 
P-time approximation
algorithm for extinction probabilities
of BPs, and termination probabilities of SCFGs and 1-exit RMCs.
The algorithm proceeds roughly as follows:

1. We begin with a preprocessing step, in which we determine all
variables $x_i$ which have value 0 or 1 in the LFP $q^*$ and remove
them from the system.

2. On the remaining system of equations, $x=P(x)$, with
an LFP $q^*$ such that $\mathbf{0} < q^* < \mathbf{1}$,
we apply 
Newton's method, starting at initial vector $x^{(0)} := {\mathbf 0}$.
Our key result is to  show that, once variables $x_i$ with $q^*_i \in \{0,1\}$ 
have been removed,
Newton's method only requires polynomially many iterations 
(in fact, only {\em linearly} many iterations) 
as a function of both the encoding 
size of the equation system and of $\log(1/\epsilon)$ to converge
to within additive error $\epsilon > 0$ of the vector $q^*$.
To do this, we build on the previous works \cite{rmc,lfppoly,fixp}, and
extend them with new techniques.

3. The result in the previous step applies to the
unit-cost arithmetic RAM model of computation, where we assume that each
iteration of Newton's method is carried out in {\em exact} arithmetic.
The problem with this, of course, is that in general 
after only a linear
number of iterations, the number of bits required to represent the
rational numbers in Newton's method can be exponential in the 
input's encoding size.
We resolve this by showing, via a careful round-off analysis,
that if after each iteration
of Newton's method the positive rational numbers in question are all {\em rounded down}
to 
a suitably long but polynomial encoding length
(as a function of both the input size and of the 
desired error $\epsilon > 0$), then the resulting ``approximate''
Newton iterations will still be well-defined and will still
converge to $q^*$, within
the desired error $\epsilon > 0$, in polynomially (in fact {\em linearly})
many iterations. The correctness of the rounding 
relies crucially on the properties of PPSs shown in step 2, and it does not
work in general for other types of equation systems.\footnote{In particular, 
there are examples of PPSs which {\em do} have $q^*_i=1$ 
for some $i$,
such that this rounding method fails completely because 
of very severe {\em ill-conditioning} (see \cite{lfppoly}). Also, for
{\em quasi-birth-death} (QBDs) processes, a stochastic model studied heavily
in qeueing systems, different
monotone polynomial equations can be associated with key
probabilities,
and Newton's method converges in polynomially many iterations
\cite{EWY10},
but this rounding fails, and in fact it is  an open problem whether 
the key probabilities for QBDs can 
be computed in P-time in the {\em Turing} model.}

Extinction probabilities of BPs and termination probabilities of SCFGs 
are basic quantities 
that are important in many other analyses of these processes. 
We illustrate  an application
of these results by solving in polynomial time 
some other important problems for SCFGs
that are at least as hard as the termination problem. We show two results in this regard:

(1) Given a SCFG and a string, we can compute the probability of the string 
to any desired accuracy in polynomial time. This algorithm uses the following construction:

(2) Given an SCFG, we can compute in P-time another SCFG in Chomsky normal form (CNF)
that is approximately equivalent in a well-defined sense.

These are the first P-time algorithms for these problems that work
for {\em arbitrary} SCFGs, including grammars that contain $\epsilon$-rules.
Many tasks for SCFGs, including computation of string probabilities, 
become much easier for SCFGs in CNF, 
and in fact, many papers start by assuming that the given SCFG is in CNF.
In the nonstochastic case, there are standard efficient 
algorithms for
transforming any  CFG to an equivalent one in CNF. 
However, for stochastic grammars
this is not the case and things are much more complicated.
It is known that every SCFG has an equivalent one in CNF \cite{AMP99},
however, as remarked in \cite{AMP99}, their proof is nonconstructive and
does not yield any algorithm (not even an exponential-time algorithm).
Furthermore, it is possible that even though the given SCFG has rational
rule probabilities, the probabilities for any equivalent SCFG in CNF must be irrational,
hence they have to be computed approximately.
We provide here an efficient P-time algorithm for computing such a CNF SCFG.
To do this requires our P-time algorithm for SCFG
termination probabilities, and requires the development of 
considerable additional machinery to handle the elimination of 
probabilistic $\epsilon$-rules
and unary rules, while keeping numerical values polynomially
bounded in size, yet ensuring that the final SCFG meets the desired 
accuracy bounds.

The paper is organized as follows: Section 2 gives basic definitions
and background; Section 3 addresses the solution of PPSs, showing a linear bound
on the number of Newton iterations; Section 4 shows a polynomial time bound in the
Turing model; Section 5 describes briefly the applications to SCFGs.
Due to space, most proofs and technical development are
in the Appendix.

\section{Definitions and Background}

A (finite) {\bf\em multi-type Branching Process} ({\bf BP}),  $G=(V,R)$, 
consists of a (finite) set  $V = \{S_1,\ldots,S_n\}$ of {\em types}, 
and a (finite) set $R = \cup^n_{i=1} R_i$  of {\em rules},
which are
partitioned into distinct rule sets, $R_i$, associated
with each type $S_i$.   Each rule $r \in R_i$ has the form
 $S_i \stackrel{p_r}{\rightarrow} \alpha_r$, 
where $p_r \in (0,1]$,
 and $\alpha_r$ is a finite multiset (possibly the empty multiset) 
whose elements are in $V$.
Furthermore, for every 
 type $S_i$, we have $\sum_{r \in R_i} p_r = 1$.  
 The rule $S_i \stackrel{p_r}{\rightarrow} \alpha_r$ specifies the
 probability with which an entity (or object) of type $S_i$
 generates the multiset $\alpha_r$ of offsprings 
 in the next generation. 
As usual, rule probabilities $p_r$ 
are assumed to be
 rational for computational purposes. 
Multisets $\alpha_r$ over $V$ 
can be encoded by giving a vector $v(\alpha_r) \in \nat^n$, 
with the $i$'th coordinate $v(\alpha_r)_i$
representing the number of elements of type $S_i$ in the multiset $\alpha_r$.
We assume instead that the multisets  
$\alpha_r$ are represented even more succinctly in {\em sparse representation},
by specifying only the non-zero coordinates of the vector $v(\alpha_r)$, encoded in binary.

A BP, $G=(V,R)$, 
defines a discrete-time stochastic (Markov) process, whose states are
multisets 
over $V$, or equivalently elements of $\nat^n$. 
If the state at time $t$ is $\alpha^{t}$, then the next state $\alpha^{t+1}$ 
at time $t+1$
is determined by 
{\em independently} choosing, for each object of each type $S_i$ in  
the multiset $\alpha^t$,
a random rule $r \in R_i$ of the form 
$S_i \stackrel{p_r}{\rightarrow} \alpha_r$,
according to the probability $p_r$ of that rule, yielding
the multiset $\alpha_r$ as the ``offsprings'' of that object in one 
generation.  
The multiset $\alpha^{t+1}$ is then given by the
{\em multiset union} of all such offspring multisets, 
randomly and independently chosen for each object in the multiset $\alpha^t$.
A trajectory ({\em sample path}) of this stochastic process,
starting at time $0$ in initial multiset $\alpha^0$, 
is a sequence
$\alpha^0, \alpha^1, \alpha^2, \ldots$ of multisets over $V$.
Note that if ever the process reaches {\em extinction}, i.e., 
if ever $\alpha^t = \{ \}$ at some time $t \geq 0$,  then 
$\alpha^{t'} = \{ \}$ for all times $t' \geq t$.

Very fundamental quantities associated with a BP,
which are a key to many analyses of BPs,  are its 
vector of {\bf\em extinction
probabilities},  $q^* \in [0,1]^n$,  where $q^*_i$ is defined
as the probability that, starting with initial multiset 
$\alpha^0 := \{ S_i \}$ at time $0$, i.e., starting with a single object of type $S_i$,
the stochastic process eventually reaches extinction, i.e., 
that $\alpha^{t} = \{\}$ at some time $t > 0$. 

Given a BP, $G=(V,R)$,
there is a system of polynomial equations  in $n = |V|$ variables, $x=P(x)$,
that we can associate with $G$, such that the {\em least} non-negative
solution vector for $x = P(x)$  is the vector of extinction
probabilities $q^*$ (see, e.g., \cite{Harris63, rmc}).
Let us define these equations.
For an $n$-vector of variables $x = (x_1,\ldots,x_n)$, and 
a vector $v \in \nat^n$,  we use the shorthand $x^v$ to denote
the monomial $x_1^{v_1}\ldots x^{v_n}_{n}$.
Given BP  $G = (V,R)$, we define equation $x_i = P_i(x)$ by:
$x_i =   \sum_{r \in R_i} p_r x^{v(\alpha_r)}$.
This yields $n$ polynomial equations 
in $n$ variables, which we denote by $x = P(x)$.
It is not hard to establish that $q^* = P(q^*)$.  In fact,
$q^*$  is the {\em least} non-negative solution of $x = P(x)$.
In other words, if $q' = P(q')$ for $q' \in \real^n_{\geq 0}$, then
$q' \geq q^* \geq 0$, i.e., $q'_i \geq q^*_i$ for all $i= 1,\ldots,n$.

Note that this system of polynomial equations $x = P(x)$ has
very special properties.
Namely, 
(I): the coefficients and constant of each polynomial
$P_i(x) = \sum_{r \in R_i} p_r x^{v(\alpha_r)}$
are nonnegative, i.e., $p_r \geq 0$ for all $r$.
Furthermore, (II): the coefficients sum to 1, i.e., 
$\sum_{r \in R_i} p_r = 1$.
We call $x= P(x)$ 
a {\bf\em probabilistic polynomial system} of equations
({\bf PPS}) if it has properties (I) and (II) except that for convenience 
we weaken (II) and also allow
(II$'$): $\sum_{r \in R_i} p_r \leq 1$.   
If a system of equations $x = P(x)$ only satisfies (I), then
we call it a {\bf\em monotone polynomial system} of equations
({\bf MPS}).  

For any PPS, $x=P(x)$,  
$P(x)$ defines a 
{\em monotone}  operator $P: [0,1]^n \rightarrow [0,1]^n$,
i.e., if $y \geq x \geq {\textbf 0}$ then
$P(y) \geq P(x)$. 
For any BP with corresponding PPS $x=P(x)$, $q^*$ is precisely the 
{\em least fixed point} ({\bf LFP}) of the monotone operator 
$P: [0,1]^n \rightarrow [0,1]^n$ (see \cite{rmc}).
A MPS, $x=P(x)$, also defines a  monotone operator 
$P: \real^n_{\geq 0} \rightarrow 
\real^n_{\geq 0}$ on the non-negative 
orthant $\real^n_{\geq 0}$.   An MPS need not in general have any solution in 
$\real^n_{\geq 0}$, but when it does so, it has a {\em least fixed point}
solution $q^* = P(q^*)$ such that $0 \leq q' = P(q')$ implies $q^* \leq q'$.

Note that {\em any}  PPS (with rational coefficients) 
can be obtained as the system of
equations $x=P(x)$ for a corresponding BP $G$ (with
rational rule probabilities), and vice versa.\footnote{``{\em Leaky}'' 
PPSs where
$\sum_{r \in R_i} p_r < 1$ can be translated easily to BPs by adding an
extra dummy type $S_{n+1}$, with rule $S_{n+1} 
\stackrel{1}{\rightarrow} \{ S_{n+1}, S_{n+1} \}$, so $q^*_{n+1} =0$,
and adding to $R_i$, for each ``leaky'' $i$,
the rule $S_i \stackrel{p'_i}{\rightarrow} \{ S_{n+1}, S_{n+1} \}$
with probability $p'_i := (1-\sum_{r \in R_i} p_r)$.
The probabilities $q^*$ 
for the BP (ignoring $q^*_{n+1} = 0$) give the LFP of the PPS .}
Thus, the computational problem of computing the extinction
probabilities of a given BP is equivalent to the problem
of computing the least fixed point (LFP) solution $q^*$ of a
given PPS, $x=P(x)$.
For a PPS $x=P(x)$, we shall use $|P|$ to
denote the sum of the number $n$ of variables and
the numbers of bits of all the nonzero coefficients and
nonzero exponents of all the polynomials in the PPS.
Note that the encoding length of a PPS in sparse representation 
is at least $|P|$ (and at most $O(|P| \log n)$).

The probabilities $q^*$ can in general be irrational, and
even deciding whether $q^*_i \geq 1/2$ is as hard as long
standing open problems, including the {\em square-root sum}
problem, which are not even known to be in NP (see \cite{rmc}).
We instead want to approximate $q^*$ within a desired 
additive error $\epsilon > 0$.
In other words, we want to compute a 
rational vector $v' \in \rat^n \cap [0,1]^n$
such that $\| q^* - v' \|_\infty < \epsilon$.

A PPS, $x = P(x)$, is said to be in
{\em Simple Normal Form} (SNF)
if for every $i = 1,\ldots,n$,  the polynomial $P_i(x)$ has one of two forms:
(1) \underline{Form$_{*}$}: $P_i(x) \equiv x_j x_k$ 
is simply a quadratic monomial;  \ or  (2) \underline{Form$_{+}$}: $P_i(x)$ is a {\em linear} 
expression 
$\sum_{j \in \mathcal{C}_i} p_{i,j} x_j + p_{i,0}$,  for some rational 
non-negative
coefficients $p_{i,j}$ and $p_{i,0}$, and
some index set $\mathcal{C}_i \subseteq
\{1,\ldots,n\}$,
where $\sum_{j \in \mathcal{C}_i \cup \{0\}} p_{i,j} \leq 1$.
We call such a linear equation {\em leaky}
if  $\sum_{j \in \mathcal{C}_i \cup \{0\}} p_{i,j} < 1$.
An MPS is said to be in SNF if the same conditions hold except
we do not require $\sum_{j \in \mathcal{C}_i \cup \{0\}} p_{i,j} \leq 1$.
The following is proved in the appendix.

\begin{proposition}[cf. Proposition 7.3 \cite{rmc}]
\label{prop:snf-form}
Every PPS (MPS), $x = P(x)$, can be transformed in P-time
to an ``{\em equivalent}''  PPS (MPS, respectively),  $y=Q(y)$  in SNF form,
such that 
$|Q| \in O( |P|  )$.
More precisely, the variables $x$ are 
a subset of the variables $y$, and  $y=Q(y)$ has LFP
$p^* \in \real^m_{\geq 0}$ iff $x=P(x)$ has LFP $q^* \in \real^n_{\geq 0}$,
and projecting $p^*$ onto the $x$ variables yields 
$q^*$.
\end{proposition}

\begin{proposition}[\cite{rmc}]
\label{prob1-ptime-scfg-prop}
There is a P-time algorithm that, given a PPS, $x=P(x)$, over $n$ variables,
with LFP $q^* \in \real^n_{\geq 0}$,
determines for every $i = 1,\ldots,n$ whether $q^*_i = 0$ or $q^*_i = 1$
or $0< q^*_i <1$.
\end{proposition}

\noindent Thus, for every PPS, we can detect in P-time 
all the variables $x_j$ such that $q^*_j = 0$
or $q^*_j = 1$.   We can then remove these variables
and  their corresponding equation $x_j = P_j(x)$,
and substitute their values on the right hand sides (RHS)
of the remaining equations.
This yields a new PPS, $x' = P'(x')$, where its LFP solution, 
$q'^*$, is ${\textbf 0} < q'^* < {\textbf 1}$, which corresponds to 
the remaining coordinates of $q^*$.

{\bf\em We can thus henceforth assume, w.l.o.g., that 
any given PPS, $x= P(x)$, 
is in SNF form and has an LFP 
solution $q^*$ such that $\textbf{0} < q^* < \textbf{1}$.}

For a MPS or PPS, $x=P(x)$,  its variable
{\em dependency graph} is defined to be the digraph $H = (V,E)$, with 
vertices $V=\{x_1,\ldots,x_n\}$, such that 
$(x_i,x_j) \in E$ iff in $P_i(x) = \sum_{r \in R_i} p_r x^{v(\alpha_r)}$ 
there is a coefficient $p_r > 0$ such that $v(\alpha_r)_j > 0$.
Intuitively, $(x_i, x_j) \in E$ means that $x_i$ ``depends directly''
on $x_j$.
A MPS or PPS, $x=P(x)$, is called {\bf\em strongly connected} if its 
dependency graph $H$
is strongly connected.
As in \cite{rmc}, for  analysing PPSs we will find it very useful to 
decompose the PPS based on the {\em strongly connected components} 
(SCCs) of its variable dependency graph.

\section{Polynomial upper bounds for Newton's method on PPSs}
\label{sec:poly-bound-newton}

To find a solution for a differentiable system of equations 
$F(x) = {\textbf 0}$, in $n$ variables, {\em Newton's method} uses the
following iteration scheme: start with some initial vector 
$x^{(0)} \in \real^n$, and 
for $k > 0$ let:\\
$x^{(k+1)} := x^{(k)} - F'(x^{(k)})^{-1}(F(x^{(k)}))$,
where $F'(x)$ is the  Jacobian matrix of $F(x)$.

Let $x=P(x)$ be a given PPS (or MPS) in $n$ variables.
Let $B(x) := P'(x)$ denote the Jacobian matrix of $P(x)$.
In other words, $B(x)$ is an $n \times n$ matrix such that 
$B(x)_{i,j} = \frac{\partial P_i(x)}{\partial x_j}$.
Using Newton iteration, starting at $n$-vector $x^{(0)} :=  {\bf 0}$,
yields the following iteration:
\begin{equation}\label{newton-eq}
x^{(k+1)} := x^{(k)} + (I-B(x^{(k)}))^{-1} (P(x^{(k)}) - x^{(k)}))
\end{equation}

\noindent For a vector $z \in \real^n$, 
assuming that matrix $(I - B(z))$ is non-singular, 
we define a single iteration of Newton's method for $x=P(x)$ on $z$ 
via the following operator:
\begin{equation}\label{newton-one-it-eq}
\mathcal{N}_P(z) :=     z + (I-B(z))^{-1} (P(z) - z)
\end{equation}

\noindent It was shown in \cite{rmc} that 
for any MPS, $x= P(x)$, with LFP $q^* \in \real^N_{\geq 0}$,
if we first find and remove the variables that have value $0$ in the 
LFP, $q^*$,
and apply a decomposed variant of Newton's method
that decomposes the system according to the strongly connected 
components (SCCs)
of the dependency graph and processes them bottom-up,
then the values converge
{\em monotonically} to $q^*$.
PPSs are a special case of MPSs, so the same applies to PPSs.
In \cite{lfppoly}, it was pointed out that if $q^* > 0$, i.e., after we remove
the variables $x_i$ where $q^*_i = 0$, decomposition into SCCs 
isn't strictly necessary. (Decomposition is nevertheless very useful 
in practice, as well as in the theoretical analysis,
including in this paper.).
Thus:

\begin{proposition}[cf.  Theorem 6.1 of \cite{rmc} and Theorem 4.1 of \cite{lfppoly}]
\label{prop:monotone-conv-newt}
Let $x=P(x)$ be a MPS, with LFP 
$q^* > {\textbf 0}$.
Then starting at $x^{(0)} := \textbf{0}$, the Newton iterations $x^{(k+1)} := 
\mathcal{N}_P(x^{(k)})$ are  well
defined and monotonically converge to $q^*$, i.e.
$\lim_{k \rightarrow \infty}  x^{(k)} = q^*$,
and $x^{(k+1)} \geq x^{(k)} \geq \textbf{0}$ for all $k \geq 0$.
\end{proposition}

We will actually establish an
extension of this result in this paper,
because in Section \ref{sec:p-time} we will need to show 
that even when each iterate is 
suitably {\em rounded off}, the rounded Newton iterations
are all well-defined and converge to $q^*$.
The main goal of this section 
is to show that for PPSs, $x=P(x)$, with LFP $0 < q^* < 1$, 
polynomially many iterations of Newton's method,
{\em using exact rational arithmetic}, suffice,
as a function of $|P|$ and  $j$,
to compute $q^*$ to within additive error $1/2^j$.
In fact, we show a much stronger
{\em linear} upper bound with small explicit constants:

\begin{thm}[{\bf Main Theorem of Section \ref{sec:poly-bound-newton}}] 
\label{linearconvgen} 
Let $x=P(x)$ be any PPS in SNF form, with
LFP $q^*$, such that $\mathbf{0} < q^* < \mathbf{1}$.
If we start Newton iteration at $x^{(0)} := \textbf{0}$,
with $x^{(k+1)} :=  \mathcal{N}_P(x^{(k)})$, 
then for any integer $j \geq 0$ the following inequality holds:
$\| q^* - x^{(j + 4|P|)}\|_{\infty} \leq 2^{-j}$.
\end{thm}

\noindent We need a sequence of Lemmas.
The next two Lemmas are proved in the appendix.

\begin{lem} \label{int} Let $x=P(x)$ be a MPS, with $n$ variables, in SNF form,
and let $a, b \in \real^n$. Then:
$$P(a)-P(b) = B(\frac{a + b}{2})(a-b)= \frac{B(a) + B(b)}{2}(a-b)$$ \end{lem}

\begin{lem} \label{newton} 
Let $x=P(x)$ be a MPS in SNF form.
Let $z \in \real^n$ be any vector  
such that $(I - B(z))$ is non-singular, and thus 
$\mathcal{N}_P(z)$ is defined. 
Then:
$$q^* -  \mathcal{N}_P(z) = (I-B(z))^{-1}\frac{B(q^*) - B(z)}{2}(q^* - z)$$
\end{lem}

To prove their exponential upper bounds for {\em strongly connected} PPSs, 
\cite{lfppoly} used the notion of a {\em cone vector}
for the matrix $B(q^*)$, 
that is a vector $d > 0$ such that $B(q^*)d \leq d$.
For a strongly connected MPS, $x=P(x)$, with $q^* > 0$,
the matrix $B(q^*) \geq 0$ is irreducible, and thus 
has a positive eigenvector. 
They used this eigenvector as their 
cone vector $d > 0$. 
However, such an eigenvector 
yields only weak (exponential) bounds.
Instead, we show there is  a different cone vector for
$B(q^*)$, and even for $B(\frac{1}{2}(\textbf{1} + q^*))$, that works
for arbitrary (not necessarily strongly-connected) PPSs:

\begin{lem} \label{1-qmean} If $x= P(x)$ is a PPS in 
$n$ variables, in SNF form,
with LFP ${\textbf 0} < q^* < {\textbf 1}$, and where
$P(x)$ has
Jacobian $B(x)$, then $\forall z \in \real^n$ 
such that $\textbf{0} \leq z \leq \frac{1}{2}(\textbf{1} + q^*)$:  \ \ $B(z)(\textbf{1}-q^*) \leq (\textbf{1}-q^*)$.\\  
In particular, $B(\frac{1}{2}(\textbf{1} + q^*))(\textbf{1}-q^*) \leq (\textbf{1}-q^*)$, and 
$B(q^*)(\textbf{1}-q^*) \leq (\textbf{1}-q^*)$.
\end{lem}

 \begin{proof} Lemma \ref{int} applied to $\textbf{1}$ and $q^*$ gives:
$P(\textbf{1}) - P(q^*) = 
P(\textbf{1}) - q^* = B(\frac{1}{2}(\textbf{1} + q^*))(\textbf{1}-q^*)$.
But note that $P(\textbf{1}) \leq \textbf{1}$,
because for any PPS, since
the nonnegative coefficients of each polynomial $P_i(x)$ sum to $\leq 1$,
$P(x)$ maps $[0,1]^n$ to $[0,1]^n$.
Thus  $\textbf{1} - q^* \geq 
P(\textbf{1}) - q^* = 
B(\frac{1}{2}(\textbf{1} + q^*))(\textbf{1}-q^*)$.
Now observe that for $0 \leq z \leq \frac{1}{2}(\textbf{1} + q^*)$,
$B(\frac{1}{2}(\textbf{1} + q^*)) \geq B(z) \geq 0$, because
the entries of Jacobian $B(x)$ have nonnegative coefficients.
Thus since $(\textbf{1} - q^*) \geq 0$, 
we have $(\textbf{1}-q^*) \geq B(z)(\textbf{1}-q^*)$.\end{proof}

\noindent For a square matrix $A$, let $\rho(A)$ denote the spectral radius of $A$.

\begin{thm}\label{thm:spec-full} For any PPS, $x=P(x)$,
in SNF form, if we have $0 < q^* < 1$, 
then for all $0 \leq z \leq q^*$,
$\rho(B(z)) < 1$ and $(I-B(z))^{-1}$ exists and is nonnegative.
\end{thm}

The proof (in the appendix) uses, among other things, Lemma \ref{1-qmean}.
Note that this theorem tells us, in particular,
that for {\em every} $z$ (including $q^*$), such that ${\textbf 0} \leq z \leq q^*$, 
the Newton iteration
${\mathcal N}_P(z)$ is well-defined.
This will be important in Section \ref{sec:p-time}.
We need the following Lemma from \cite{lfppoly}.
(To be self-contained, and to clarify our assumptions, we provide a short
proof in the appendix.)

\begin{lem}[Lemma 5.4 from \cite{lfppoly}] 
\label{cone} 
Let $x=P(x)$ be a MPS, with polynomials of 
degree bounded by 2, with LFP, $q^* \geq 0$.
Let $B(x)$ denote the Jacobian matrix of $P(x)$.
For any positive vector 
$\textbf{d} \in \mathbb{R}^n_{> 0}$ 
that satisfies $B(q^*) \textbf{d} \leq \textbf{d}$, any positive real value $\lambda > 0$, 
and any nonnegative vector $z \in \real^n_{\geq 0}$, 
if $q^* - z \leq \lambda \textbf{d}$, and $(I-B(z))^{-1}$
exists and is nonnegative, then 
$q^* - {\mathcal N}_P(z) \leq \frac{\lambda}{2} \textbf{d}$.\end{lem}

\noindent For a vector $b \in \real^n$, we shall use the
following notation:  $b_{\text{min}} = \min_i b_i$, and $b_{\text{max}} = \max_i b_i$.

\begin{cor} 
\label{num-iteration-on-cone-lem}
Let $x=P(x)$ be MPS,  
with LFP $q^* > 0$, 
and let $B(x)$ be the Jacobian
matrix for $P(x)$. 
Suppose there is a vector $d \in \mathbb{R}^n$, 
$\textbf{0} < d \leq \textbf{1}$, such that $B(q^*)d \leq d$.
For any positive integer $j >0$, 
if we perform Newton's method starting at $x^{(0)} := \textbf{0}$, then: 
$ \| q^* - x^{(j - \lfloor \log_2 d_{\text{min}}  \rfloor)} \|_{\infty} \leq 2^{-j}$.
\end{cor}
 \begin{proof} 
By induction on $k$, we show
$q^* - x^{(k)} \leq  2^{-k}\frac{1}{d_\text{min}}d$.
For the base case, $k=0$, since $d > 0$,  
$\frac{1}{d_\text{min}}d \geq \mathbf{1} \geq q^* = q^* - x^{(0)}$.
For $k > 0$, 
apply Lemma \ref{cone}, setting
$z := x^{(k-1)}$, 
$\lambda := \frac{1}{d_\text{min}} 2^{-(k-1)}$
and ${\mathbf d} := d$.   This yields $q^* - x^{(k)} \leq  
\frac{\lambda}{2} {\mathbf d} = 
2^{-k}\frac{1}{d_\text{min}}d$.
Since we assume $\|d\|_\infty \leq 1$,
we have
$\| 2^{-(j - \lfloor \log_2 d_\text{min} \rfloor)}\frac{1}{d_\text{min}}d \|_{\infty}  \leq 2^{-j} $, and thus 
$\| q^* - x^{(j - \lfloor \log_2 d_{\text{min}}  \rfloor)} \|_{\infty} \leq 2^{-j}$.
\end{proof}

\begin{lem} \label{linearconv} For a PPS in SNF form,
with LFP $q^*$, where ${\textbf 0} < q^* < {\textbf 1}$, if 
we start Newton iteration at $x^{(0)} := \textbf{0}$, then:

\vspace*{-0.28in}
 
$$ \| q^* - x^{(j + {\Large \lceil} (\log_2  
\frac{({\textbf 1}-q^*)_\text{max}}{({\textbf 1}-q^*)_\text{min}}) {\Large \rceil} )}  \|_{\infty} \leq 2^{-j}$$
\end{lem}

\vspace*{-0.1in}

\begin{proof}
For $d := \frac{\textbf{1} - q^*}{\|\textbf{1} - q^*\|_\infty}$,  
$d_{min} = \frac{(\textbf{1}-q^*)_\text{min}}{(\textbf{1}-q^*)_\text{max}}$.
By Lemma \ref{1-qmean}, $B(q^*)d \leq d$. 
Apply Corollary
\ref{num-iteration-on-cone-lem}. 
\end{proof}

\begin{lem} \label{sccspprelterm} For a strongly connected PPS,  $x= P(x)$,
with LFP $q^*$, where $0 < q^* < 1$,
for any two coordinates $k, l$ of $\textbf{1}-q^*$:

\vspace*{-0.25in}

$$\frac{(\textbf{1}-q^*)_k}{(\textbf{1}-q^*)_l} \geq 2^{-(2|P|)}$$
\end{lem}

\vspace*{-0.12in}

\begin{proof} Lemma \ref{1-qmean} says that
$B(\frac{1}{2}(\textbf{1} + q^*))(\textbf{1} - q^*) \leq (\textbf{1} - q^*)$.
Since every entry of the vector $\frac{1}{2}(\textbf{1} + q^*))$ is $\geq 1/2$,
every non-zero entry of the matrix
$B(\frac{1}{2}(\textbf{1} + q^*))$ is at least
$1/2$ times a coefficient
of some monomial in some polynomial $P_i(x)$ of $P(x)$.
Moreover, $B(\frac{1}{2}(\textbf{1} + q^*))$ is irreducible.
Calling the entries of $B(\frac{1}{2}(\textbf{1} + q^*))$, $b_{i,j}$, we have a sequence
of {\em distinct} indices, $i_1,i_2,\ldots, i_m$, with $l = i_1$, $k = i_m$, $m \leq n$, where
 each $b_{i_ji_{j+1}} > 0$.
(Just take the ``shortest positive path'' 
 from $l$ to $k$.)
For any $j$:

\vspace*{-0.06in}

$$(B(\frac{1}{2}(\textbf{1} + q^*))(\textbf{1} - q^*))_{i_{j+1}} \geq b_{i_ji_{j+1}}  (\textbf{1} - q^*)_j$$

\noindent Using Lemma \ref{1-qmean} again,
$(\textbf{1} - q^*)_{i_{j+1}} \geq b_{i_ji_{j+1}}  (\textbf{1} - q^*)_{i_j}$.
By simple induction:
$(\textbf{1} - q^*)_k \geq (\prod_{j=1}^{l-1} b_{i_ji_{j+1}}) (\textbf{1} - q^*)_l$.
Note that  $|P|$ includes the encoding size of
each positive coefficient of every polynomial $P_i(x)$. 
We argued before that each $b_{i_ji_{j+1}} \geq c_{i}/2$ for 
some coefficient $c_i > 0$ of some monomial in $P_i(x)$.
Therefore, since each such $c_i$ is
a distinct coefficient that is accounted for in 
$|P|$,  we must have $\prod_{j=1}^{l-1} b_{i_ji_{j+1}} \geq 2^{-(|P| + n)} \geq
2^{-(2|P|)}$, and thus we have:
$(\textbf{1} - q^*)_k \geq 2^{-(2|P|)} (\textbf{1} - q^*)_l$.
\end{proof}

\noindent Combining Lemma \ref{linearconv} 
with Lemma \ref{sccspprelterm} establishes the following:

\begin{thm} \label{scclinearconv} 
For a strongly connected PPS, $x=P(x)$ in $n$ variables, in SNF form, with LFP $q^*$, such that 
$\textbf{0} < q^* < \textbf{1}$, if we start Newton iteration at $x^{(0)} := {\textbf 0}$,
then: $ \| q^* - x^{(j + 2|P|)}\|_{\infty} \leq 2^{-j}$.
\end{thm}

To get a polynomial upper bound on the number of iterations
of Newton's method for general PPSs, we can apply  
Lemma \ref{linearconv} combined with
a Lemma in \cite{fixp} (Lemma 7.2 of \cite{fixp}), 
which implies that for a PPS
$x=P(x)$ with $n$ variables,
in SNF form,  with LFP $q^*$, where  $q^* < \textbf{1}$,
$(\textbf{1} - q^*)_\text{min} \geq 1/{2n2^{|P|^{c}}}$ for
some constant $c$.
Instead, we prove the following much stronger result:

\begin{thm} \label{1comp} 
For a PPS, $x=P(x)$ in $n$ variables, in SNF form, with LFP $q^*$, such that 
$0 < q^* < 1$, 
for all $i= 1, \ldots, n$: $ 1-q^*_i \geq 2^{-4|P|}$.
In other words, $\| q^* \|_{\infty} \leq 1-2^{-4|P|}$. 
\end{thm}

\noindent The proof of Theorem \ref{1comp} is in the appendix.
We thus get the Main Theorem of this section:\\
{\em Proof of Theorem \ref{linearconvgen} 
{\bf (Main Theorem of
Sec. 3)}}.
By Lemma \ref{linearconv},  
$\| q^* - 
x^{(j + {\Large \lceil} (\log
\frac{(\textbf{1}-q^*)_\text{max}}{(\textbf{1}-q^*)_\text{min}}) 
{\Large \rceil})} \|_{\infty} \leq 2^{-j}$.
But by Theorem \ref{1comp}, 
$\lceil (\log \frac{({\textbf 1}-q^*)_{max}}{({\textbf 1}-q^*)_{min}}) \rceil  
 \leq  \lceil \log \frac{1}{(\mathbf{1}-q^*)_{min}} \rceil 
 \leq 
\lceil \log 2^{\; 4 |P|} \rceil    
=  4 |P|$. \qed \\

\noindent {\bf\large  Addendum.}  In Appendix 
\ref{sec:quadratic-conv-quant-decision} we extend Theorem
\ref{linearconvgen}, to show
that, given a PPS, $x=P(x)$, with LFP $0 < q^* < 1$, 
if we start Newton iteration
at $x^{(0)} := {\mathbf 0}$, then 
for all $i \geq 1$, \ 
$\| q^* - x^{(32|P| + 2 + 2i)}\|_{\infty} \leq \frac{1}{2^{2^{i}}}$.
We then use this (explicit) ``quadratic convergence'' 
result to show that the quantitative {\bf\em decision problem}
for the LFP $q^*$ of PPSs, which asks,
given a PPS $x=P(x)$ over $n$ variables, 
and given a rational number $r \in [0,1]$, decide
whether $q^*_i > r$, is decidable {\em in the unit-cost arithmetic
RAM model of computation} in polynomial time 
(and thus is reducible to PosSLP).
These results were not mentioned in the STOC'12 conference version
of this paper.  They follow from results in this paper combined
with some results we established in a subsequent paper at 
ICALP'12 (\cite{bmdp}).

\section{Polynomial time in the standard Turing model of computation}

\label{sec:p-time}

The previous section showed that for a PPS, $x=P(x)$,  
using $(4|P| + j)$
iterations of Newton's method starting at $x^{(0)}:=0$, we obtain $q^*$ within
additive error $2^{-j}$.  
However, performing even $|P|$ iterations of Newton's method {\em exactly} 
may not be feasible in P-time in the {\em Turing} model, because  
the encoding size of iterates $x^{(k)}$ can become very large.  
Specifically, by repeated squaring, the
rational numbers representing the iterate 
$x^{(|P|)}$ may require encoding size exponential in $|P|$.

In this section, we show 
that we can nevertheless approximate in P-time the LFP $q^*$ 
of a PPS, $x=P(x)$.
We do so by showing that we can
{\em round down} all coordinates of each Newton iterate $x^{(k)}$
to a suitable polynomial length, and still have a well-defined
iteration that converges in nearly
the same number of iterations to 
$q^*$.   Throughout this section we assume every PPS is in SNF form.

\begin{defn}{(``Rounded down Newton's method'', with rounding parameter $h$.)}
Given a PPS, $x=P(x)$,  
with LFP $q^*$, 
where ${\textbf 0} < q^* < {\textbf 1}$,
in the ``rounded down Newton's method'' with integer 
rounding parameter $h > 0$, 
we compute a sequence of 
iteration vectors $x^{[k]}$,  where the initial starting vector is again $x^{[0]} 
:= \mathbf{0}$,
and such that for each $k \geq 0$, given $x^{[k]}$, we compute 
$x^{[k+1]}$ as follows:

\begin{enumerate}
\item  First, compute $x^{\{k+1\}} :=  \mathcal{N}_P(x^{[k]})$, 
where the Newton
iteration operator $\mathcal{N}_P(x)$ was defined in equation 
(\ref{newton-one-it-eq}).
(Of course we need to show that all such Newton iterations are defined.)

\vspace*{-0.06in}

\item For each coordinate $i=1,\ldots,n$, set $x^{[k+1]}_i$ 
to be equal 
to the maximum (non-negative) multiple of $2^{-h}$ which is $\leq \max (x^{\{k+1\}}_i, 0)$.
(In other words, round down $x^{\{ k+1\}}$ to the nearest multiple of 
$2^{-h}$, while making sure that the result is non-negative.)
\end{enumerate}
\end{defn}

\begin{thm}[{\bf Main Theorem of Section \ref{sec:p-time}}]
\label{Ptime}  Given a PPS, $x=P(x)$, with LFP $q^*$,
such that $0 < q^* < 1$,  
if we  use the rounded down Newton's method with parameter $h = j + 2 + 4 |P|$,
then the iterations are all defined,  
for every $k \geq 0$ we have  $0 \leq x^{[k]} \leq q^*$, 
and furthermore after $h= j+2 + 4|P|$ iterations we have:
$ \hspace*{0.2in} \| q^* - x^{[j+2+4|P|]} \|_{\infty} \leq 2^{-j}$.

\end{thm}

We prove this via some lemmas.
The next lemma proves that the iterations are always well-defined, 
and yield vectors $x^{[k]}$ such that $\mathbf{0} \leq x^{[k]} \leq q^*$.
Note however that, unlike 
Newton iteration using exact arithmetic, 
we {\em do not} claim (as in Proposition 
\ref{prop:monotone-conv-newt}) that $x^{[k]}$
 converges {\em monotonically} to $q^*$. 
It may not.
It turns out we don't need this: all we need is
that $0 \leq x^{[k]} \leq q^*$, for all $k$.
In particular, it may not hold that $P(x^{[k]}) \geq x^{[k]}$. 
For establishing the
monotone convergence of Newton's method on MPSs (Proposition \ref{prop:monotone-conv-newt}), 
the fact that $P(x^{(k)}) \geq x^{(k)}$ 
is key (see \cite{rmc}). 
Indeed, note that for PPSs, once we know that $(P(x^{(k)}) - x^{(k)}) \geq 0$, 
Theorem \ref{thm:spec-full} and the defining equation  
of Newton iteration, (\ref{newton-eq}), already proves monotone convergence: 
$x^{(k)}$ is well-defined 
and $x^{(k+1)} \geq x^{(k)} \geq 0$, for all $k$.
However,  $P(x^{[k]}) \geq x^{[k]}$ may no longer hold after rounding down.
If, for instance, the polynomial $P_i(x)$ has degree 1
(i.e., has Form$_{+}$), 
 then one can show that after any positive number of iterations $k \geq 1$, 
we will have that 
$P_i(x^{\{k\}}) = x^{\{k\}}_i$.  So,
if we are unlucky, rounding down each coordinate of $x^{\{k\}}$
to a multiple of $2^{-h}$
could indeed give $(P(x^{[k+1]}))_i < x^{[k+1]}_i$.

\begin{lem} 
\label{round-exist}
If we run the rounded down Newton method
starting with $x^{[0]} := \textbf{0}$
on a PPS, $x=P(x)$, with LFP $q^*$,  ${\mathbf 0} < q^* < {\mathbf 1}$, 
then for all $k \geq 0$, $x^{[k]}$ is well-defined and $0 \leq x^{[k]} \leq q^*$.
\end{lem}

\noindent 
The next key lemma shows that the  rounded version still makes good progress
towards the LFP.

\begin{lem}
\label{lem:explicit-bound-rounded}
For a PPS, $x=P(x)$, with LFP $q^*$, such that $0 < q^* < 1$, if we
apply the rounded down Newton's method with parameter $h$,
starting at $x^{[0]} := \textbf{0}$, then for all $j' \geq 0$, we have:
$$ \| q^* - x^{[j'+1]} \|_{\infty} \leq 2^{-j'} + 2^{-h+1 + 4|P|}$$
\end{lem}

The proofs of Lemmas \ref{round-exist} and \ref{lem:explicit-bound-rounded}
use the results of the previous section and bound the effects of the rounding.
The proofs are given in the Appendix.
We can then show the main theorem:

\begin{proof}[Proof of Theorem \ref{Ptime} ({\bf Main Theorem of Sec. 4})] 
In 
Lemma \ref{lem:explicit-bound-rounded}
let $j' := j+4|P|+1$ and $h := j + 2 + 4|P|$.  We have:
$\| q^* - x^{[j + 2 + 4|P|]} \|_{\infty} \leq 2^{-(j+1+4|P|) } + 
2^{-(j+1)} \leq 2^{-(j+1)} + 2^{-(j+1)} = 2^{-j}.$
\end{proof}

\begin{cor} \label{main-alg-corollary} 
Given any PPS, $x=P(x)$, with LFP $q^*$,
we can approximate $q^*$ within additive error $2^{-j}$
in time polynomial in $|P|$ and $j$
(in the standard  Turing model of computation).
More precisely, we can compute a vector $v$, ${\mathbf 0} \leq v \leq q^*$, such that
$\| q^* - v  \|_\infty \leq 1/2^{-j}$.\end{cor}

\section{Application to probabilistic parsing of general SCFGs}

\label{sec:application-scfg}

We briefly describe an application to some important
problems for stochastic context-free grammars (SCFGs).
For definitions, background, and a detailed treatment we
refer to Appendix \ref{app-sec:application-scfg}. 

We are given a (SCFG) 
$G =  (V, \Sigma,  R, S)$ with set $V$ of nonterminals, 
set $\Sigma$ of terminals, set $R$ of probabilistic rules 
with rational probabilities, and start symbol $S \in V$. 
The SCFG induces probabilities on terminal strings,
where the probability $p_{G,w}$ of string $w \in \Sigma^*$
is the probability that $G$ derives $w$.
A basic computational problem is: 
Given a SCFG $G$ and string $w$, 
compute the probability $p_{G,w}$ of $w$. 
This probability is generally irrational, thus we want 
to compute it approximately to desired accuracy $\delta > 0$.
We give the first polynomial-time algorithm for this
problem that works for arbitrary SCFGs.

\begin{thm}
\label{string-prob}
There is a polynomial-time algorithm that, given as input
a SCFG $G$, a string $w$ and a rational $\delta >0$ in
binary representation, approximates the probability 
$p_{G,w}$ within $\delta$, i.e., computes a
value $v$ such that $|v - p_{G,w}| < \delta$.
\end{thm}

The heart of the algorithm involves the transformation of
the given SCFG $G$ to another ``approximately equivalent" SCFG $G'$ with rational rule probabilities that is in Chomsky Normal Form (CNF).
More precisely, for every SCFG $G$ there is a SCFG $G''$
in CNF that is {\em equivalent} to $G$ 
in the sense that it gives
the same probability to all the strings of $\Sigma^*$.
However, it may be the case that any such grammar must have
irrational probabilities, and thus cannot
be computed explicitly. Our algorithm computes a CNF SCFG
grammar $G'$ that has the same structure (i.e. rules) 
of such an equivalent CNF grammar $G''$ and has 
rational rule probabilities that approximate the
probabilities of $G''$ to sufficient accuracy $\delta$
(we say that $G'$ {\em $\delta$-approximates} $G''$)
such that $p_{G',w}$ provides the desired approximation
to $p_{G,w}$, and in fact this holds for all strings
up to any given length $N$.

\begin{thm}
\label{cnf}
There is a polynomial-time algorithm that, given a SCFG $G$,
a natural number $N$ in unary, and a rational $\delta >0$ 
in binary, computes a new SCFG $G'$ in CNF that 
$\delta$-approximates a SCFG in CNF that is equivalent to $G$,
and furthermore $|p_{G,w} - p_{G',w}| \leq \delta$ for
all strings $w$ of length at most $N$.
\end{thm}

The algorithm for Theorem \ref{cnf} involves a series of transformations.
There are two complex steps in the series. 
The first is elimination of $\epsilon$-rules.
This requires introduction of irrational
rule probabilities if we were to preserve equivalence,
and thus can only be done approximately. 
We effectively show that computation of
the desired rule probabilities in this transformation can
be reduced to computation of termination probabilities
for certain auxiliary grammars; 
furthermore, the structure of the construction
has the property that the reduction essentially preserves approximations.
The second complicated step is the elimination of
unary rules. This requires the solution of certain
linear systems whose coefficients are irrational
(hence can only be approximately computed)
and furthermore some of them may be extremely 
(doubly exponentially) small, which could potentially
cause the system to be very ill-conditioned.
We show that fortunately this does not happen, by
a careful analysis of the structure of the constructed
grammar and the associated system.
The details are quite involved and are given in the appendix.

Once we have an approximately equivalent CNF SCFG $G'$
we can compute $p_{G',w}$ using a well-known variant
of the CKY parsing algorithm, which runs in polynomial time
in the unit-cost RAM model, but 
the numbers may become exponentially long. 
We show that we can do the computation approximately with sufficient accuracy to obtain a good approximation
of the desired probability $p_{G,w}$ in P-time in the Turing model, and thus prove Theorem \ref{string-prob}.

We note that Chomsky Normal Form is used as the starting point
in the literature  for several other important problems
concerning SCFGs.
We expect that Theorem \ref{cnf} and the
techniques we developed in this paper will enable
the development of efficient P-time algorithms for these
problems that work for arbitrary SCFGs.

\bibliographystyle{plain}

\appendix

\section{Appendix A: Missing proofs in Sections 2-4}

\subsection{Proof of Proposition \ref{prop:snf-form}}

\noindent {\bf Proposition \ref{prop:snf-form}}
[cf. also Proposition 7.3 \cite{rmc}].
{\em 
Every PPS (MPS), $x = P(x)$, can be transformed in P-time
to an ``{\em equivalent}''  PPS (MPS, respectively),  $y=Q(y)$  in SNF form,
such that 
$|Q| \in O( |P|  )$.
More precisely, the variables $x$ are 
a subset of the variables $y$, and  $y=Q(y)$ has LFP
$p^* \in \real^m_{\geq 0}$ iff $x=P(x)$ has LFP $q^* \in \real^n_{\geq 0}$,
and projecting $p^*$ onto the $x$ variables yields 
$q^*$.}

\begin{proof}
We prove we can convert any PPS (MPS),  $x=P(x)$, 
to SNF form by adding new auxiliary variables, obtaining a different system of 
polynomial equations $y=Q(y)$ with $|Q|$  linear in $|P|$. 

To do this, we simply observe that we can use repeated squaring
and Horner's rule
to express any monomial
$x^{\alpha}$ via a circuit (straight-line program) with gates  $*$,
and with the variables $x_i$ as input.
Such a circuit will have size $O(m)$ where $m$ is the sum of the 
numbers of bits of the positive elements in the vector $\alpha$ of exponents.
We can then convert such a circuit to a system of equations,
by simply replacing the original monomial $x^\alpha$ by a new variable $y$,
and 
by simply using auxiliary variables in place of the gates of the
circuit to ``compute'' the monomial $x^\alpha$ that the variable
$y$ should be equal to.

Note that by doing this every monomial on the RHS of any of the original 
equations $x_i = P_i(x)$ will have
been replaced by a single variable, and thus those original equations
will now become Form$_{+}$ linear equations, and note that
all internal gates of the circuit 
for representing $x^\alpha$, represented by a variable $y_i$, 
give simply the product of two other variables, and
thus their corresponding equations are simply of the form $y_i = y_j y_k$,
which constitutes a Form$_{*}$ equation.

Importantly, note that the system of equations so obtained will
still remain  a system of monotone (and respectively, probabilistic) 
polynomial equations, if the original system was monotone (respectively,
probabilistic), because each new auxiliary variable $y_i$, 
that we introduce (which acts as a gate in the circuit for the
monomial $x^\alpha$),  will be associated with an equation
of the form $y_i = y_j y_j$, which indeed is both a monotone and
probabilistic equation.

Furthermore, the new system of equations $y = Q(y)$ has the property
that (a) any  solution $p'' \in \real^n_{\geq 0}$ of $y = Q(y)$, when projected on to the 
$x$ variables,
yields a solution $p' \in \real^n_{\geq 0}$ to the original system of equations, 
$x= P(x)$, and
(b) any solution $q' \in \real^n_{\geq 0}$ 
to the original system of equations $x=P(x)$
yields a {\em unique} solution $q''$ to the expanded system of  
equations, $y=Q(y)$, by uniquely solving for the values of the new auxiliary
variables using their equations (which are derived from the
arithmetic circuit).

The $O(|P| )$  bound that is claimed for $|Q|$ 
follows easily from the fact that the circuit
representing each monomial $x^\alpha$ has size $O(m )$,
where $m$ is the sum of the  numbers of bits of the positive elements in the vector $\alpha$.
\end{proof}

\subsection{Proof of Lemma \ref{int}.}

\noindent {\bf Lemma \ref{int}}.
{\em Let $x=P(x)$ be a MPS, with $n$ variables, in SNF form,
and let $a, b \in \real^n$. Then:
$$P(a)-P(b) = B(\frac{a + b}{2})(a-b)= \frac{B(a) + B(b)}{2}(a-b)$$
}

\begin{proof}
Let the function  $f: \real \rightarrow \real^n$ be 
given by $f(t) := t a + (1-t) b = b + t (a-b)$.   Define $G(t) := P(f(t))$.
 
From the fundamental theorem of calculus, and
using the matrix form of the chain rule from multi-variable calculus
(see, e.g., \cite{apostol74} Section 12.10), we have:
$$P(a)-P(b) = G(1) - G(0) = \int^1_0  B(f(t))(a-b) \,dt$$
By linearity, we can just take out $(a-b)$ from the integral 
as a constant, and we get:
$$P(a)-P(b) = (\int^1_0  B(ta + (1-t)b) \,dt) (a-b)$$
We need to show that
$$\int^1_0  B(ta + (1-t)b) \,dt = B(\frac{a + b}{2}) = \frac{B(a) + B(b)}{2} $$
Since all monomials in $P(x)$ have degree at most 2,
each entry of the Jacobian matrix $B(x)$ is a polynomial of degree $1$ over variables in $x$. For any integers $i,j$, with  
$0 \leq i \leq n$, $0 \leq j \leq n$, there are thus 
real values $\alpha$ and $\beta$ with 
$$(B(ta + (1-t)b))_{ij} = \alpha  + \beta t$$
Then
$$(\int^1_0  B(ta + (1-t)b) \,dt)_{ij} = \int^1_0 
(\alpha  + \beta t) \,dt = \alpha + \frac{\beta}{2}$$
$$(B(\frac{a + b}{2}))_{ij} = \alpha + \frac{\beta}{2}$$
$$(\frac{B(a) + B(b)}{2})_{ij} = \frac{1}{2}((\alpha + \beta) + \alpha) = \alpha + \frac{\beta}{2}$$
\end{proof}

\subsection{Proof of Lemma \ref{newton}}

\noindent {\bf Lemma \ref{newton}.}
{\em  
Let $x=P(x)$ be a MPS in SNF form.
Let $z \in \real^n$ be any vector  
such that $(I - B(z))$ is non-singular, and thus 
$\mathcal{N}_P(z)$ is defined. 
Then \footnote{Our proof of this does not use the fact that 
$x=P(x)$ is a MPS.  
We only use the fact that
$q^*$ is {\em some} solution to $x=P(x)$, that
$\mathcal{N}_P(z)$  is well-defined,
and that $P(x)$ consists
of polynomials of degree bounded by at most 2.}:

$$q^* -  \mathcal{N}_P(z) = (I-B(z))^{-1}\frac{B(q^*) - B(z)}{2}(q^* - z)$$
}

\begin{proof}
\noindent Lemma \ref{int}, applied to $q^*$ and $z$, gives:
$q^* - P(z) = \frac{B(q^*) + B(z)}{2}(q^* - z)$.
Rearranging, we get:
\begin{equation}
\label{mid-eq-for-prop-one-newt}
P(z) - z = (I - \frac{B(q^*) + B(z)}{2})(q^* - z)
\end{equation}

\noindent Replacing $(P(z) - z)$ in equation (\ref{newton-one-it-eq})
by the right hand side of equation (\ref{mid-eq-for-prop-one-newt})
and subtracting
both sides of (\ref{newton-one-it-eq}) from $q^*$, gives:

\vspace*{-0.2in}

\begin{eqnarray*}q^* -  
\mathcal{N}_P(z) & = &  (q^* - z) - (I-B(z))^{-1}
(I - \frac{B(q^*) + B(z)}{2})(q^* - z)\\
& = &   (I-B(z))^{-1} (I-B(z)) (q^* - z) - (I-B(z))^{-1}
(I - \frac{B(q^*) + B(z)}{2})(q^* - z)\\
& = &  (I-B(z))^{-1} ( (I-B(z))  - 
(I - \frac{B(q^*) + B(z)}{2}) ) (q^* - z)\\
& = &  (I-B(z))^{-1}   
(\frac{B(q^*) - B(z)}{2})  (q^* - z)
\end{eqnarray*}

\vspace*{-0.13in}
\end{proof}

\subsection{Proof of Theorem \ref{thm:spec-full}}

\noindent {\bf Theorem \ref{thm:spec-full}.} 
{\em For any PPS, $x=P(x)$,
in SNF form, if we have $0 < q^* < 1$, 
then for all $0 \leq z \leq q^*$,
$\rho(B(z)) < 1$ and $(I-B(z))^{-1}$ exists and is nonnegative.}

\begin{proof}

For any square matrix $A$, let $\rho(A)$ denote the spectral radius of $A$. 
We need the following basic fact:

\begin{lem}[see, e.g., \cite{HornJohnson85}]
\label{series} If A is a square matrix with $\rho(A) < 1$ then $(I-A)$ is non-singular, 
the series $\sum_{k=0}^\infty A^k$ converges, and
	$(I - A)^{-1} = \sum_{k=0}^\infty A^k$.
\end{lem}

For all $0 \leq z \leq q^*$,  $B(z)$ is a nonnegative matrix,
and since the entries of the Jacobian matrix $B(x)$ 
have nonnegative coefficients, $B(x)$ is monotone in $x$, i.e., if 
$0 \leq z \leq q^*$,  then $0 \leq B(z) \leq B(q^*)$,
and thus by basic facts about non-negative matrices 
$\rho(B(z)) \leq \rho(B(q^*))$.
Thus by  Lemma \ref{series} it suffices to establish
that $\rho(B(q^*)) < 1$.
We will first prove this for {\em strongly connected} PPSs:

\begin{lem} \label{scc-inverse-lem} 
For any strongly connected PPS,  $x = P(x)$,
in SNF form with LFP $q^*$, such that 
$\textbf{0} < q^* < \textbf{1}$, we have $\rho(B(q^*)) < 1$.
\end{lem}

\begin{proof}
If the Jacobian $B(x)$ is constant, then $B(q^*) = B(\textbf{1}) = B$.
In this case, $B$ is actually an irreducible substochastic matrix,
and since we have removed all variables $x_i$ such that $q^*_i =0$,
it is easy to see that some polynomial $P_i(x)$ must have contained
a positive constant term, and therefore, in the (constant) Jacobian matrix $B$ 
there is some row whose entries sum to $< 1$.
Since $B$ is also irreducible, we then 
clearly have that $\lim_{m \rightarrow \infty} B^m = 0$.
But this is equivalent to saying that $\rho(B) < 1$.
Thus we can assume that the Jacobian $B(x)$ is non-constant.
By Lemma \ref{1-qmean}:
 $$B(\frac{1}{2}(\textbf{1} + q^*))(\textbf{1}-q^*) \leq (\textbf{1}-q^*)$$
We have $\textbf{1}-q^* > 0$, and $B(\frac{1}{2}(\textbf{1} + q^*)) \geq 0$.
Thus, by induction, for any positive integer power $k$, we have
\begin{equation}\label{power-ineq}
B(\frac{1}{2}(\textbf{1} + q^*))^k (\textbf{1}-q^*) \leq (\textbf{1}-q^*)
\end{equation}
Now, since $B(x)$ is non-constant, and $B(x)$ is monotone in $x$,
and since  $q^* < \frac{1}{2}(\textbf{1} + q^*)$,
we have $B(q^*) \leq B( \frac{1}{2}(\textbf{1} + q^*))$
and furthermore there is some entry $(i,j)$ such that 
 $B(q^*)_{i,j} < B( \frac{1}{2}(\textbf{1} + q^*))_{i,j}$,
it follows that:
$$  (B(q^*)(\textbf{1}-q^*))_i <   (B( \frac{1}{2}(\textbf{1} + q^*))
(\textbf{1}-q^*))_i \leq  (\textbf{1}-q^*)_i$$
Therefore, since $B(q^*)$ is irreducible, it follows that 
for any coordinate $r$ 
there exists a power $k \leq n$ such that 
$(B(q^*)^k(\textbf{1}-q^*))_r < (\textbf{1}-q^*)_r$.
Therefore, 
$B(q^*)^n(\textbf{1}-q^*) < (\textbf{1}-q^*)$.
Thus, there exists some $0 < \beta < 1$, such that
$B(q^*)^n(\textbf{1}-q^*) \leq \beta (\textbf{1}-q^*)$.
Thus, by induction on $m$, for all $m \geq 1$,
we have $B(q^*)^{nm}(\textbf{1}-q^*) \leq \beta^m (\textbf{1}-q^*)$.
But $\lim_{m \rightarrow \infty} \beta^m = 0$,
and thus since $(\textbf{1}-q^*)>0$, it must be the case
that $\lim_{m \rightarrow \infty} B(q^*)^{nm} = 0$  (in all coordinates).
But this last statement is equivalent to saying that $\rho(B(q^*)) < 1$.
\end{proof}

Now we can proceed to arbitrary PPSs.  
We want to show that $\rho(B(q^*)) < 1$. 
Consider an eigenvector $v \in \real^n_{\geq 0}$, $v \neq 0$, 
of $B(q^*)$, associated with the eigenvalue $\rho(B(q^*))$, 
with $B(q^*)v = \rho(B(q^*))v$.   Such an eigenvector exists
by standard fact in Perron-Frobenius theory  
(see, e.g., Theorem 8.3.1 \cite{HornJohnson85}).

Consider any subset $S \subseteq \{1,\ldots,n\}$ of variable indices, 
and let
$x_S = P_S(x_S,x_{D_S})$ 
denote  the subsystem of $x=P(x)$ associated with the vector $x_S$ of 
variables in set $S$, where $x_{D_S}$ denotes the variables not in $S$.
Note that  $x_S = P_S(x_S,q^*_{D_S})$ is itself a PPS.
We call $S$ {\em strongly connected} if $x_S = P_S(x_S,q^*_{D_S})$ is
a strongly connected PPS. 

By Lemma \ref{scc-inverse-lem},
for any such strongly connected PPS given by indices $S$,
if we define its Jacobian by $B_S(x)$, then
$\rho(B_S(q^*)) < 1$. If $S$ defines a bottom strongly connected component that depends on no other
components in the system $x=P(x)$, then
we would have that $B_S(q^*)v_S = \rho(B(q^*))v_S$ where $v_S$ is the subvector of $v$ with coordinates 
in $S$.  Unfortunately $v_S$ might in general be the zero vector. 
However, if we take $S$ to be a strongly connected component 
that has $v_S \not= 0$ and such that the SCC $S$ only depends on SCCs $S'$ with 
$v_{S'} = 0$, then we still have $B_S(q^*)v_S = \rho(B(q^*))v_S$. 
Thus, by another standard fact from Perron-Frobenius
theory (see Theorem 8.3.2 of \cite{HornJohnson85}),  $\rho(B_S(q^*)) \geq \rho(B(q^*))$.
But since $\rho(B_S(q^*)) < 1$, this implies $\rho(B(q^*)) < 1$.  
\end{proof}

\subsection{Proof of Lemma \ref{cone}}

\noindent {\bf Lemma \ref{cone}}
[Lemma 5.4 from \cite{lfppoly}]{\bf .} 
{\em 
Let $x=P(x)$ be a MPS, with polynomials of 
degree bounded by 2, with LFP, $q^* \geq 0$.
Let $B(x)$ denoted the Jacobian matrix of $P(x)$.
For any positive vector 
$\textbf{d} \in \mathbb{R}^n_{> 0}$ 
that satisfies $B(q^*) \textbf{d} \leq \textbf{d}$, any positive real value $\lambda > 0$, 
and any nonnegative vector $z \in \real^n_{\geq 0}$, 
if $q^* - z \leq \lambda \textbf{d}$, and $(I-B(z))^{-1}$
exists and is nonnegative, then 
$q^* - {\mathcal N}_P(z) \leq \frac{\lambda}{2} \textbf{d}$.}

\begin{proof} By Lemma \ref{newton},  
$q^* - {\mathcal N}_P(z) = (I-B(z))^{-1}\frac{1}{2}(B(q^*) - B(z))(q^* - z)$.
Note that matrix $(I-B(z))^{-1}\frac{1}{2}(B(q^*) - B(z))$ is nonnegative:
we assumed $(I-B(z))^{-1} \geq 0$ and the positive
coefficients in $P(x)$ and in $B(x)$ mean $(B(q^*) -
B(z)) \geq 0$.  This and the assumption that
$q^* - z \leq \lambda d$ yields:  
$q^* -  {\mathcal N}_P(z)
 \leq   (I-B(z))^{-1}\frac{1}{2}(B(q^*) - B(z))\lambda \textbf{d}$.
We can rearrange as follows:
\begin{eqnarray*}
q^* -  {\mathcal N}_P(z) & \leq &  (I-B(z))^{-1}\frac{1}{2}(B(q^*) - B(z))\lambda \textbf{d}\\
& = &     (I-B(z))^{-1}\frac{1}{2}((I - B(z)) - (I - B(q^*)))\lambda \textbf{d}\\
& = &   \frac{\lambda}{2} ( I -   \;  (I-B(z))^{-1} (I - B(q^*)) \; ) \textbf{d}\\
& = &  \frac{\lambda}{2} \textbf{d} -  \; \frac{\lambda}{2} (I-B(z))^{-1} (I - B(q^*)) \textbf{d}
\end{eqnarray*}
If we can show that
$\frac{\lambda}{2}(I-B(z))^{-1}(I - B(q^*)) \textbf{d} \geq 0$, we are done. 
By assumption:
$(I - B(q^*)) \textbf{d} \geq 0$,
and since we assumed $(I - B(z))^{-1} \geq 0$ 
and $\lambda > 0$, we have:
$\frac{\lambda}{2} (I-B(z))^{-1}(I - B(q^*))\textbf{d} \geq 0$.
\end{proof}

\subsection{Proof of Theorem \ref{1comp}.}
\label{sec:app-proof-of-1comp}

Recall again that we assume that the PPS,  $x=P(x)$,
is in SNF form, 
where each equation $x_i = P_i(x)$ is either of the form $x_i= x_j x_k$,
or is of the form $x_i = \sum_{j} p_{i,j} x_j + p_{i,0}$.
There is one equation for each variable.
If $n$ is the number of variables, we can assume w.l.o.g. that 
$|P| \geq 3n$ (i.e. the input has at least 3 bits per variable).

We know that the ratio of largest and smallest non-zero components of
$\textbf{1}-q^*$ is smaller than $2^{2|P|}$ in the strongly connected case
(Lemma \ref{sccspprelterm}).  In
the general case, two variables may not depend on each other,
even indirectly. Nevertheless,  we can establish a good  upper bound on
coordinates of $q^* < 1$. As before, we start with the strongly connected
case:
\begin{thm} 
\label{thm:scc_one_fixed_lfp_bound}
Given a strongly connected PPS, $x=P(x)$, with $P(\textbf{1}) =
\textbf{1}$, with LFP $q^*$, such that $\textbf{0} < q^* < \textbf{1}$, and
  with rational coefficients, then
$$q^*_i < 1 - 2^{-3|P|}$$
for some $1 \leq i \leq n$.
\end{thm}

\begin{proof} Consider the vector $(I-B(\textbf{1}))(\textbf{1}-q^*)$. As $P(\textbf{1}) = \textbf{1}$, 
by Lemma \ref{int} we have\\ $B(\frac{1}{2}(\textbf{1} + q^*))(\textbf{1}-q^*) = \textbf{1} - q^*$ and so
$$(B(\textbf{1}) - I)(\textbf{1}-q^*) = (B(\textbf{1}) - B(\frac{1}{2}(\textbf{1} + q^*))) (1 - q^*)$$
This is zero except for coordinates of $\text{Form}_{*}$ as rows of $B(\frac{1}{2}(\textbf{1} + q^*))$ and $B(\textbf{1})$ that correspond to $\text{Form}_{+}$ 
equations are identical. If we have an expression of Form$_{*}$, $(P(x))_i = x_jx_k$, then
\begin{eqnarray*}
(B(\textbf{1}) - I) (1-q^*)_i & = & (B(\textbf{1})  - B(\frac{1}{2}(\textbf{1} + q^*))) (1 - q^*))_i\\
&= &   (1/2)(1 - q^*_k)(1-q^*_j) +  (1/2) (1-q^*_j)(1-q^*_k)\\
& = &   (1-q^*_k)(1-q^*_j)
\end{eqnarray*} 
Consequently:
\begin{equation}\label{bound-on-1-B-times_1-q-ineq}
\|(I-B(\textbf{1}))(\textbf{1}-q^*)\|_\infty \leq  \|(\textbf{1}-q^*)\|_\infty^2
\end{equation}

Now suppose that $(I-B(\textbf{1}))$ is non-singular.
In that case, we have that:
$$\textbf{1}-q^* = (I-B(\textbf{1}))^{-1}(I-B(\textbf{1}))(\textbf{1}-q^*)$$
$$\| \textbf{1}-q^* \|_\infty \leq \|(I-B(\textbf{1}))^{-1}\|_\infty \|(I-B(\textbf{1}))(\textbf{1}-q^*)\|_\infty$$
$$\|\textbf{1}-q^*\|_\infty \leq \|(I-B(\textbf{1}))^{-1}\|_\infty  \|(\textbf{1}-q^*)\|_\infty^2$$

\begin{equation}\label{bound-app-1-q-eq} \|\textbf{1}-q^*\|_\infty \geq \frac{1}{\|(I-B(\textbf{1}))^{-1}\|_\infty}
\end{equation}
where $\| \cdot \|_\infty$ on matrices is the induced norm of $\| \cdot \|_\infty$ on vectors. $\|A\|_\infty$ for an 
$n \times m$ matrix $A$ with entries $a_{ij}$ is the maximum absolute value row sum $\text{max}_{i=1}^n\sum_{j=1}^n |a_{ij}|$.

So an upper bound on $\|(I-B(\textbf{1}))^{-1}\|_\infty$ will give the lower bound on $\|\textbf{1}-q^*\|_\infty$ we are looking for.
\begin{lem} \label{invdet} Let A be a non-singular $n \times n$ matrix with rational entries. If the product of the denominators of all these entries is $m$, then
$$\|A^{-1}\|_\infty \leq nm\|A\|_\infty^n$$\end{lem}
\begin{proof}
The $i,j$th entry of $A^{-1}$ satisfies:
$$(A^{-1})_{ij} = \frac{\text{det}(M_{ij})}{\text{det}(A)}$$
where $M_{ij}$ is the $i,j$th minor of $A$, made by deleting row $i$ and column $j$.  $\|M_{ij}\|_\infty \leq \|A\|_\infty$ as we've removed entries from rows. We always have $|\text{det}(M_{ij})| \leq \|M_{ij}\|_\infty^n$ (see, e.g., \cite{HornJohnson85} page 351), so:
\begin{equation}\label{det-eq}
|(A^{-1})_{ij}| \leq \frac{\|A\|_\infty^n}{|\text{det} (A)|}
\end{equation}
Meanwhile $\text{det} (A)$ is a non-zero rational number (because by assumption $A$ is non-singular).
If we consider the expansion for the determinant 
$\text{det}(A) = \sum_\sigma \text{sgn} \sigma \prod_{i = 1}^n a_{i\sigma(i)}$, 
then the denominator of each term $\prod_{i = 1}^n a_{i\sigma(i)}$ is a product of denominators of 
distinct entries $a_{i\sigma(i)}$ and therefore divides $m$.
Since every term can thus be rewritten with denominator $m$, the sum can also be written
with denominator $m$, and therefore  $|\text{det} (A)| \geq \frac{1}{m}$.
Thus, plugging into inequality (\ref{det-eq}), we have:
$$|(A^{-1})_{ij}| \leq m\|A\|_\infty^n$$
Taking the maximum row sum $\|A^{-1}\|_\infty$,
$$\|A^{-1}\|_\infty \leq nm\|A\|_\infty^n$$
\end{proof}

If we take $(I - B(\textbf{1}))$ to be the matrix $A$ of Lemma \ref{invdet}, then
noting that the product of all the denominators in $(I-B(\textbf{1}))$ is at most $2^{|P|}$,
this  gives:
$$\|(I-B(\textbf{1}))^{-1}\|_\infty \leq n2^{|P|}\|(I-B(\textbf{1}))\|_\infty^n$$
Of course $\|(I-B(\textbf{1}))\|_\infty \leq 1 + \|B(\textbf{1})\|_\infty \leq 3$ \  (note 
that here
we are using the fact that the system is in SNF normal form).  Thus
$$\|(I-B(\textbf{1}))^{-1}\|_\infty \leq 3^nn2^{|P|}$$
Using inequality (\ref{bound-app-1-q-eq}), and since  as discussed,
w.l.o.g., $|P| \geq  3 n \geq  n \log 3 + \log n$,  this gives:
\begin{eqnarray*} \|\textbf{1}-q^*\|_\infty & \geq &  \frac{1}{n}2^{-|P|}3^{-n}
> 2^{-2|P|}
\end{eqnarray*}

\noindent Now consider the other case where $(I-B(\textbf{1}))$ is singular. 
We can look for a small solution $v$ to:

\begin{equation}\label{silly-eq}
(I-B(\textbf{1}))v = (I-B(\textbf{1}))(1-q^*)
\end{equation}

\begin{lem} Suppose we have an equation 
$Ax =b$, with $A$ a singular $n \times n$ matrix, $b$ a non-zero vector, and we know that $Ax =b$ has a solution. 
Then it must have a solution $AA'^{-1}b = b$ where $A'$ is a non-singular matrix generated from $A$ by replacing some rows with rows that have a single $1$ entry and the rest $0$.
\end{lem}

\begin{proof} 
  If $A$ has rank $r<n$, then there are linearly independent
  vectors $a_1,a_2,\ldots,a_r$ such that $a_1^T, a_2^T, \ldots, a_r^T$ are
  rows of $A$ and other rows of $A$ are linear combinations of these.
  Let $e_1, e_2,\ldots,e_n$ be the canonical basis of $\mathbb{R}^n$, i.e.
  each $e_i$ has $i$th coordinate 1 and the rest 0.  By the 
  well known fact that the set of linearly independent subsets of a vector space form a matroid,
  and in particular satisfy the exchange property of a matroid 
(see any good linear algebra or
  combinatorics text, e.g,. \cite{cameron94}, Proposition 12.8.2) ,
  we know there is a basis for $\mathbb{R}^n$ of the form
  $\{a_1,a_2, \ldots , a_r,e_{i_{r+1}}, e_{i_{r+2}},\ldots ,e_{i_{n}}\}$ for some
  choice of $i_{r+1}, i_{r+2},\ldots ,i_n$.  We form a matrix $A'$ with
  elements of this basis as rows by starting with $A$ and keeping $r$
  rows corresponding to $a_1^T, a_2^T, \ldots
  a_r^T$, and replacing the others in some order with $e_{i_{r+1}}^T,
  e_{i_{r+2}}^T,\ldots ,e_{i_{n}}^T$.  Specifically, there
  is a permutation $\sigma$ of $\{1,\ldots ,n\}$ such that if $1 \leq k
  \leq r$, the $\sigma(k)$'th row of $A'$ and $A$ are $a_k^T$ and if $r
  < k \leq n$, the $\sigma(k)$'th row of $A'$ is $e_{i_k}^T$.

$A'$ is non-singular since its rows form a basis of $\mathbb{R}^n$. It remains to show that $AA'^{-1}b = b$.
Since $Ax=b$ has a solution and the set $R$ of rows $a_1^T, \ldots, a_r^T$ spans the row space of $A$,
every equation corresponding to a row of $Ax=b$ is a linear
combination of the $r$ equations corresponding to the rows in $R$.
Therefore, if $x$ any vector that satisfies the $r$ equations corresponding to the rows in $R$ then it satisfies
all the equations of $Ax=b$. The vector $A'^{-1}b$ satisfies these $r$ equations
by the definition of $A'$. Therefore, $AA'^{-1}b = b$.
\end{proof}

We can replace some rows of $(I-B(\textbf{1}))$ to get an $A'$ using
this Lemma and then use Lemma \ref{invdet} on 
$$v' =
A'^{-1}(I-B(\textbf{1}))(1-q^*)$$ 
We still have $\|A'\|_\infty \leq 3$
and the product of all the denominators of non-zero entries is smaller
than $2^{|P|}$. As for $\|(I-B(1))^{-1}\|_{\infty}$ before:
$$\|A'^{-1}\|_\infty \leq 3^nn2^{|P|}$$
Now, using inequality (\ref{bound-on-1-B-times_1-q-ineq}), we have

\begin{equation}\label{v-prime-upperbound-eq}
\|v'\|_\infty \leq 3^n n2^{|P|} \|( \textbf{1}-q^*)\|_\infty^2
\end{equation}
Now by equation (\ref{silly-eq}), we have that
$(I - B(1)) ((1-q^*) -v') = 0$.  Thus
$(1-q^*) - v'$ is an eigenvector of $B(\textbf{1})$ with eigenvalue 1. 
But we know that $B(\textbf{1})$ is nonnegative, irreducible, and has spectral radius bigger than 1
(because $q^* < \textbf{1}$ by assumption, see e.g., \cite{rmc} proof 
of Theorem 8.1).  
Thus Perron-Frobenius theory (e.g., see Corollary 8.1.29 in \cite{HornJohnson85}) gives 
us that $(1-q^*) - v'_i$ is not a positive vector
(because the only positive eigenvectors are associated with the top eigenvalue). Thus some coordinate $i$ has:
$$v'_i \geq 1 - q^*_i$$
Thus, by inequality (\ref{v-prime-upperbound-eq}), we have:
$$1-q^*_i \leq 
 3^n n2^{|P|} \|( \textbf{1}-q^*)\|_\infty^2$$
but the proof of Lemma \ref{sccspprelterm} gave that:
$$(1-q^*_i) 2^{|P|+n} \geq \|(\textbf{1}-q^*)\|_\infty$$
Combining these inequalities, we have
\begin{eqnarray*}
1-q^*_i & \leq & 
 3^n n2^{|P|} \|( \textbf{1}-q^*)\|_\infty^2 \\
& \leq &   3^n n2^{|P|} (1-q^*_i) 2^{|P|+n} \|( \textbf{1}-q^*)\|_\infty
\end{eqnarray*}

Dividing both sides by $(1-q^*_i)$, we have that:
\begin{eqnarray*}
\|(\textbf{1}-q^*)\|_\infty & \geq &  \frac{1}{ 6^n n2^{2|P|}} \\
& > & 2^{-3|P|}
\end{eqnarray*}
\end{proof}

\begin{thm} 
\label{1comp-appendix} Given $x= P(x)$, a general PPS in SNF normal form 
with rational coefficients and with LFP,  $\textbf{0} < q^* < \textbf{1}$, then
$$q^*_i < 1 - 2^{-4|P|}$$
for all $1 \leq i \leq n$.
\end{thm}
\begin{proof}

\begin{lem} 
\label{lem:depend-cases}
Any variable $x_i$ either depends (directly or indirectly)\footnote{meaning
that in the dependency graph the other variable's node can be reached from 
the node corresponding to $x_i$.} 
on a variable in 
a bottom SCC $S$ such that 
$P_S(\textbf{1}) = \textbf{1}$
(meaning there is no directly ``leaking'' variable in that SCC), 
or it depends (directly or indirectly) on some variable $x_j$ of 
$\text{Form}_\text{+}$ with $P(x)_j = p_{j,0} + \sum_{j=1}^n p_{i,j} x_j$
where $\sum_{j=0}^m p_{i,j} < 1$ (thus, a leaky variable).\end{lem}

\begin{proof} Suppose that in the set of variables $x_i$ depends on, $D_i$, every variable of $\text{Form}_\text{+}$, $x_j$, with $P(x)_j = p_{j,0} + \sum_{k=1}^n p_{j,k} x_k$
has $\sum_{j=0}^m p_{i,j} = 1$. Then we can verify that $P_{D_i}(\textbf{1}) = \textbf{1}$. $D_i$ contains some bottom SCC $S \subseteq D_i$. For this SCC $P_S(\textbf{1}) = \textbf{1}$\end{proof}

Suppose that $x_j$ is of $\text{Form}_\text{+}$ with $P(x)_j = p_{j,0} + \sum_{k=1}^n p_{j,k} x_k$ where $\sum_{k=0}^m p_{j,k} < 1$.
Then $q^*_j = P(q^*)_j$ has $q^*_j \leq \sum_{k=0}^m p_{j,k}$. $1 - \sum_{k=0}^m p_{j,k}$ is a rational with a denominator smaller than the product of the denominators of all the  $p_{j,k}$. We have:
$$1 - \sum_{k=0}^m p_{j,k} \geq 2^{-|P|}$$
Thus in such a case:
$$q^*_j \leq 1 - 2^{-|P|}$$

Lemma \ref{lem:depend-cases} says that any $x_i$ either depends 
on such a variable, or on 
a variable to which Theorem \ref{thm:scc_one_fixed_lfp_bound} applies. 
That is, $x_i$ depends on some $x_j$ with
$$q^*_j \leq 1 - 2^{-3|P|}$$
There is some sequence $x_{l_l}, x_{l_2}, \ldots, x_{l_m}$ with $l_1 = j$, $l_2 - i$ and for every $0 \leq k < m $, $P(x_{l_{k+1}})$ contains a term with $x_{l_{k+1}}$.
If $x_{l_{k+1}}$ has $\text{Form}_\text{*}$, then $q^*_{l_{k+1}} \leq q^*_{l_k}$.  If $x_{l_{k+1}}$ has $\text{Form}_\text{+}$, then 
$1 - q^*_{l_{k+1}} \geq p_{l_{k+1},l_k} (1-q^*_{l_k})$. By an easy induction:
$$1 - q^*_i \geq (\prod_{x_{l_k} \text{ has Form}_\text{+}} p_{l_{k+1},l_k} ) (1- q^*_j)$$
Again, $|P|$ is at least the number of bits describing these rationals $p_{l_{k+1},l_k}$, and thus
$$1 - q^*_i \geq 2^{-|P|} (1- q^*_j)$$
Since we already know that $q^*_j \leq 1- 2^{-3|P|}$, i.e.,
that $(1-q^*_j) \geq  2^{-3|P|}$, we obtain:
$$1 - q^*_i \geq 2^{-|P|} 2^{-3|P|} = 2^{-4|P|}$$
This completes the proof of the theorem.
\end{proof}

\subsection{Proof of Lemma \ref{round-exist}}

\noindent {\bf Lemma  \ref{round-exist}}.
{\em If we run the rounded down Newton method
starting with $x^{[0]} := \textbf{0}$
on a PPS, $x=P(x)$, with LFP $q^*$,  ${\mathbf 0} < q^* < {\mathbf 1}$, 
then for all $k \geq 0$, $x^{[k]}$ is well-defined and $0 \leq x^{[k]} \leq q^*$.}

\begin{proof} We prove this by induction
on $k$.  The base case $x^{[0]} = 0$ is immediate. Suppose the claim
holds for $k$ and thus $0 \leq x^{[k]} \leq q^*$. Lemma \ref{newton} 
tells us that
$$q^* -  x^{\{k+1\}} = (I-B(x^{[k]}))^{-1}\frac{B(q^*) - B(x^{[k]})}{2}(q^* - x^{[k]})$$
Now the fact that $0 \leq x^{[k]} \leq q^*$ yields that 
each of the following inequalities hold: 
$(q^* - x^{[k]}) \geq 0$, $B(q^*) - B(x^{[k]}) \geq 
0$.
Furthermore, by Theorem \ref{thm:spec-full}, we have that
$\rho(B(x^{[k]})) < 1$, and thus that $(I-B(x^{[k]}))$ is non-singular and
$(I-B(x^{[k]}))^{-1} \geq 0$.
We thus conclude that $q^* -  x^{\{k+1\}} \geq 0$,
i.e., that $x^{\{k\}} \leq q^*$.
The rounding down 
ensures that $0 \leq x_i^{[k+1]} \leq x_i^{\{k+1\}}$
unless $x_i^{\{k+1\}} < 0$, in which case $x_i^{[k+1]} =0$.
in both cases, we have that  $0 \leq x^{[k+1]} \leq q^*$. 
So we are done by induction.
\end{proof}

\subsection{Proof of Lemma \ref{lem:explicit-bound-rounded}}

\noindent {\bf Lemma \ref{lem:explicit-bound-rounded}}.
{\em For a PPS, $x=P(x)$, with LFP $q^*$, such that $0 < q^* < 1$, if we
apply the rounded down Newton's method with parameter $h$,
starting at $x^{[0]} := \textbf{0}$, then for all $j' \geq 0$, we have:
$$ \| q^* - x^{[j'+1]} \|_{\infty} \leq 2^{-j'} + 2^{-h+1 + 4|P|}$$}

\begin{proof} 
Since $x^{[0]} := 0$:

\begin{equation}
\label{base-equation}
q^* - x^{[0]} = q^* \leq \textbf{1}   \leq \frac{1}{(\textbf{1} - q^*)_\text{min}} (\textbf{1} - q^*)
\end{equation}
For any $k \geq 0$, if
$q^*-x^{[k]} \leq \lambda (\textbf{1} - q^*)$, then
by Lemma \ref{cone}
we have:

\begin{equation}
\label{eq:half-in-round}
q^*-x^{\{k+1\}}\leq (\frac{\lambda}{2}) (\textbf{1} - q^*)
\end{equation}
Observe that after every iteration $k > 0$,  in every coordinate $i$ we have:

\begin{equation}
\label{eq:almost-bigger-in-round}
x_i^{[k]} \geq x_i^{\{k\}} - 2^{-h}
\end{equation}
This holds simply because we are rounding down $x^{\{k\}}_i$ by
at most $2^{-h}$, unless it is negative in which case $x^{[k]}_i =0> x^{\{k\}}_i$.
Combining the two inequalities (\ref{eq:half-in-round}) and
(\ref{eq:almost-bigger-in-round})
 yields the following inequality:

$$q^*-x^{[k+1]}\leq (\frac{\lambda}{2}) (\textbf{1} - q^*) + 2^{-h} \textbf{1} \leq
(\frac{\lambda}{2} + \frac{2^{-h}}{(\textbf{1} - q^*)_\text{min}}) (\textbf{1} - q^*)$$

\noindent Taking inequality (\ref{base-equation}) as the base case (with $\lambda =
\frac{1}{(\textbf{1} - q^*)_\text{min}}$), by 
induction on $k$, for all 
$k \geq 0$:

$$q^*-x^{[k+1]} \leq (2^{-k} + \sum_{i = 0}^k 2^{-(h+i)}) \frac{1}{(\textbf{1} - q^*)_\text{min}} (\textbf{1} - q^*)$$

\noindent But $\sum_{i = 0}^k 2^{-(h+i)} \leq 2^{-h+1}$ and 
$\frac{\|\textbf{1} - q^*\|_\infty}{(\textbf{1} - q^*)_\text{min}} 
\leq \frac{1}{ (\textbf{1} - q^*)_\text{min}} \leq 2^{4|P|}$,
by Theorem \ref{1comp}.  Thus:
$$q^*-x^{[k+1]} \leq (2^{-k} + 2^{-h+1})2^{4|P|} \textbf{1}$$
Clearly, we have $q^*-x^{[k]} \geq 0$ for all $k$. Thus we have shown
that for all $k \geq 0$:
$$\|q^*-x^{[k+1]}\|_\infty \leq (2^{-k} + 2^{-h+1})2^{4|P|} = 2^{-k} + 2^{-h + 1 + 4 |P|}.$$
\end{proof}

\subsection{Proof of Corollary \ref{main-alg-corollary}}

\noindent {\bf Corollary \ref{main-alg-corollary}.}
{\em Given any PPS, $x=P(x)$, with LFP $q^*$,
we can approximate $q^*$ within additive error $2^{-j}$
in time polynomial in $|P|$ and $j$
(in the standard  Turing model of computation).
More precisely, we can compute a vector $v \leq q^*$ such that $v \in [0,1]^n$ and
$\| q^* - v  \|_\infty \leq 1/2^{-j}$.}

\begin{proof} 
Firstly, by Propositions \ref{prop:snf-form} and
\ref{prob1-ptime-scfg-prop},  we
can assume $x=P(x)$ is in SNF form, and that ${\textbf 0} < q^* < {\textbf 1}$.
By Theorem \ref{Ptime}, 
the rounded down Newton's method 
with parameter $h = j + 2 + 4 |P|$, for $h = j+2 + 4|P|$ iterations, 
computes a rational vector $v = x^{[h]}$
such that $v \in [0,1]^n$, and 
$\| q^* - v  \|_\infty \leq 1/2^{-h}$.

Furthermore, for all $k$, with $0 \leq k \leq h$,
$x^{[k]}$ 
has encoding size polynomial in $|P|$ and $j$.
We then simply need to note that
all the linear algebra operations, that is: matrix multiplication,
addition, and matrix inversion, required in a single
iteration of Newton's method, can be performed exactly on rational inputs in
polynomial time and yield rational results with a polynomial size.
\end{proof}

\section{Appendix B: Application to parsing for general SCFGs}

\label{app-sec:application-scfg}

In this section, we use our P-time algorithm for approximating 
the extinction probabilities $q^*$ of BPs, 
equivalently the termination probabilities of a
{\em stochastic context-free grammar} (SCFG), a.k.a., 
the {\em partition function} of an SCFG, 
in order to 
provide the first P-time algorithm for solving  
an approximate version of a key probabilistic 
parsing problem, with respect
to {\em any} given SCFG,
including grammars that contain $\epsilon$ rules, i.e., rules of the form 
$A \stackrel{p}{\rightarrow} \epsilon$, where $\epsilon$
denotes the empty string and $A$ is an arbitrary nonterminal 
of the grammar.

{\em String Probability Problem:} 
Given a SCFG $G$, and a finite string $w \in \Sigma^*$ over
the terminal alphabet $\Sigma$ of $G$, 
compute the probability, $p_{G,w}$,  that the 
stochastic grammar $G$  
generates the finite string $w$.

For SCFGs that are already in Chomsky Normal Form (CNF)
there is a well known dynamic programming algorithm 
for this
problem (see, e.g., \cite{ManSch99}).  This is based on 
a direct extension to the probabilistic
setting of the classic Cocke-Kasami-Younger (CKY)
dynamic programming algorithm 
for parsing, i.e., determining whether an ordinary context-free grammar
in CNF form
can generate a given string $w$
(see, e.g., \cite{HU79}).

As is well known, any ordinary (non-stochastic)
context-free grammar can be converted to one in CNF form
that generates exactly the same set of strings.

However, the situation for SCFGs is more subtle.
It is known that for every SCFG, $G$, 
there {\em exists} another
SCFG, $G'$, that is in CNF form,  which has the same
probability of generating
any finite string.  This was shown by 
Abney, McAllester, and Pereira in \cite{AMP99}.
As mentioned in \cite{AMP99} their proof of this is nonconstructive 
and yielded no algorithm to obtain $G'$
from $G$.  Moreover, as we shall see,
when $G$ contains only rules with rational
probabilities, it can nevertheless be the case that there does not exist
any $G'$ in CNF which also  has only rational rule
probabilities and which generates every string with 
the same probability as $G$.   In other words, 
the claim that every SCFG $G$ with rational rule probabilities is 
``equivalent'' to an SCFG $G'$ in CNF form 
form only holds in general if $G'$ is allowed to have irrational
rule probabilities (even though $G$ does not).

We shall nevertheless show that both these issues
can be overcome: (1) the nonconstructive
nature of the prior existence arguments, and (2) the fact that CNF form SCFGs
must in general have irrational rule probabilities.   
In fact, we shall give a constructive transformation
from any SCFG to one in CNF form, and we shall
show that given any SCFG $G$
with rational rule probabilities,
it is possible to compute in P-time a new CNF SCFG, $G'$,
which also has only rational rule probabilities, and which
suitably approximates $G$.
Our proof of this shall make crucial use of our P-time algorithm
for approximating the LFP $q^*$ of a PPS, i.e., approximating
termination probabilities for arbitrary SCFGs.

When an SCFG $G$ is already in CNF form, assuming its rule
probabilities are rational,
probabilistic versions of the CKY algorithm can 
be used to compute 
the {\em exact and rational} 
probabilities $p_{G,w}$.
These algorithms use only polynomially many $\{+,*\}$-arithmetic
operations over the rule probabilities of the SCFG $G$,
and
thus they run in P-time in the unit-cost arithmetic
RAM model of computation.
Variants of these  algorithms can be made to work more generally,
when the SCFG is not in CNF form but contains no $\epsilon$-rules.
However, for general SCFGs that do
contain arbitrary $\epsilon$-rules,  these algorithms do not work.

Let us see why,
unlike ordinary CFGs,
it is not possible to ``convert'' an arbitrary SCFG
with $\epsilon$-rules and rational rule probabilities 
to an ``exactly equivalent'' SCFG in CNF form, 
without irrational rule probabilities.
This follows easily from the fact that
the total parse probabilities $p_{G,w}$ can in general be irrational,
even when all rule probabilities of $G$ are rational,
and therefore we can not possibly compute
$p_{G,w}$ exactly using only finitely many $\{+,*\}$-arithmetic operations
over the rational rule probabilities.  
To see that $p_{G,w}$ can be irrational,
consider the following simple grammar:
$S \stackrel{1}{\rightarrow} aA$,
$A \stackrel{1/2}{\rightarrow} AA$, 
$A \stackrel{1/4}{\rightarrow} \epsilon$,
$A  \stackrel{1/4}{\rightarrow} NN$,
and $N \stackrel{1}{\rightarrow} NN$.
It is easy to see that the probability that this grammar generates
the string $a$ is precisely 
 the probability of termination starting at nonterminal $A$, which is
the least non-negative solution of $x = (1/2) x^2 + 1/4$, which
is $1- 1/\sqrt{2}$.

As we saw, the {\em total} probability
of $G$ generating $w$, $p_{G,w}$,
may be irrational for SCFGs $G$ with $\epsilon$-rules.
In typical applications one will not need to compute
$p_{G,w}$ 
``exactly''.
It will suffice to approximate it ``well'',
or to answer questions about it such as whether $p_{G,w} \geq r$
for a given rational probability $r \in [0,1]$.  It is easy to see
that these questions are at least as hard as the corresponding questions
for the termination probabilities of a SCFG.

\begin{prop}
Given a SCFG $G$ (with rational probabilities) we can construct in polynomial time another SCFG $H$
such that the termination probability $p_G$ of $G$ is equal to
the probability $P_{H,\epsilon}$ that $H$ generates the empty string $\epsilon$
(or any other string $w$).
Hence, the problem of deciding if $p_{H,\epsilon} \geq 1/2$ for a given SCFG $H$
is SQRT-SUM-hard and PosSLP-hard.
\end{prop}
\begin{proof}
Given a SCFG $G$, remove all the terminal symbols from the right hand sides
of all the rules,  and let $H$ be the resulting SCFG.
Clearly, there is a 1-to-1 correspondence between the derivations of $G$ and $H$.
A derivation of $G$ derives a terminal string iff the corresponding derivation of $H$
derives the empty string $\epsilon$.
Therefore, $p_G = p_{H,\epsilon}$.
The SQRT-SUM-hardness and PosSLP-hardness of the problem of deciding
whether  $p_{H,\epsilon} \geq 1/2$ follows from the hardness
of the analogous question $p_{G} \geq 1/2 ?$ for the termination probability (\cite{rmc}).

We can modify the  reduction, if desired, to show the same result for
any string $w$ instead of $\epsilon$: just add to $H$
a new start nonterminal $S'$, with the rule 
$S' \stackrel{1}{\rightarrow} S w$.
\end{proof}

Since deciding whether $p_{G,w} \geq p$ is hard, 
we focus on {\em approximating} the probabilities
$p_{G,w}$. 
The proposition implies  that the problem of approximating the probability $p_{G,w}$ 
for a given SCFG $G$ and string $w$
is at least as hard as the problem of approximating the termination probability of
a given SCFG. We will show in this section the converse:
 As we shall see, we will be able to use our P-time
approximation algorithm for SCFG termination probabilities, combined
with other ideas, to approximate $p_{G,w}$ in P-time.  It is important
to note that all of our P-time algorithms are in the standard Turing
model of computation, not in the more powerful unit-cost arithmetic
RAM model of computation.

Computation of the total probabilities $p_{G,w}$ 
forms a key ingredient in the well known {\em inside-outside algorithm} 
\cite{Baker79,LariYoung90} for learning the maximum likelihood parameters
of an SCFG by using example parses as training instances.
Namely, the {\em inside} subroutine of the inside-outside algorithm
is precisely a subroutine for computing $p_{G,w}$.  
However, as mentioned, the well known CKY-based dynamic programming 
algorithm for computing $p_{G,w}$  
applies only to stochastic grammars that are
already in CNF form, and in particular
to grammars that have
no $\epsilon$ rules, or else have only one $\epsilon$ rule 
associated with
the start nonterminal $S$, where $S$ can not appear on the
right hand side (RHS) of any other rule.

The inside-outside algorithm can be viewed as an extension of the Baum-Welch 
forward-backward algorithm 
for finite-state hidden Markov models 
(or, more generally, a version of the EM algorithm for finding
maximum likelihood parameters for statistical models),
to the setting of SCFGs with hidden parameters.
It 
is used heavily both in statistical NLP 
(\cite{ManSch99})
and in biological sequence
analysis (see \cite{DEKM99}) for learning the parameters
of a stochastic grammar based on parsed training instances.
In the case of biological sequence analysis, SCFGs are used 
for predicting the secondary structure (i.e., two dimensional
folding structure) of RNA molecules based on their nucleic acid sequence,
where a given parse of a sequence corresponds to a particular folding pattern.
The parameters of the SCFG are learned
using given folded RNA strands as training instances  
(see, e.g., \cite{SBHMSUH94,DEKM99}).
Some standard and natural grammars used for analysing 
RNA secondary structure do contain $\epsilon$-rules,
and are not in CNF form (see, e.g., the RNA grammar in \cite{DEKM99}, on page 273).
For these grammars, typically what researchers do is to devise 
tailor-made algorithms specific to the grammar 
for computing probabilities like $p_{G,w}$. 
It is clearly desirable to instead 
have a version of the inside-outside algorithm,
as well as versions of other algorithms related to probabilistic parsing, 
that are applicable on arbitrary SCFGs, including those with $epsilon$-rules, 
since in many applications
the most natural grammar may not be in CNF form and may require
$\epsilon$-rules.

Note that, as mentioned, 
in general, it is not true that any SCFG $G$ with rational rule
probabilities can be converted to a CNF form SCFG $G'$ with rational
rule probabilities which generates {\em exactly} the same probability 
distribution on strings.  
To obtain exactly the same probability distribution on strings,
the CNF grammar $G'$ may need to contain rules with
irrational probabilities.

We now begin a detailed formal treatment.
Recall that an SCFG $G =  (V, \Sigma,  R, S)$ consists
of a finite set $V$ of {\em nonterminals}, a 
{\em start} nonterminal $S \in V$,
a finite set $\Sigma$ of {\em alphabet} symbols,
and a finite list of {\em rules}, $R = \langle r_1, \ldots, r_k \rangle$,
where each rule $r_i$ is specified by a triple
$(A,p, \gamma)$ which we shall 
denote by $A \stackrel{p}{\rightarrow} \gamma$,  where $A \in V$ is a nonterminal,
$p \in [0,1]$ is the probability  associated with this rule, 
and $\gamma \in (V \cup \Sigma)^*$ is
a {\em possibly empty} string of terminals and nonterminals.
A rule 
of the form $A \stackrel{p}{\rightarrow} \epsilon$, where
$\epsilon$ is the empty string, is called an {\em $\epsilon$-rule}.

For technical convenience, we allow for the possibility that two
distinct rules $r_i$ and $r_j$, $i \neq j$, may nevertheless correspond to 
exactly the same triple $(A,p,\gamma)$, and that they may 
have rules with both identical left hand side (LHS) and right hand side (RHS).
For these reasons, we distinguish different rules $r_i$ and $r_j$ by their 
indices $i$ and $j$, and that is why $R$ is not viewed as a {\em set},
but a list of rules.

For a rule $r \in R$ whose corresponding triple is $(A,p,\gamma)$, we 
define $LHS(r) := A$,  $p(r) := p$, and $RHS(r) := \gamma$. 
For a rule $r$, where $LHS(r) = A$, we say
that rule $r$ is {\em associated} with the nonterminal $A$.
Let $R_A$ denote the set of rules associated with nonterminal $A$.
We call $G$ a {\em stochastic} (or {\em probabilistic}) 
context-free grammar (SCFG or PCFG), if 
for every nonterminal 
$A \in V$, we have $p_A \leq 1$, where $p_A$ is
defined as the
sum of the probabilities of rules in $R_A$,
i.e.:  \[ p_A := \sum_{r \in R_A} p(r) \] 

An SCFG is called {\em proper}
if $p_A = 1$ for all nonterminals $A$. 
It is however easy to see that requiring properness 
for SCFGs is without loss
of generality, even when the grammar needs to be in a
special normal form, such as CNF, because we can always 
make the stochastic grammar proper by adding an extra rule 
$A \stackrel{1-p_A}{\rightarrow}  NN$  which
carries the residual probability $(1-p_A)$,
where $N$ is a new nonterminal and where there is a 
new rule $N \stackrel{1}{\rightarrow} NN$.
This yields a new proper SCFG which has exactly the same probability
of generating any particular finite string of terminals as did the 
old SCFG, and has the same finite parse trees with the same probability.
We can therefore assume, w.l.o.g., that all the input SCFGs we consider
in this paper are proper.
\footnote{In some of our algorithms, while processing input SCFGs,
they may become improper, in which case we can clearly
then convert them back again to proper SCFGs, by the same method.
It is worth noting that 
in some definitions of PCFGs,
notably in \cite{NedSat08b}, 
the authors even permit the sum of the probabilities
of rules associated with a given nonterminal $A$ to be $p' >  1$.
Specifically, they define PCFGs to be
any {\em weighted context-free grammar} where all rules have a
weight $p \in [0,1]$,  without the condition that the weights associated
with a given nonterminal $A$ must sum to $\leq 1$,  but with the
added condition that the {\em total weight} of generating any finite
string must be in $[0,1]$.  The ``weight'' of a given
string generated by a weighted grammar is defined and computed analogously 
to the way we compute the total probability of a string being
generated by an SCFG.  We shall not elaborate on the definition here.
As we shall discuss, this definition
of PCFGs as a more general subclass of weighted context-free grammars is 
in fact too general in several important ways.  In particular, 
we showed in \cite{rmc}
that such weighted grammars subsume the general RMC model, 
for which we proved in \cite{rmc}
that computing or even approximating termination probabilities 
to within 
any nontrivial approximation threshold is already 
at least as hard as some long standing open problems
in numerical computation, namely SQRT-SUM and PosSLP, 
which are problems not even known to be in NP.
Thus it is unlikely that one could devise a P-time algorithm
for approximating the ``termination probability'' for the
generalized definition of PCFGs based on weighted grammars
that is given by \cite{NedSat08b}.
However, the important point is that, as we show, we don't need to 
solve this more general problem in order to approximate $p_{G,w}$ 
for standard SCFGs.
We restrict ourselves in this paper to the more standard definition of
SCFGs (or PCFGs).  Namely, we assume that probability of rules
associated with each nonterminal must sum to $\leq 1$, and in fact
w.l.o.g., that they sum to exactly $1$.}

An SCFG is in {\em Chomsky Normal Form} (CNF) if it satisfies 
the following conditions:
\begin{itemize}
\item  The grammar does not contain any $\epsilon$-rule except possibly for a rule $S \stackrel{p}{\rightarrow}
\epsilon$
associated with the start nonterminal $S$;   if it does contain such a rule then $S$ does not
appear on the right hand side of any rule in the grammar.
  
\item Every rule, other than $S \stackrel{p}{\rightarrow} \epsilon$,
is either of the form $A \stackrel{p}{\rightarrow} BC$,
or of the form $A \stackrel{p}{\rightarrow} a$
where $A$, $B$, and $C$ are nonterminals in $V$ and $a \in \Sigma$ is
a terminal symbol.
\end{itemize}

We define a finite {\em parse tree}, $t$, for a string $w \in \Sigma^*$
in a SCFG, $G$, starting at (or {\em rooted at}) nonterminal $A$,
to be a rooted, labeled, ordered, finite tree,
such that all leaf nodes of $t$ are labeled by a terminal symbol in 
$\Sigma$ or by $\epsilon$, and such that all
internal (i.e., non-leaf)
nodes are labeled by a pair $(B,r)$,
where $B \in V$ is some nonterminal of the grammar,
and $r \in R_B$ is some rule of the SCFG associated with $B$.
For an internal node $z$ that has label $(B,r)$, we define
$L_1(z) := B$, and  $L_2(z) := r$ to describe its two labels.

If an internal node $z$ has $L_2(z) = r$,  
and $RHS(r) =  \gamma$,
then node $z$ must have exactly $|\gamma|$ children, 
unless $\gamma = \epsilon$.
The children of $z$ are then labeled, from left to right, 
by the sequence of symbols in $\gamma$.  (If $\gamma = \epsilon$,
then the single child is a leaf labeled by $\epsilon$.).
Finally, for $t$ to be a finite parse tree for the string $w \in \Sigma^*$,
it must be the case that the labels on the leaves of the tree, 
when concatenated together from left
to right, form precisely the string $w$.  Note that the
empty string, $\epsilon$, is an identity element for the 
concatenation operator. 

We consider two parse trees to be identical
if they are isomorphic as rooted labeled ordered trees,
where the label of an internal node $z$ includes both the associated
nonterminal $L_1(z)$ and the associated rule $L_2(z)$.
We use $T^A_{G,w}$ to denote the set of distinct parse trees
rooted at nonterminal $A$ 
for the string $w \in \Sigma^*$ 
in the SCFG $G$.

We now describe a probabilistic 
{\em derivation} for a SCFG,  $G$, starting at a nonterminal 
$A$, as a stochastic process that proceeds
to generate a random derivation tree, which may either be infinite, or may
terminate at a finite parse tree.  The derivation
starts as a tree $T_0$ with a single root node $root$ which 
the process will randomly ``grow'' into a tree as follows.
The root is labeled by the start nonterminal $A$, so $L_1(root) := A$. 
At each step of the derivation, for every current leaf node, $z$, of the 
current derivation tree, $T_j$, such that the leaf node $z$ has $L_1(z) = A$,
we ``expand'' that  
nonterminal occurrence by randomly choosing a 
rule $r \in R_A$, letting $L_2(z) := r$, 
where the rule $r$ is chosen independently at random
for each such leaf node, according to the probabilities $p(r)$ 
of the rules $r \in R_A$. \footnote{Note that if the SCFG is not
proper, then as described this is not a well-defined 
stochastic process.  We rectify this by simply asserting
that if the SCFG is not proper, then with the
residual probability $(1-p_A)$ at every leaf labeled by $A$
we generate two new children labeled by a new special nonterminal $N$
which will generate an infinite tree with probability 1,
via a rule $N \stackrel{1}{\rightarrow} NN$.  This corresponds
to the way we converted any (CNF-form) SCFG into a proper (CNF-form) SCFG.}
We then 
use the chosen rule $r$ to
add $|RHS(r)|$ new children for that leaf node, where
these children are labeled, from left to right, by the
sequence of terminal and nonterminal symbols 
in $\gamma = RHS(r)$.  If $\gamma=\epsilon$, then there
is only one child added, labeled by $\epsilon$. 
We continue to repeat this ``expansion'' process 
until the derivation yields a finite
parse tree having only terminal symbols (including possibly $\epsilon$) 
labeling all of its leaves, in which 
case the process stops. Otherwise, i.e., if the process never encounters a
finite parse tree having only terminal symbols labeling the leaves, 
then the derivation never stops and goes on forever, generating an
a infinite sequence of larger and larger 
derivation trees.   If the derivation stops 
and generates a finite parse tree, $t$,
then if the concatenation of the sequence of symbols on the
leaves of that parse tree $t$ is a string $w \in \Sigma^*$,
we say that the derivation process on the SCFG $G$, starting at 
nonterminal $A$, 
has generated the string $w$.
We use 
$P_{G,A}(t)$ to denote the probability
that the finite parse tree $t$ rooted at $A$
is generated by grammar $G$ starting at nonterminal $A$.
It is clear that $P_{G,A}(t)$ is
the product 
over all internal nodes of $t$ of
the probability of the rule associated with that internal node.  
In other words:

\begin{equation}\label{prod-parse-tree-prob-eq}
P_{G,A}(t) =  \prod_{\{ \mbox{$z$} \; \mid \; \mbox{$z$ is an internal node of $t$\} }} p(L_2(z)) 
\end{equation}

We denote by
$p^A_{G,w}$ the probability that starting
at nonterminal $A$ of the grammar $G$ the derivation process generates 
the string $w$.
Clearly we have:
\begin{equation}\label{sum-of-parse-probs-eq}
p^A_{G,w} = \sum_{t \in T^A_{G,w}} P_{G,A}(t)
\end{equation}

We now extend the definition of ``derivation'' process,
so that it can start not just at a nonterminal, but at
a string of terminals and nonterminals, as follows.

For any string $\gamma \in (V \cup \Sigma)^*$ of terminals and 
nonterminal of $G$, if $\gamma = \epsilon$, the
derivation process simply begins and ends with a tree consisting
of one node labeled by $\epsilon$.
Otherwise, if $\gamma = \gamma_1 \ldots \gamma_m$,
where $\gamma_i \in (V \cup \Sigma)$ for $i=1,\ldots, m$,
the derivation process consists of a sequence of derivation
processes, starting at each symbol $\gamma_i$, for $i$ going from 
$1$ to $m$.  If $\gamma_i$ is a nonterminal $A$, then the derivation
process is the same as that starting at $A$.  If $\gamma_i$
is a terminal symbol, then the termination process starting at $\gamma_i$
simply begins and ends with a tree consisting of
one node labeled by the terminal symbol $\gamma_i$.
If the entire sequence of derivation processes terminate and
generate finite parse trees, then if the sequential concatenation
of the strings generated by each of this sequence of parse trees 
yields a string $w \in \Sigma^*$ we say that this derivation
process starting at $\gamma$ generated the string $w$.
Let  $p^\gamma_{G,w}$ denote the probability that 
derivation of the SCFG $G$,
starting 
with the string $\gamma$,
generates the terminal string $w$.

The {\em termination probability} of an SCFG, $G$,
starting at nonterminal $A$, denoted $q^{A}_{G}$,
is the probability with which the derivation
process starting at $A$ eventually stops and generates a finite string, and a 
finite parse tree. It 
is clearly given by:
\[ q^A_{G} = \sum_{w \in \Sigma^*}  p^A_{G,w} \]

An SCFG $G$ is called {\em consistent} if $q^S_G = 1$,
where $S$ is the start nonterminal of $G$. 
Note that even if the given SCFG $G$ is proper
(meaning the probabilities of rules associated with every nonterminal
 sum to 1), 
this \underline{does not} necessarily imply that $G$ is 
{\em consistent}.  
Indeed, we know that proper SCFGs need
not terminate with probability 1.  For example, the SCFG  given by
$S \stackrel{2/3}{\rightarrow} SS$ , $S \stackrel{1/3}{\rightarrow} a$,
is proper but only terminates with probability 1/2.

For the encoding of {\em input} SCFGs, for purposes
of analyzing the complexity of algorithms,
we assume that the probabilities associated with each rule of the
input SCFG are rational values encoded in the usual
way, by giving their numerator and denominator in binary.
We shall use $|G|$ to denote the encoding size of an input SCFG $G$,
i.e., the number of bits required to represent $G$ with this
binary encoding for the rational rule probabilities.
In our formal analysis, when reasoning about our algorithms,
we will in fact need to consider 
SCFGs whose rule probabilities
can be irrational, but we shall not  need to actually compute these
probabilities exactly, only approximately.

The general statement of the {\em approximate total
probability parsing problem} is as follows.
We are given as input: an SCFG, $G$,  which w.l.o.g. we assume
to be {\em proper},
we are  also given a finite word  $w = w_1 \ldots w_n \in \Sigma^*$   
over the   terminal alphabet $\Sigma$ of the SCFG $G$.  Finally,
we are also given a rational error threshold  $\delta > 0$.
As output,
we wish to approximate within error $\delta$ 
the probability that a probabilistic derivation
of $G$ generates the string $w$, which we
denote by $p_{G,w} := p^S_{G,w}$, where $S$ is
the start nonterminal of $G$.
In other words, we want our algorithm to output a rational value $v \in [0,1]$ 
such that $| v - p_{G,w} | < \delta$.
Importantly, we allow the grammar $G$ to have $\epsilon$-rules of
the form  $A \stackrel{p}{\rightarrow} \epsilon$.  
As we have discussed, 
allowing such rules makes this problem substantially more difficult.

Our first main aim is to prove the following theorem:

\begin{thm}\label{main-scfg-parsing-thm}(Approximation of the total parse
probability of a string on an SCFG)
There is a polynomial-time algorithm for 
approximating the total probability that a given string $w$
is generated by a given arbitrary SCFG, $G$, including an SCFG that contains
arbitrary $\epsilon$-rules.

More precisely, there is a polynomial-time algorithm
that, given as input any (proper) SCFG, $G$,
with rational rule probabilities
and with a terminal alphabet $\Sigma$, 
given any string $w \in \Sigma^*$,
and given any 
rational value $\delta >0$ in standard binary representation,
computes a rational value $v \in [0,1]$ such that $| v - p_{G,w} | < \delta$.
\end{thm}

Crucial for establishing these results is the 
following normalization theorem, which is of more general applicability.
It says that any SCFG 
can be converted in P-time to a suitable  ``approximate'' SCFG
which is in CNF form.   Let us give a precise definition of a notion
of approximate SCFG.

\begin{defn}
For any SCFG $G = (V,\Sigma,R,S)$,  and any  $\delta > 0$, 
we define a set of SCFGs, denoted,  $B_{\delta}(G)$,
called the {\em $\delta$-ball around $G$}, as following.
$B_{\delta}(G)$ consists of all SCFGs, $G' = (V,\Sigma,R',S)$,
such that $G'$ has exactly
the same nonterminals $V$, terminal alphabet $\Sigma$, 
start nonterminal $S$, as $G$,
and furthermore such that the rules in $R'$
of $G'$ that have non-zero probability are exactly the same as the rules $R$ of $G$
that have non-zero probability, and furthermore for every
rule $r \in R$ of $G$, the corresponding rule $r' \in R'$ of $G'$, 
we have $|p(r) - p(r')| \leq \delta$.\\
For any $G' \in B_{\delta}(G)$, we say that $G'$  $\delta$-approximates $G$.
\end{defn}

\begin{thm}\label{main-scfg-cnf-form-thm}(Approximation of an SCFG by an SCFG in CNF form)
There is a polynomial-time algorithm 
that, given as input any (proper) SCFG, $G$,
with rational rule probabilities, given any 
natural number $N$ {\em represented in unary},
and given any 
rational value $\delta >0$ in standard binary representation,
computes a new SCFG, $G'$, such that $G'$ is in Chomsky Normal Form,
and has rational rule probabilities,
and such that $G' \in B_\delta(G'')$, where
$G''$ is an SCFG in Chomsky Normal Form, which possibly has {\em irrational}
rule probabilities, but such that for all string $w \in \Sigma^*$ 
we have $p_{G,w} = p_{G'',w}$.
Furthermore, the $\delta$-approximation $G'$ of $G''$ is such
that for all strings $w \in \Sigma^*$, such that $|w| \leq N$, we have:
$ |p_{G,w} - p_{G',w} | \leq \delta $.
\end{thm}

In other words, the P-time computed CNF form SCFG, $G'$, is a ``good
enough approximation'' of $G$ when it comes
to the total parse probability of all strings up to the desired length,
$N$, and $G'$ also $\delta$-approximates  a CNF form SCFG, $G''$, with irrational
rule probabilities, such that $G$ and $G''$ generate exactly the same probability
distribution on finite strings.
We emphasize however that the length $N$ needs to be given in
unary for this algorithm to run in P-time.

Let us note here again that Abney, McAllister, and Pereira \cite{AMP99}
(Theorem 4),
have established the {\em existence} of a SCFG in CNF
form that has exactly the same probability of generating any nonempty string as the
original SCFG.   However, as they mention, 
their existence result is completely non-constructive, and yields
no algorithm for computing or approximating such an SCFG.
Of course, we note again that any such SCFG may require irrational
rule probabilities.

We shall thus show that the
non-constructive result of \cite{AMP99} can be made entirely constructive,
and that in fact an approximate version can be carried out in P-time.
Specifically, we show that an SCFG, $G$, can 
be put through
a sequence of ``{constructive  transformations}'', some of which we don't
actually compute explicitly, because they involve irrational rule probabilities,
which ultimately leads, firstly, to an SCFG (with irrational rule
probabilities) which
is in CNF form, and which has exactly the same probability of generating
any string,  and secondly, thereafter to an ``approximate SCFG'' which
has approximately the same probability of generating any string up
to a desired length $N$, and which can be computed in P-time from
the original SCFG.

To begin our series of SCFG transformations,
let us first observe that an obvious adaptation
of the methods we used to prove Proposition \ref{prop:snf-form}, which
showed that we can
 convert any MPS or PPS into one which is in {\em simple normal
form} (SNF), can be used to also convert
any SCFG $G$ to one that is also in a 
{\em simple normal form} (SNF).
By definition, a SCFG is SNF form
if it only contains
the following four kinds of rules:

\begin{enumerate}

\item  $A \stackrel{p}{\rightarrow} B C$, where $B$, $C$ are
nonterminals in $V$.

\item  $A \stackrel{p}{\rightarrow} B$, where $B$ is a nonterminal
symbol.

\item  $A  \stackrel{p}{\rightarrow}  a$, where 
where $a \in \Sigma$ is a terminal symbol.

\item  $A  \stackrel{p}{\rightarrow} \epsilon$,
$\epsilon$ denotes the empty string.
\end{enumerate}

\begin{lem}
\label{grammar-sd2nf-lem}
Any SCFG, $G$, with rational rule probabilities, 
can be converted in P-time to a SCFG, $G^{(1)}$, in SNF
form such that $G^{(1)}$ has the same terminal symbols $\Sigma$ 
as $G$, such that $G^{(1)}$ has rational rule probabilities,
and such that $G^{(1)}$ generates exactly the same probability
distribution on finite strings in $\Sigma^*$,
i.e., such that $p_{G^{(1)},w} = p_{G,w}$ for all strings 
$w \in \Sigma^*$,
and (thus) also $G$ and $G^{(1)}$ have 
the same probability of termination.

Furthermore, if the original grammar $G$ had no $\epsilon$-rules
then the new SNF grammar $G^{(1)}$ 
will also have no $\epsilon$-rules.

Likewise, if the grammar $G$ only has a single $\epsilon$-rule 
$S \stackrel{p}{\rightarrow} \epsilon$, where $S$ is the start 
nonterminal, and where $S$ doesn't appear on the
RHS of any rule in $G$, then the same would hold for $G^{(1)}$. 
\end{lem}

The proof is analogous to  
the proof of Proposition \ref{prop:snf-form}.  It 
involves adding new auxiliary nonterminals  and 
new auxiliary rules, each  having probability 1,  
in order to suitably ``abbreviate'' the sequences of symbols $\gamma$
on right hand side (RHS) of rules 
$A \stackrel{p}{\rightarrow} \gamma$, whenever $|\gamma| \geq 3$.
We do this repeatedly until all such RHSs, $\gamma$, have $|\gamma| \leq 2$.
To obtain the normal form,
we may then also need to introduce nonterminals that generate 
a single terminal symbol  with probability 1.
We leave the rest of the proof as an easy exercise for the reader.

Clearly, every SCFG in Chomsky Normal Form form is in SNF form (but clearly not the
other way around).
We shall show that an SNF normal form SCFG can be ``transformed''
to CNF form, albeit to a CNF SCFG which may possibly
have {\em irrational} rule probabilities, and which we will
not actually compute.
If however the SNF SCFG happens to contain no $\epsilon$ rules,
then we shall see that the resulting CNF SCFG has only
rational rule probabilities, and can be computed exactly in P-time.

We will
use $E_G(A) := p^A_{G,\epsilon}$ to denote the probability
that starting with the nonterminal $A$,
the grammar $G$ generates the empty string
$\epsilon$, and we will use $NE_G(A) = 1-E_G(A)$ to denote
the probability that nonterminal $A$ does {\em not} generate
the empty string.  When the grammar $G$ itself is clear
from the context, we will use $E(A)$ to denote $E_G(A)$
and $NE(A)$ to denote $NE_G(A)$.

\begin{lem} \label{e-and-ne-in-ptime-lem}
\mbox{}
\begin{enumerate}
\item There is a P-time algorithm that,
given a SCFG, $G$, for each nonterminal $A$ of $G$
determines whether $E(A) = 1$, i.e., $NE(A)=0$,
and likewise whether $E(A)=0$, i.e.,
$NE(A) = 1$.

\item There is a P-time algorithm
that, given an SCFG $G$,
any nonterminal $A$ of $G$, and given any rational $\delta' > 0$,  
computes a $\delta'$-approximation of $E(A)$ (and of $NE(A)$), 
i.e., computes a
rational value $v^A_E \in [0,1]$ 
such that
$|E(A) - v^A_{E} | \leq \delta'$.
And  
thus, letting $v^A_{NE} := 1-v^A_E$,  we
also compute $v^A_{NE}$ such that $| NE(A) - v^A_{NE} | 
\leq \delta$. 
\end{enumerate} 
\end{lem}

\begin{proof}
Part (1.) of Lemma \ref{e-and-ne-in-ptime-lem}
follows directly from the fact (\cite{rmc}) that we can 
decide whether the termination probability of a SCFG is 
$=1$, or is $=0$, in P-time. 

Part (2.) of Lemma \ref{e-and-ne-in-ptime-lem} follows
directly from our main result, Corollary 
\ref{main-alg-corollary},
which says that we can $\delta'$-approximate the termination
probability of a SCFG in P-time.

To see why these hold, simply note that $E(A)$ 
is precisely the termination probability, starting at nonterminal
$A$, of a new SCFG obtained from $G$ by removing
all rules that have any terminal
symbol $a \in \Sigma$ occurring on the 
RHS.\footnote{More precisely, we remove each such rule 
$B \stackrel{p}{\rightarrow} \gamma$, and in order to still
maintain a {\em proper} SCFG, we add a new ``dead'' rule
$B \stackrel{p}{\rightarrow} NN$,  where 
$N$ is a new ``dead'' nonterminal
symbol that has associated with it the rule $N \stackrel{1}{\rightarrow} NN$.}
\end{proof}

Consider a SCFG, $G^{(1)}$, in SNF form.
As the next step of our transformation, 
we shall obtain a new SNF SCFG, $G^{(2)}$, where we
remove all nonterminals $A$ from the SCFG, $G^{(1)}$,
such that $E(A) = 1$.   We do so as follows:  
first, compute in P-time whether $E(A) = 1$ for every
nonterminal $A$.
If $E(A) =1$, then remove all rules associated with $A$, 
i.e., all rules in $R_A$, and
furthermore remove
every occurrence
of $A$ from the RHS of any rule. In other words,
if $\gamma$ is the right hand side of some rule and $A$
occurs in $\gamma$, then remove those occurrences of $A$,
and leave the remaining symbols in their original order.
If this results in an empty
RHS of a rule, then the RHS becomes the empty string $\epsilon$.

In the special case where
$S$ is the start nonterminal of $G^{(1)}$  and $E(S) = 1$,
the SCFG $G^{(1)}$ generates the empty string with
probability 1, and in this case we make $G^{(2)}$ the {\em trivial
SCFG} consisting of only one rule: 
$S \stackrel{1}{\rightarrow} \epsilon$.

\begin{defn}
We call a SCFG, $G$, in SNF form \underline{\em cleaned}
if it contains no nonterminals
$A$ such that $E(A) = 1$, unless  $E(S) =1$ where
$S$ is the start nonterminal of $G$,
in which case $G$
is the \underline{\em trivial} SCFG consisting of a single rule  given by
$S \stackrel{1}{\rightarrow} \epsilon$.
\end{defn}

\noindent The above discussion establishes
the following Lemma.

\begin{lem} \label{cleaning-2-scfg-lemma}
{(Cleaned SCFG: removal of trivial nonterminals)}
Given an input SCFG, $G^{(1)}$ in SNF form,
we can
compute in P-time  a \underline{\em cleaned} SCFG in SNF form,  $G^{(2)}$,
such that for all strings $w \in \Sigma^*$ we have
\[  p_{G^{(1)},w} =  p_{G^{(2)},w} \]
\end{lem}

We are now ready for a critical step in our ``transformation''
which involves irrational probabilities.  We
will not actually compute this ``transformation'' exactly in our
algorithms, but rather we will later do so ``approximately'' 
in an appropriate
way.

\begin{lem}{(Conditioned SCFG: removal of epsilon rules)}
\label{lem-cnf-construct-step1}
Given any \underline{cleaned} non-trivial SCFG, $G^{(2)}$, 
in SNF form,
there is an SCFG, $G^{(3)}$ which has the
same terminals and nonterminals as $G^{(2)}$ and which is also in 
SNF form, but which does not contain any $\epsilon$ rules, and such 
that for all non-empty strings $w \in \Sigma^+$ and all nonterminals $A$ we have:

\[p^A_{G^{(2)},w} =  p^A_{G^{(3)},w} * NE_{G^{(2)}}(A) \]
The SCFG $G^{(3)}$ may contain rules with irrational
probabilities, even if $G^{(2)}$ does not.\footnote{In fact,
our proof establishes a more precise relationship between the parse trees
of $G^{(2)}$ and  $G^{(3)}$ and their 
respective probabilities, but since we
will not later use this stronger fact, we refrain from describing the 
precise relationship within the statement of the Lemma.}
\end{lem}

According to this Lemma,
for any cleaned non-trivial SCFG, $G^{(2)}$, and any non-empty string $w$,
\[p^A_{G^{(3)},w} = p^A_{G^{(2)},w}/NE_{G^{(2)}}(A)\]
In other words, the probability 
of generating the non-empty string $w$ 
in $G^{(3)}$ starting at nonterminal $A$,
is precisely the
{\em conditional probability}  that $G^{(2)}$ generates 
the string $w$ starting at nonterminal 
$A$, conditioned on the event that $G^{(2)}$ does
not generate the empty string starting at $A$.
This is why we call $G^{(3)}$ the {\em ``conditioned SCFG''} 
for $G^{(2)}$.\footnote{Let
us mention here that our proof of this Lemma is related in spirit to, 
but is quite different from, our 
proof in \cite{rmc-model-checking} of a {\em conditioned summary chain}
construction for Recursive Markov Chains.}

\begin{proof}
Given $G^{(2)} = (V,\Sigma,R^{(2)},S)$,
we define the new SCFG, $G^{(3)} = (V,\Sigma,R^{(3)},S)$,
as follows. Below, whenever we refer to E(A) or NE(A)
for any nonterminal $A$, these are with respect to the SCFG
$G^{(2)}$, i.e., $E(A) := E_{G^{(2)}}(A)$ and $NE(A) := NE_{G^{(2)}}(A)$:

\begin{enumerate} 
\item  For each rule $r$ of the form $A \stackrel{p}{\rightarrow} a$ 
in $R^{(2)}$,
where $a \in \Sigma$ is a
single terminal symbol,  we put into $R^{(3)}$ the following rule: 

{\Large
\vspace*{-0.1in}
\begin{equation*} \begin{CD}
r': \ A  @>\frac{p}{NE(A)}>> a
\end{CD}
\end{equation*}}

\item  For each rule $r$ of the form $A \stackrel{p}{\rightarrow} B$ 
in $R^{(2)}$,
where $B \in V$ is a
single nonterminal symbol,  we put into $R^{(3)}$ the following rule:

{\Large
\begin{equation*} \begin{CD}
r': \ A  @>\frac{p * NE(B)}{NE(A)}>> B
\end{CD}
\end{equation*}
}

\item For each rule $r$ of the form $A \stackrel{p}{\rightarrow} B C$
in $R^{(2)}$,  where $B, C \in V$ are nonterminals, we put all
of the following three rules into $R^{(3)}$:

{\Large
\begin{equation*} 
\begin{CD}
r'(1): \ A  @>\frac{p * NE(B) * NE(C)}{NE(A)}>> B C
\end{CD}
\end{equation*}
}

{\Large
\begin{equation*} 
\begin{CD}
r'(2):  \ A  @>\frac{p * NE(B) * E(C)}{NE(A)}>> B 
\end{CD}
\end{equation*}}

{\Large
\begin{equation*}
 \begin{CD}
r'(3): \ A  @>\frac{p * E(B) * NE(C)}{NE(A)}>> C
\end{CD}
\end{equation*}}

We do not put any other rules into $R^{(3)}$.  This 
completes the definition of $G^{(3)}$.

Notice that it is possible that the rule probability for
some of these rules will be $0$, because $E(B)$ and
$E(C)$ can be $0$.  In such a case, those rules have probability
$0$, meaning we can simply remove them from $R^{(3)}$.
Notice also that the rule probabilities for rules in $R^{(3)}$ are
all well-defined, because $G^{(2)}$ is a cleaned SCFG, and
thus $NE(A) >  0$ for all nonterminals $A$.

\end{enumerate}

\begin{claim}
If $G^{(2)}$ is a proper SCFG, then 
so is $G^{(3)}$.  
\end{claim}
\begin{proof}
 To see why this claim holds, observe that for
every nonterminal $A$, 
\[NE(A) = \sum_{r \in R_A} 
p(r) * (1-p^{RHS(r)}_{G^{(2)},\epsilon}) \]
where the sum is over all rules $r$ associated with nonterminal $A$
in $G^{(2)}$.
In other words, the probability that $A$ does not generate the
empty string in $G^{(2)}$
is equal
to the weighted sum of the probabilities that the RHSs $\gamma$ of
rules associated with $A$ do not generate the empty string.
But then note that:

\begin{enumerate} 
\item For a rule of the form $A \stackrel{p}{\rightarrow} a$,
the probability that the RHS $a$ doesn't generate the empty string is 1.

\item For a rule of the form $A \stackrel{p}{\rightarrow} B$,
the probability that the RHS $B$ does not generate the empty
string is $NE(B)$.

\item For a rule of the form $A \stackrel{p}{\rightarrow} B C$,
the probability that the RHS $BC$ does not generate the
empty string is:  $NE(B) * NE(C) + E(B) * NE(C) + NE(B) * E(C)$.
This is because we need at least one of $B$ or $C$ to not
generate the empty string, and whether each of them does so
or not is an independent event.
\end{enumerate}

From this we see that if we sum the probabilities
of the rules associated with $A$ in $G^{(3)}$, 
assuming that $G^{(2)}$ is proper, the
numerators of these sums will sum up to $NE(A)$ and thus since all of them
have denominator $NE(A)$, the SCFG $G^{(3)}$ is also proper.
\end{proof}

\noindent We next have the key claim:

\begin{claim}\label{condprob-claim}
For any nonterminal $A$, and 
for all non-empty strings $w \in \Sigma^+$, we have:

\[p^A_{G^{(2)},w} =  p^A_{G^{(3)},w} * NE_{G^{(2)}}(A) \]

\end{claim}
\begin{proof}

We shall prove this key claim as follows.
For every non-empty string $w \in \Sigma^+$, and
every nonterminal $A$
we will define a mapping, $g_{w,A}$, from 
finite parse trees $t \in T^A_{G^{(2)},w}$ 
to finite parse trees $g_{w,A}(t) \in T^A_{G^{(3)},w}$.
We shall establish that the mapping $g_{w,A}$ has the following properties:

\begin{itemize}
 \item[(1.)] The mapping $g_{w,A}$  is well-defined, meaning that 
if $T^A_{G^{(2)},w} \neq \emptyset$, then for  
any parse tree $t \in T^A_{G^{(2)},w}$,
we have $g_{w,A}(t) \in T^A_{G^{(3)},w}$.

\item[(2.)] The mapping $g_{w,A}$ is onto,  meaning if 
$T^A_{G^{(3)},w} \neq \emptyset$,
then for any tree $t' \in T^A_{G^{(3)},w}$ 
we have  
$g_{w,A}^{-1}(t') \neq \emptyset$.

\item[(3.)]  Finally, the following equality 
holds for all parse trees $t' \in T^A_{G^{(3)},w}$:

\begin{equation}\label{map-of-parse-prob-eq}
 P_{G^{(3)},A}(t') * NE_{G^{(2)}}(A) 
= \sum_{t \in g_{w,A}^{-1}(t')}  P_{G^{(2)},A}(t)    
\end{equation}

In other words, the probability of parse tree $t'$ of $w$ rooted at $A$ in 
$G^{(3)}$
times the probability that 
the nonterminal $A$ does not generate the empty string $\epsilon$ in
$G^{(2)}$, 
is the same as the sum of the probabilities of all parse trees $t$ 
of $w$ in $G^{(2)}$ rooted at $A$
that are mapped to $t'$ by $g_{w,A}$. 
\end{itemize}

\noindent Once we establish the above three properties for the mapping
$g_{w,A}$, which we shall define shortly, the Claim \ref{condprob-claim} 
follows basically immediately,
because:

\begin{eqnarray*}
p^A_{G^{(2)},w} & = &  \sum_{t \in T^A_{G^{(2)},w}}  
P_{G^{(2)},A}(t)   \hspace{0.2in}  \mbox{(by equation (\ref{sum-of-parse-probs-eq}))}\\
& = & \sum_{t' \in T^A_{G^{(3)},w}}  \sum_{t \in g^{-1}_{w,A}(t')} 
P_{G^{(2)},A}(t) \\
& = &  \sum_{t' \in T^A_{G^{(3)},w}}
P_{G^{(3)},A}(t')  * NE_{G^{(2)}}(A) \hspace{0.2in}
\mbox{(by equation (\ref{map-of-parse-prob-eq}))}\\
& = & (\sum_{t' \in T^A_{G^{(3)},w}}
P_{G^{(3)},A}(t') )  * NE_{G^{(2)}}(A)\\
& = & p^A_{G^{(3)},w} * NE_{G^{(2)}}(A)  \hspace{0.2in}
\mbox{(by equation (\ref{sum-of-parse-probs-eq}))}
\end{eqnarray*}

It now remains to define $g_{w,A}$, and then to establish properties (1.)-(3.).

Given a parse tree $t \in T^A_{G^{(2)},w}$,  we define 
$g_{w,A}(t)$ via a simple kind of "pruning" of $t$, as follows.
Let us call a subtree $t^*$ of $t$ a {\em $\epsilon$-maximal subtree}
if, firstly, all
leaves of $t^*$ are labeled by $\epsilon$, 
and secondly, 
either $t^*$ is $t$, or else it is not the case that
all leaves of the subtree 
rooted at the immediate parent of the root of $t^*$, within $t$,
are also $\epsilon$.
So, $\epsilon$-maximal subtrees are maximal sub-parse-trees of $t$ that
generate the empty string.

We shall define $g_{w,A}(t)$ to be the ``pruning'' of $t$ obtained by 
removing all $\epsilon$-maximal subtrees of $t$.
We do not replace the removed subtrees by anything. 
To be precise, when we remove one of the ordered children
of an internal node of $t$, and the subtree rooted
at that child, we retain the relative ordering of the other
children with respect to each other.

Firstly, note that $g_{w,A}(t)$ will indeed retain the root node
of $t$, labeled $A$.  This is because $t$ is a parse tree of the
non-empty string $w$, and thus it can not be the case that
all leaves of $t$ are labeled by $\epsilon$.   

Our definition of $g_{w,A}(t)$ is not yet complete.
In more detail, we have to define all labels, including rule
labels, of all nodes in $g_{w,A}(t)$.  We do so as follows.

To do so, first note that 
for every node $z$ of $t$ that is retained in $g_{w,A}(t)$,
note that $z$ is a leaf in $g_{w,A}(t)$ if and only if it was already
a leaf in $t$.  

For every leaf node $z$ of $g_{w,A}(t)$ we retain exactly the
same label, $L(z) = a \in \Sigma$,  which was its label in $t$.
 
For every internal node $z$ of $g_{w,A}(t)$, 
we retain exactly the same nonterminal label, $L_1(z) = B$, 
labeling the corresponding node $z$ in $t$.
Furthermore, if in $t$ the node $z$ was labeled by a rule $r= L_2(z)$,
then we do as follows:

\begin{enumerate}
\item If the rule $r$ is of the form $B \stackrel{p}{\rightarrow} a$,
for some terminal symbol $a \in \Sigma$,
then in $g_{w,A}(t)$ we let $L_2(z)$ be the corresponding rule $r'$ given by 
$B \stackrel{p/NE(B)}{\rightarrow} a$.

\item If the rule $r$ is of the form $A' \stackrel{p}{\rightarrow} B'$,
for some nonterminal symbol $B'$,
then in $g_{w,A}(t)$ we let $L_2(z)$ be the corresponding rule $r'$ 
specified by 
$B \stackrel{p*NE(C)/NE(B)}{\rightarrow} C$.

\item If the rule $r$ is of the form $B \stackrel{p}{\rightarrow} C D$
for nonterminal symbols $C$ and $D$, then 
in $g_{w,A}(t)$ we shall assign $L_2(z)$
one of the three corresponding rules $r'(1)$, $r'(2)$, or $r'(3)$,
based on the following:

If in $t$ neither child of $z$ was an $\epsilon$-maximal subtree,
then in $g_{w,A}(t)$ we let $L_2(z) := r'(1)$.
If in $t$ the right child of $z$ was an $\epsilon$-maximal subtree,
then in $g_{w,A}(t)$ we let $L_2(z) := r'(2)$.
If in $t$ the left child of $z$ was an $\epsilon$-maximal subtree,
then in $g_{w,A}(t)$ we let $L_2(z) := r'(3)$.

\end{enumerate}

The reader can easily confirm that these are the only possibilities,
since otherwise the node $z$ would have been ``pruned out'', by
the definition of the tree defining $g_{w,A}(t)$.
So this mapping of rules to nodes of $g_{w,A}(t)$ is well-defined,
i.e., $g_{w,A}(t) \in T^A_{G^{(3)},w}$.
Indeed, consider the parent node, $z'$, of the root of an $\epsilon$-maximal
subtree, $t^*$, of $t$, and suppose that node $z'$ is labeled by
a nonterminal
$A'$.  Note that the rule $r$ associated
with the parent node, $z'$, in $t$ must be of the form
$A' \stackrel{p}{\rightarrow} B' C'$, where $B'$ and $C'$
are non-terminals, and one of them, say $B'$ w.l.o.g., 
is the label of the root of the
$\epsilon$-maximal subtree $t^*$.  
This is because if the rule associated with $z'$ was a {\em linear} 
rule of the form $A' \stackrel{p}{\rightarrow} B'$, then $t^*$ would
not be an $\epsilon$-maximal subtree in $t$.
Thus if $t^*$ is rooted at the left child, labeled $B'$,
of node $z'$, by construction, the rules $R^{(3)}$ of $G^{(3)}$
will include the rule $r'(3)$ given by:
{\large
\begin{equation*} 
\begin{CD}
A'  @>\frac{p * E(B') * NE(C')}{NE(A')}>> C'
\end{CD}
\end{equation*}}

\noindent and we have defined the parse tree $g_{w,A}(t)$ so that it 
uses this rule
of $G^{(3)}$ at the node $z'$.
In other words, we have let $L_2(z') =r'(3)$.
Similarly, if $t^*$ is rooted at the right child of $z'$
labeled by $C'$, we have made $g_{w,A}(t)$ use the rule 
{\large
\begin{equation*} 
\begin{CD}
A'  @>\frac{p * NE(B') * E(C')}{NE(A')}>> B'
\end{CD}
\end{equation*}}
which exists by definition of $R^{(3)}$.
Thus $g_{w,A}(t)$, as defined, is a parse tree of $G^{(3)}$ and
is clearly a parse tree of the string $w$.
We have thus established (1.), i.e., 
that indeed $g_{w,A}(t) \in T^A_{G^{(3)},w}$.

Next we establish $(2.)$, that the mapping $g_{w,A}$ is onto.  
Suppose $T^A_{G^{(3)},w} \neq \emptyset$, and
suppose that $t' \in T^A_{G^{(3)},w}$.
If $t'$ does not contain any internal node labeled with
a rule of the types $r'(2)$ or $r'(3)$, then
it is easy to see that exactly the same parse tree, where
we replace every rule label $r'$ or $r'(1)$ by their corresponding version
$r$ in $G^{(2)}$, is indeed a parse tree $t \in T^A_{G^{(2)},w}$
such that $g_{w,A}(t) = t'$.

If on the other hand $t'$ does have some internal node labeled
with a rule of the type $r'(2)$ or $r'(3)$, then without
loss of generality (by symmetric arguments) 
suppose it is a rule of the form 
$r'(2)$ given by $A' \stackrel{p*NE(B')*E(C')/NE(A')}{\rightarrow} B'$.
Note that  
it must be the case that $E(C') > 0$ (otherwise $t'$ has
probability $0$ and thus is not a parse tree in $G^{(3)}$), and thus
there is some parse tree rooted at $C'$ which generates the empty
string $\epsilon$.

Thus, for all such rules of the form $r'(2)$ 
labeling a node $z$ of $t'$, we will be able to convert the rule
at the corresponding node $z$ of a parse tree $t$ of $G^{(2)}$ 
to the original rule $r$ of the form $A' \stackrel{p}{\rightarrow} B' C'$
from which $r'(2)$ was generated, 
and then we can add {\em any} sub-parse tree for the empty string $\epsilon$,
rooted at the child of $z$ in $t$ labeled by nonterminal $C'$.
We can also obviously do the symmetric thing for nodes labeled by
rules $r'(3)$ in $t'$.
In this way, we will have constructed a tree $t \in T^A_{G^{(2)},w}$ such
that $g_{w,A}(t) = t'$
This establishes property $(2.)$, namely that $g_{w,A}$ is onto.

Finally, we have to establish the key property $(3.)$,
namely that for every parse tree $t' \in T^{A}_{G^{(3)},w}$,
we have:
\[ P_{G^{(3)},A}(t') * NE_{G^{(2)}}(A) = \sum_{t \in g^{-1}_{w,A}(t')} P_{g^{(2)},A}(t) \]

The key to establishing this equality is the following inductive claim.
Let us define a mapping $h$ from rules of $G^{(3)}$ back to their
``corresponding'' rule in $G^{(2)}$.  Specifically, for every rule
$r'$ of $G^{(3)}$ we see easily that by our definition of $R^{(3)}$ 
this rule was generated directly from a ``corresponding'' rule $r$ 
in $R^{(2)}$. We simply define $h(r') := r$.

We extend this mapping $h$ to a tree $t' \in T^{A}_{G^{(3)},w}$,
by defining $h(t')$ to be the {\em multi-set} of rules in $R^{(2)}$
that arise by mapping back the rule label $L_2(z)$ of every internal
node in $t'$ to its corresponding rule $h(L_2(z))$ in $R^{(2)}$.
It is important that $h(t')$ is a multi-set, i.e., that it retains
$k$ copies of the same rule $r$ if there are $k$ nodes $z$ of $t'$ 
for which $h(L_2(z)) = r$. 

We need some more definitions.  For a tree $t' \in T^{A}_{G^{(3)},w}$,
let us define two other multi-sets of rules in $G^{(2)}$,  namely,
$Z_{t',2}$ and $Z_{t',3}$,  where $Z_{t',2}$ is a multi-set of rules
in $R^{(2)}$
containing one copy of a rule $r \in R^{(2)}$ for every 
instance of the corresponding rule $r'(2)$
that labels some node $z$ of $t'$.
Similarly, $Z_{t',3}$ is a multi-set containing one copy of
a rule $r \in R^{(2)}$ for every instance of the corresponding
rule $r'(3)$ that labels some node $z$ of $t'$.
Notice that all rules in the multi-sets $Z_{t',2}$ and $Z_{t',3}$ are
of the form $A' \stackrel{p}{\rightarrow} B' C'$.

Let us define the following multi-sets corresponding to $Z_{t',2}$
and $Z_{t',3}$.  Namely, let $K_{t',2}$ be the multi-set of
nonterminals in $G^{(2)}$ defined by taking every rule instance 
$r \in Z_{t',2}$ and if $r$ has the form  $A' \stackrel{p}{\rightarrow} B' C'$,
then adding a copy of $C'$ to $K_{t',2}$.
Likewise let $K_{t',3}$ be the multi-set of
nonterminals in $G^{(2)}$ defined by taking every rule instance 
$r \in Z_{t',3}$ and if $r$ has the form  $A' \stackrel{p}{\rightarrow} B' C'$,
then adding a copy of $B'$ to $K_{t',3}$.

We are now ready to state and prove a key claim.

\begin{claim}\label{3-to-2-claim}
For every parse tree $t' \in T^A_{G^{(3)},w}$, we have

{\large
\begin{equation}\label{parse-in-3-claim-eq}
 P_{G^{(3)},w}(t') =   
\frac{(\prod_{r \in h(t')} p(r)) * (\prod_{C' \in K_{t',2}} E(C')) * 
(\prod_{B' \in K_{t',3}} E(B'))}{ NE(A)}
\end{equation} }

Note that the products are indexed over multi-sets, not sets.
\end{claim}
\begin{proof}
We prove this claim by induction on the depth of
the parse tree $t'$.

For the base case, if the parse tree $t'$ has depth 1, then
it has only one internal node which is the root, and 
that root is labeled by a rule $r'$
of the form:

\vspace*{-0.1in}

{\large
\begin{equation*} \begin{CD}
\ A  @>\frac{p(r)}{NE(A)}>> a
\end{CD}
\end{equation*}}

\vspace*{0.05in}

\noindent Thus $t'$ is a parse tree of the string $w= a$, and
$P_{G^{(3)},w}(t') =  p/NE(A)$.
But since $K_{t',2} = K_{t',3} = \emptyset$, we see that 
the right hand side of equation (\ref{parse-in-3-claim-eq})
is also equal to $p(r)/NE(A)$.

Inductively, suppose that $t'$ has depth $\geq 2$.  
There are different cases to consider, based on the rule labeling
the root of $t'$. 

\begin{enumerate}

\item Suppose that the root $z$ of $t'$ is labeled by a rule
$L_2(z) = r'$ which has the form:

\vspace*{-0.1in}

{\large
\begin{equation*} \begin{CD}
\ A  @>\frac{p(r) * NE(B)}{NE(A)}>> B
\end{CD}
\end{equation*}}

\vspace*{0.05in}

\noindent and that 
$h(r') = r \in R^{(2)}$, where rule $r$
has the form $A \stackrel{p(r)}{\longrightarrow} B$.

Thus the root $z$ of $t'$ has only one child node in $t'$, call it $z^*$.
Let $t^* \in T^{B}_{G^{(3)},w}$ denote the parse subtree of $t'$ 
rooted at $z^*$.
 We know that $L_1(z') = B$, and by
inductive assumption we know that 
\[ P_{G^{(3)},w}(t^*) =   
\frac{(\prod_{r \in h(t^*)} p(r)) * (\prod_{C' \in K_{t^*,2}} E(C')) * 
(\prod_{B' \in K_{t^*,3}} E(B'))}{ NE(B)} \]

But note that $P_{G^{(3)},w}(t') =  P_{G^{(3)},w}(t^*) * p(r')$, 
and by multiplying 
and canceling, we get 
$P_{G^{(3)},w}(t') =  P_{G^{(3)},w}(t^*) * (p(r)*NE(B)/NE(A)) =
\frac{(\prod_{r \in h(t')} p(r)) * (\prod_{C' \in K_{t',2}} E(C')) * 
(\prod_{B' \in K_{t',3}} E(B'))}{ NE(A)}$. 

That completes the induction in this case.

\item Suppose that the root $z$ of $t'$ is labeled by a rule
$L_2(z) = r'(1)$ which has the form:

\vspace*{-0.1in}

{\large
\begin{equation*} 
\begin{CD}
A  @>\frac{p(r) * NE(B_1) * NE(B_2)}{NE(A)}>> B_1 B_2
\end{CD}
\end{equation*}}

\vspace*{-0.05in}

\noindent and that $h(r'(1))= r \in R^{(2)}$, such that rule $r$
has the form $A \stackrel{p(r)}{\rightarrow} B_1 B_2$.

In this case, the root $z$ of $t'$ has two children,
a left child $z_1$ and a right child $z_2$.   
Let $t_1 \in T^{B_1}_{G^{(3)},w_1}$ and $t_2 \in T^{B_2}_{G^{(3)},w_2}$
be the two parse trees rooted at $z_1$ and $z_2$ respectively.
Clearly we must have $w = w_1 w_2$.

We know that $L_1(z_1) = B_1$ and $L_1(z_2) = B_2$.
Moreover, by inductive assumption, we know that for $i=1,2$, we have 
\[ P_{G^{(3)},w_i}(t_i) =   
\frac{(\prod_{r \in h(t_i)} p(r)) * (\prod_{C' \in K_{t_i,2}} E(C')) * 
(\prod_{B' \in K_{t_i,3}} E(B'))}{ NE(B_i)} \]

Note again that $P_{G^{(3)},w}(t') =  P_{G^{(3)},w_1}(t_1) * 
P_{G^{(3)},w_2}(t_2) * p(r')$.
 
Again, by multiplying 
and canceling, we get 
$P_{G^{(3)},w}(t') =  P_{G^{(3)},w_1}(t_1) * P_{G^{(3)},w_2}(t_2) * 
(p(r)*NE(B_1)* NE(B_2)/NE(A)) =
\frac{(\prod_{r \in h(t')} p(r)) * (\prod_{C' \in K_{t',2}} E(C')) * 
(\prod_{B' \in K_{t',3}} E(B'))}{ NE(A)}$. 

This establishes the inductive claim in this case.

\item   Suppose that the root $z$ of $t'$ is labeled by a rule
$L_2(z) = r'(2)$ which has the form:

\vspace*{-0.1in}

{\large
\begin{equation*} 
\begin{CD}
A  @>\frac{p(r) * NE(B_1) * E(B_2)}{NE(A)}>> B_1
\end{CD}
\end{equation*}}

\vspace*{-0.05in}

\noindent and that $h(r'(2))= r \in R^{(2)}$, such that rule $r$
has the form $A \stackrel{p(r)}{\longrightarrow} B_1 B_2$.

In this case, the root $z$ of $t'$ has one child,
$z_1$.
Let $t_1 \in T^{B_1}_{G^{(3)},w}$ 
be the parse tree rooted at $z_1$.
We know that $L_1(z_1) = B_1$.
Moreover,  by inductive assumption, we know that 
\[ P_{G^{(3)},w}(t_1) =   
\frac{(\prod_{r \in h(t_1)} p(r)) * (\prod_{C' \in K_{t_1,2}} E(C')) * 
(\prod_{B' \in K_{t_1,3}} E(B'))}{ NE(B_1)} \]

Note again that $P_{G^{(3)},w}(t') =  P_{G^{(3)},w_1}(t_1) * 
p(r')$.
 
Observe that the multiset $K_{t',2}$ consists of $K_{t_1,2}  \cup \{ B_2 \}$ 
where the union here denotes a {\em multi-set union},
so it contains an added copy of $B_2$.   
Thus, by multiplying 
and canceling, we get 
\begin{eqnarray*}
P_{G^{(3)},w}(t') & = &  P_{G^{(3)},w_1}(t_1) * 
(p(r)*NE(B_1)* E(B_2)/NE(A))\\ & = & 
\frac{(\prod_{r \in h(t')} p(r)) * (\prod_{C' \in K_{t',2}} E(C')) * 
(\prod_{B' \in K_{t',3}} E(B'))}{ NE(A)}
\end{eqnarray*} 

This establishes the inductive claim in this case.

\item Suppose that the root $z$ of $t'$ is labeled by a rule
$L_2(z) = r'(3)$ which has the form:

\vspace*{-0.1in}

{\large
\begin{equation*} 
\begin{CD}
A  @>\frac{p(r) * E(B_1) * NE(B_2)}{NE(A)}>> B_2
\end{CD}
\end{equation*}}

\vspace*{-0.05in}

\noindent and that $h(r'(3))= r \in R^{(2)}$, such that rule $r$
has the form $A \stackrel{p(r)}{\rightarrow} B_1 B_2$.

This case is entirely analogous (and symmetric) to the previous one, 
and thus an
identical  argument shows that the inductive claim
holds also in this case.

\end{enumerate}

This completes the inductive proof of the claim, since we
have considered all possible rule types that can label the
root node of $t'$.

\end{proof}

We now use Claim \ref{3-to-2-claim} to show that property (3.) holds
for the mapping $g_{w,A}$.

Consider a parse tree $t' \in T^{A}_{G^{(3)},w}$.
Claim \ref{3-to-2-claim} tells us that
\[ P_{G^{(3)},w}(t') =   
\frac{(\prod_{r \in h(t')} p(r)) * (\prod_{C' \in K_{t',2}} E(C')) * 
(\prod_{B' \in K_{t',3}} E(B'))}{ NE(A)} \] 

Note that for any nonterminal $B$,
$E(B)$ is the sum of the probabilities of all distinct parse
trees rooted at $B$ which generate the empty string $\epsilon$.
Let $ET(B)$ be the set of all these parse trees that generate $\epsilon$ from $B$.

Let us now consider the probability of parse trees 
in $g^{-1}_{w,A}(t') \subseteq T^{A}_{G^{(2)},w}$.
Note that
each such parse tree $t \in g^{-1}_{w,A}(t')$
can be specified by specifying how $t$ has ``expanded''
every nonterminal in the multisets $K_{t',2}$,
and $K_{t',3}$ into parse trees of the string $\epsilon$.
Specifically, $t$ is determined
by specifying for each occurrence of each nonterminal $B$ 
in both $K_{t',2}$ and $K_{t',3}$, which
parse tree of $ET(B)$ is used to expand $B$ into a parse
tree for $\epsilon$.  
Then the probability of $t$ is given by 
the product of $(\prod_{r \in h(t')} p(r))$ multiplied
by the product of all these chosen parse trees of $\epsilon$
chosen to expand every nonterminal occurrence in $K_{t',2}$ and in 
$K_{t',3}$.

But then since for every nonterminal $B$, $E(B)$ is the
sum of the probabilities of all distinct parse trees rooted
at $B$ which generate $\epsilon$,
we can see that, by summing over all parse trees $t \in g^{-1}_{w,A}(t')$,
and then collecting  like terms, 
we get 
\[\sum_{t \in g^{-1}_{w,A}(t')}  P_{G^{(2)},w}(t) 
 = (\prod_{r \in h(t')} p(r)) * (\prod_{C' \in K_{t',2}} E(C')) * 
(\prod_{B' \in K_{t',3}} E(B'))\]

But then by Claim \ref{3-to-2-claim}, the identity 
(\ref{map-of-parse-prob-eq}) follows,
and thus we have established property (3.) of the mapping $g_{w,A}$,
which is the last thing we needed to establish to complete the proof
of Claim \ref{condprob-claim}.
\end{proof}

This completes the proof of Lemma \ref{lem-cnf-construct-step1},
and establishes the correctness of the quantitative
properties it asserts for the conditioned SCFG, $G^{(3)}$,
which is in SNF form, and which furthermore contains no $\epsilon$-rules.
\end{proof}

Next we show how to ``transform''  $G^{(3)}$ to get rid
of the ``linear'' rules of the form $A \stackrel{p}{\rightarrow} B$,
and thus obtain a CNF form SCFG.

\begin{lem}\label{4th-step-transform-lem}
Given any SCFG, $G^{(3)}$, which is
in SNF form\footnote{We assume, as always, that $G^{(3)}$ is
proper, and it is easy to check that all our transformations
maintain the properness of the SCFG}, 
and which contains no $\epsilon$-rules,
there is an SCFG $G^{(4)}$ in  Chomsky-Normal-Form (CNF),
such that $G^{(4)}$ has the same terminals and nonterminals as $G^{(3)}$
and such that for all nonterminals $A$ and all strings $w \in \Sigma^*$,  
we have $p^{A}_{G^{(3)},w} = p^{A}_{G^{(4)},w}$.

Furthermore, if
$G^{(3)}$ has only rational rule probabilities, then the 
transformation from $G^{(3)}$ to $G^{(4)}$ 
is effective and efficient,
in the sense that $G^{(4)}$ 
also has only rational rule probabilities, and $G^{(4)}$ can
be computed in P-time from $G^{(3)}$.
When $G^{(3)}$ has irrational rule probabilities, then $G^{(4)}$ 
still exists but may require irrational rule probabilities.
\end{lem}
\begin{proof}
The only thing we need to do in order
to obtain $G^{(4)}$ from $G^{(3)}$ is to eliminate ``linear'' rules
of the form $A \stackrel{p}{\rightarrow} B$, where $A$ and $B$
are nonterminals.

This is easy to do by 
solving a suitable system of linear equations
whose coefficients are taken from rule probabilities in the
grammar $G^{(3)}$.

Specifically, consider the following finite-state Markov
chain,  $M$,  whose states consist of the set of nonterminals of
$G^{(3)}$ as well as the set of all distinct ``right-hand sides'', 
$\gamma$, that appear in any rule $A \stackrel{p}{\rightarrow} \gamma$
of $G^{(3)}$, and where $\gamma$ is  not a single nonterminal.

The probabilistic transitions of $M$
consist of the rules of $G^{(3)}$.
In other words, if there is a grammar rule 
$A \stackrel{p}{\rightarrow} \gamma$, then there is 
a probabilistic transition  $(A, p, \gamma)$ in $M$.
Note that we can assume $G^{(3)}$ is a proper SCFG,
so this defines all the transitions out of nonterminal states of $M$.
Finally, for every state $\gamma$ that is not a single nonterminal,
we add an absorbing self-loop transition $(\gamma, 1, \gamma)$ to $M$.

Consider any RHS, $\gamma$, which is not a single nonterminal.
Let $q^*_{A,\gamma}$ denote the probability that, in the finite-state
Markov chain, $M$, starting at state $A$, we eventually hit the state $\gamma$.

Using these hitting probabilities we can easily eliminate the linear
rules of $G^{(3)}$.  We ``construct'' $G^{(4)}$ as follows.
In   $G^{(4)}$ we 
remove all linear rules from $G^{(3)}$, and for every
nonterminal $A$ and RHS $\gamma$, where $\gamma$ is not a single nonterminal,
if $q^*_{A,\gamma} > 0$, then add the rule 
$A \stackrel{q^*_{A,\gamma}}{\longrightarrow} \gamma$ to $G^{(4)}$.

To maintain the properness of the grammar, if the
sum of the probabilities of the rules of a nonterminal $A$ 
is less than 1, we add as usual a rule $A \rightarrow NN$ with
the remaining probability where $N$ is a dead nonterminal.

It is easy to see that the total probability $p^{A}_{G^{(4)},w}$ 
of generating any particular
string $w$ in $G^{(4)}$ starting at nonterminal $A$  
remains the same as the total probability $p^{A}_{G^{(3)},w}$ of
generating $w$ starting at $A$ in $G^{(3)}$.

It only remains to show that if the rule probabilities of $G^{(3)}$
are rational, then we can compute the hitting probabilities $q^*_{A,\gamma}$
in polynomial time.
But it is a well known fact that hitting probabilities can be obtained
by solving a corresponding system of linear equations.

Specifically, consider RHS, $\gamma$, which is not itself a single 
nonterminal.  We can easily determine the set of nonterminals
$A$ for which the hitting probability $q^*_{A,\gamma} > 0$ is positive.
This is the case if and only if in the underlying graph of the
Markov chain $M$ there is a path from the state $A$ to the state $\gamma$.

Suppose there are $n$ distinct nonterminals $A$ in $G^{(3)}$ such that 
$q^*_{A,\gamma} > 0$.
Let us index these $n$ nonterminals as: $A_1,\ldots, A_n$.
Let  $P$ denote the $n \times n$ substochastic matrix whose $(i,j)$'th entry
$P_{i,j}$  is the one-step transition probability
from state $A_i$ to state $A_j$ in the Markov chain $M$.
Let the column $n$-vector $b^{\gamma}$ be defined as follows:
$b^{\gamma}_i$ is the one-step transition probability from state
$A_i$ to state $\gamma$ in $M$.
Then if we let the column $n$-vector $x$ of variables, $x_i$, represent
the unknown hitting probabilities, $q^*_{A_i,\gamma}$,  we have
the following linear system of equations:

\[  x = Px + b^\gamma \]

\noindent which is equivalent to the linear system of equations

\begin{equation}\label{hit-prob-gamma-equation}
 (I-P)x = b^\gamma 
\end{equation}

Clearly, letting $x_i = q^*_{A_i,\gamma}$ is one solution
to this equation.
Moreover,
since $P$ represents the transition submatrix of all the transient
states within a finite-state Markov chain,
it follows from standard facts 
(see, e.g., \cite{BP94}, Lemma 8.3.20)
that $\rho(P) < 1$, where $\rho(P)$ denotes the spectral radius
of the substochastic matrix $P$.
It thus follows from 
Lemma \ref{series} that the matrix $(I-P)$ is non-singular,
that $(I-P)^{-1} = \sum^{\infty}_{k=0} P^k$.
Therefore there is 
a {\em unique} solution vector $x^* = (I-P)^{-1} b^\gamma$ for
the system of linear equations in (\ref{hit-prob-gamma-equation}),
where $x^*_i = q^*_{A_i,\gamma}$
are precisely the hitting probabilities for every 
$i$.

Thus, if $G^{(3)}$ has only rational rule probabilities,
then we can compute $G^{(4)}$ in P-time by solving one such
a system of linear equations for each RHS, $\gamma$, of 
a rule in $G^{(3)}$, which is
not itself a single nonterminal.

In fact, even when $G^{(3)}$ contains irrational rule
probabilities, we will later use the linear system of equations
(\ref{hit-prob-gamma-equation}) in important ways in our
approximability analysis.
\end{proof}

Lemma \ref{4th-step-transform-lem} allows us to finally ``obtain'' a
SCFG $G^{(4)}$ which is in CNF form, starting from our original
SCFG, $G$, via a sequence of transformations. 
Unfortunately, 
in the process of obtaining $G^{(4)}$ some of our transformations
required possibly introducing grammar rules with irrational rule probabilities.
We now show that we can nevertheless efficiently compute
a suitable {\em approximation} to $G^{(4)}$.
The first step toward this is to establish the following Lemma.

For any SCFG, $G$, let 
$G^{(i)}$, $i= 1, \ldots, 4$, denote the SCFG obtained from $G$
via the sequence of transformations
described in Lemmas 
\ref{grammar-sd2nf-lem} to \ref{4th-step-transform-lem}.
In general, $G^{(i)}$ may have irrational rule probabilities,
even when $G$ does not.  Nevertheless, we shall show that,
given $G$ and $\delta > 0$, we can compute in P-time a SCFG 
$G^{(i)}_\delta \in B_\delta(G^(i))$, for $i=1,\ldots,4$.
First:

\begin{lem}\label{g3-approx-lem}
There is a polynomial-time algorithm
that, given any proper SCFG $G$ with rational
rule probabilities, and given any 
rational value $\delta >0$ in standard binary representation, 
computes a proper SCFG,  $G^{(3)}_{\delta} \in B_{\delta}(G^{(3)})$,
with rational rule probabilities.
In other words, given $G$ and $\delta > 0$, 
we can compute in P-time a $\delta$-approximation of $G^{(3)}$.
\end{lem}
\begin{proof}
To prove this theorem, we will show that every
step of our transformations beginning with $G$ and resulting
in the CNF SCFG $G^{(3)}$ 
can be carried out either exactly or
approximately in P-time.

\vspace{0.1in}

\noindent \underline{$G \leadsto G^{(1)} \leadsto G^{(2)}$}:
It was already established 
in Lemma \ref{grammar-sd2nf-lem}
and 
 Lemma \ref{cleaning-2-scfg-lemma}
that we can carry out these first two steps of the
transformation {\em exactly} in P-time.   
Specifically, 
given any SCFG, $G$, with rational rule probabilities,
we can construct in P-time a {\em cleaned} SCFG, $G^{(2)}$,
in SNF form
with rational rule probabilities such that, in particular, for all 
strings $w \in \Sigma^*$, we have $p_{G,w} = p_{G^{(2)},w}$.

Furthermore, we can assume that $G^{(2)}$ is nontrivial,
meaning that the start nonterminal
does not generate the empty string with
probability $1$, because if this was the case, then we know
the transformation would have computed as $G^{(2)}$ 
the {\em trivial} 
CNF SCFG consisting of only the single rule $S \stackrel{1}{\rightarrow} \epsilon$.  In that case we would be done, so we assume w.l.o.g. that the result
was not this trivial SCFG.

\vspace*{0.1in}

\noindent \underline{$G^{(2)} \leadsto G^{(3)}_{\delta}$}:
Recall that the key 
transformation $G^{(2)} \leadsto G^{(3)}$
may introduce irrational rule probabilities into 
the ``conditioned SCFG'', $G^{(3)}$.
Given $G^{(2)}$, and given $\delta > 0$, 
we now show how to compute in P-time a proper
SCFG $G^{(3)}_{\delta} \in B_{\delta}(G^{(3)})$.

To do this, we make crucial use of our P-time approximation algorithm 
for termination probabilities of an SCFG, 
and in particular its corollary, Lemma \ref{e-and-ne-in-ptime-lem},
which tells that that given an SCFG $G^{(2)}$,
for each nonterminal $A$,  we can determine in P-time whether 
$E(A) = 0$ or $NE(A) = 0$,
or whether $E(A) = 1$ or $NE(A) =1$,
and given any $\delta' > 0$,  
we can in P-time approximate $E(A)$ and $NE(A)$ within distance
$\delta'$, i.e., we can compute 
rational values $v^A_E \in [0,1]$ and $v^A_{NE} \in [0,1]$ such that
$|E(A) - v^A_{E} | \leq \delta'$, and $| NE(A) - v_{NE} | \leq \delta'$.
We will see how to choose $\delta'$ shortly.

Note that the probability $p(r') > 0$ of a rule $r'$ in $G^{(3)}$ 
can only have one of several possible forms.
In each case, $p(r')$ is given by
an expression whose denominator is $NE(A) = NE_{G^{(2)}}(A)$ for
some nonterminal $A$ of $G^{(2)}$.

Note that $E(A)$ 
is precisely the termination probability, starting at nonterminal
$A$, of a new SCFG obtained from $G^{(2)}$ by removing
all rules that have any terminal
symbol $a \in \Sigma$ occurring on the 
RHS.  It thus follows that $NE(A)$ is the non-termination probability
starting at $A$ for that SCFG.  
By  Lemma \ref{e-and-ne-in-ptime-lem}, we can determine in P-time
whether $NE(A) = 1$, and by construction of $G^{(2)}$ we know $NE(A) \neq 0$.
Since termination probabilities 
of an SCFG are the LFP, $q^*$, of a corresponding
PPS, $x=P(x)$, which has the same
encoding size as $G$, we can conclude from Theorem
\ref{1comp} that for all
nonterminals $A$, where $NE(A) \neq 1$,  $NE(A) \geq 1/2^{4|G^{(2)}|}$.
Let us define $\zeta =  1/2^{4|G^{(2)}|}$.

For a fixed nonterminal $A$, since every rule $r'$ of $G^{(3)}$ associated
with $A$ has an expression whose denominator is the same, namely $NE(A)$,
since we have the lower bound $NE(A) \geq \zeta$, 
and since $G^{(3)}$ is proper, then
in order to make sure that our resulting SCFG $G^{(3)}_{\delta}$ is
also proper, 
it suffices
to only approximate ``sufficiently well'' the {\em numerators} 
of each rule probability $p(r')$
for rules $r'$ associated with $A$, and then normalize all these
values by their sum in order to get proper rule probabilities.

For each rule $r'$ of $G^{(3)}$ with positive probability $p(r') > 0$,
we wish to compute a rational value $v_{r'} \in (0,1]$ such that 
$| p(r') - v_{r'}| 
\leq \delta$.  

Note that $m := 3|G^{(2)}|$ is an easy upper bound on the number of
distinct rules in $G^{(3)}$.
How well do we have to approximate the numerators of probabilities, $p(r')$,
for rules $r'$  
associated with a nonterminal $A$ in $G^{(3)}$, 
in order to be sure that after normalizing
by their sum, these numerator values  yield
probabilities $v_{r'}$ such that the inequality 
$| p(r') - v_{r'} | \leq \delta$ holds for each one?
 The following claim addresses this:

\begin{claim}\label{how-good-each-rule-approx-claim}
Suppose $a_1, \ldots, a_r \in (0,1]$, where $r \leq m$,
suppose $b \in (\zeta,1]$, $0 < \zeta < 1$,
and suppose that $\sum^r_{i=1} a_i  = b$.
For any $\delta > 0$, such that $\delta \leq 1/(4m)$,
let $\delta' :=   \delta (\zeta/2)^2/ 4m$.
Suppose we find
$a'_1, \ldots, a'_m  \in (0,1]$ such that $|a_i - a'_i| \leq \delta'$
for all $i=1,\ldots, r$.
Let $b' = \sum^r_{i=1} a'_i$.  Then
for all $i= 1,\ldots, r$:
\[ | \frac{a_i}{b} -  \frac{a'_i}{b'} | \leq \delta \]
\end{claim}
\begin{proof}
First note that $|b -  b'| = | \sum^r_{i=1} (a_i - a'_i)| 
\leq r * \delta' \leq m * \delta'$.  
Then note that $0 < b' \leq b +  m \delta' \leq 1 + \delta (\zeta/2)^2/4
\leq 2$. 
Note also that $b' \geq b - m \delta' \geq \zeta/2$,
the last inequality following because $b \geq \zeta$, and
$m \delta' \leq \zeta/2$.
 Now we have:

\begin{eqnarray*}
| \frac{a_i}{b} -  \frac{a'_i}{b'}|  & =  &   
| \frac{b' a_i - b a'_i}{b b'} | \\
& \leq & | \frac{2 m \delta' \max(b,b')}{b b'} | \\
& \leq & | \frac{\delta (\zeta/2)^2}{b b'} | \\
& \leq & \delta
\end{eqnarray*}
\end{proof}

It follows from Claim 
\ref{how-good-each-rule-approx-claim} that
in order to approximate every rule probability within
distance $\delta > 0$, where we assume $\delta \leq 1/4 m$
where $m = 3 |G^{(2)}|$,
it suffices to approximate the numerators of 
the probabilities $p(r')$ for every
rule $r'$ associated with $A$ in $G^{(3)}$ within 
distance $\delta' =   \delta (\zeta/2)^2/ 4m$,  where 
$\zeta =  1/2^{4|G^{(2)}|}$ is the lower bound we know 
for $NE(A)$,  and then to normalize these approximated
numerators by their sum in order to obtain the respective 
probabilities.

To complete the proof,
we now consider 
separately all the possible forms the rule probability $p(r')$
could take, and show how to approximate each of their numerators
within $\delta'$.

\begin{enumerate}
\item Suppose $p(r') = p(r)/ NE(A) >0$, 
where $p(r)$ is the 
given rational 
rule probability of the corresponding rule $r$ of $G^{(2)}$.

In this case, we already have the exact rational numerator probability,
$p(r)$.

\item Suppose $p(r') = p(r)*NE(B)/ NE(A) >0$, 
where $p(r)$ is the 
given rational 
rule probability of the corresponding rule $r$ of $G^{(2)}$.

Since $p(r)$ is a probability, to approximate $p(r) * NE(B)$
within additive error $\delta'$, 
it suffices to approximate $NE(B)$ to within additive error
$\delta'$.   We already know how to do this in P-time,
because $NE(B)$ is the non-termination probability of a SCFG
that we can derive in P-time from $G^{(2)}$.

\item 
Suppose $p(r') = p(r)*E(B_1)*NE(B_2)/ NE(A) >0$.

To compute a $\delta'$-approximation $a_{r'}$ of the numerator,
such that  $| a_{r'} - p(r)*E(B_1) * NE(B_2) | \leq \delta'$,
we let $\delta'' = \delta'/4$,
and we compute  approximations $v^{B_1}_E$ and $v^{B_2}_{NE}$, 
of $E(B_1)$ and $NE(B_2)$, respectively,
such that $|E(B_1) - v^{B_1}_E | \leq \delta''$ and
$|NE(B_2) - v^{B_2}_{NE} | \leq \delta''$,
and we let $a_{r'} :=  p(r)*v^{B_1}_E * v^{B_2}_{NE}$.
Then:
\begin{eqnarray*}
| p(r)*v^{B_1}_E * v^{B_2}_{NE} - p(r)*E(B_1) * NE(B_2)| & \leq &
| v^{B_1}_E * v^{B_2}_{NE} - E(B_1) * NE(B_2) | \\
& \leq  &  2 \delta'' * \max(v^{B_1}_E,  v^{B_2}_{NE}, E(B_1), NE(B_2))\\
& \leq  &  \delta'
\end{eqnarray*}

\item The only remaining case is when
$p(r') = p(r)*NE(B_1)*NE(B_2)/ NE(A) >0$.  Its proof argument is
identical to  the previous case.
We can just replace $E(B_1) * NE(B_2)$ by $NE(B_1) * NE(B_2)$ 
and $v^{B_1}_E * v^{B_2}_{NE}$ by  $v^{B_1}_{NE} * v^{B_2}_{NE}$
in that argument.
\end{enumerate}

We have thus established that
there is a polynomial time algorithm that, given $G^{(2)}$ and $\delta > 0$, 
computes an SCFG  $G^{(3)}_{\delta} \in B_{\delta}(G^{(3)})$.
\end{proof}

Our next goal is to prove the following
Theorem:

\begin{thm}\label{g4-approx-step-thm}
There is a polynomial-time algorithm
that, given any proper SCFG $G$ with rational
rule probabilities, and given any 
rational value $\delta' >0$ in standard binary representation, 
computes a proper SCFG,  $G^{(4)}_{\delta'} \in B_{\delta'}(G^{(4)})$,
with rational rule probabilities.
\end{thm}

\begin{proof}
Recall that to obtain $G^{(4)}$ from $G^{(3)}$ 
we have to eliminate from $G^{(3)}$ the ``linear'' rules of the form 
$A \stackrel{p}{\rightarrow} B$, where $A$ and $B$ are nonterminals.
Lemma \ref{4th-step-transform-lem} showed how this can be done.
The proof of 
Lemma \ref{4th-step-transform-lem} considered the finite-state Markov
chain,  $M$,  whose states consist of the set of nonterminals of
$G^{(3)}$ as well as the set of all distinct ``right-hand sides'', 
$\gamma$, that appear in any rule $A \stackrel{p}{\rightarrow} \gamma$
of $G^{(3)}$, and where $\gamma$ is  not a single nonterminal,
and where the probabilistic transitions of $M$
basically correspond to the rules of $G^{(3)}$, plus 
extra absorbing self-loop transitions $(\gamma, 1, \gamma)$,
for every state $\gamma$ that in not a single nonterminal.

Note that $M$ may have irrational probabilities.
The proof of \ref{4th-step-transform-lem} showed that
the probabilities $q^*_{A_i,\gamma}$ of eventually 
reaching the state $\gamma$ from the state $A_i$ in $M$ 
can be used as the probabilities of new rules 
$A_i \stackrel{q^*_{A_i,\gamma}}{\longrightarrow} \gamma$, after eliminating
all linear rules, to obtain the SCFG $G^{(4)}$,
such that for all strings $w$,
and all nonterminals $A$,
we will have $p^A_{G^{(3)},w} = p^A_{G^(4),w}$.

Regardless whether $M$ has transitions with irrational probabilities
or not, the proof of Lemma \ref{4th-step-transform-lem} showed
that the probabilities $q^*_{A,\gamma}$ 
which we need can be obtained
as follows: we can first identify those probabilities 
$q^*_{A,\gamma}$ that are greater than $0$
by using only the underlying ``graph'' structure of the grammar rules 
with positive probability, without the need to access their actual probability.
Note that we determine which rules of $G^{(3)}$ have positive
probability by simply looking at which rules of $G^{(3)}_{\delta}
\in B_{\delta}(G^{(3)})$ 
have positive probability, for whatever $\delta > 0$ we have chosen,
because the definition of the approximate set $B_{\delta}(G^{(3)})$ requires
that positive probability rules retain positive probability in the
approximate SCFGs.

Once we have computed those cases where $q^*_{A,\gamma} = 0$,
in what remains the probabilities $q^*_{A,\gamma} > 0$ can be obtained
as the {\em unique} solution $x^* = (I-P)^{-1}b^{\gamma}$ 
of a corresponding linear system of 
equations  given in (\ref{hit-prob-gamma-equation}), which has the form:

\[  (I-P)x = b^\gamma \]

\noindent where $P$ is a substochastic $n \times n$ matrix.
Note that by basic facts about transient states in Markov chains
and substochastic matrices $P$, 
we have that $(I-P)^{-1} = \sum^{\infty}_{k=0} P^k \geq 0$.

Note that $P$ and $b^\gamma$ may have irrational entries,
because they have been derived from rule probabilities in $G^{(3)}$.
We may hope that by approximating sufficiently well the entries of 
$P$ and $b^\gamma$,
which are all rule probabilities of $G^{(3)}$, 
we can then use the resulting approximated linear system of equations, which
will hopefully
still have a unique solution which is close to the unique solution
of the original linear equation system.

Unfortunately, we can not do this in a very naive
way, because some
of the rule probabilities of $G^{(3)}$ (and thus transition
probabilities of $M$) 
have in their numerator expressions containing $E(B) = E_{G^{(2)}}(B)$ for some nonterminal 
$B$ of
$G^{(2)}$.   Since $E(B)$ amounts to the termination probability
of a SCFG (with rational rule probabilities) 
whose encoding size is $O(|G^{(2)}|)$,  it can unfortunately
be the case that some positive probabilities in entries of the matrix $P$
are {\em extremely small} (double exponentially small, and possibly irrational) 
values, 
namely as small as $1/2^{2^{c|G^{(2)}|}}$, for a fixed constant $c > 0$.

These very small entries mean that we
can not immediately rule out that the system of equations 
$(I-P)x = b^\gamma$ is potentially very {\em ill-conditioned}.

To overcome this, we observe crucially
that the Markov chain $M$ has a very special structure
which allows us to transform it into a different Markov chain $M'$,
by basically removing some small but positive probability transitions,
yielding a new chain $M'$ with 
no transition probabilities that are ``very small'',
and yet such that each state of $M'$  
has a probability of reaching state $\gamma$ that
is very close to that of reaching $\gamma$ in the original Markov chain $M$.

The special structure of $M$ arises for the following reason:
the only kinds of rules $r'$ out of a nonterminal $A$ of $G^{(3)}$ 
that have a probability $p(r')$ which contains a
probability $E(B)$
in its expression are rules of the form $r'(2)$ or $r'(3)$.
But then there must also exist a rule of the form $r'(1)$ associated
with $A$.  Since the expression $p(r'(1))$ does not contain
a probability $E(B)$  (only probabilities of the form $NE(B)$), 
we can lower bound $p(r'(1))$ sufficiently far away from zero.
Specifically,  since 
\[p(r'(1)) =  \frac{p(r) * NE(B) * NE(C)}{NE(A)} \]
where $p(r)$ is a rational rule probability in $G^{(2)}$
(which only has rational rule probabilities),
and $NE(A)$, $NE(B)$, and $NE(C)$ are all probabilities
of not generating $\epsilon$ starting at different nonterminals
in $G^{(2)}$, and since we know that these probabilities
can be rephrased as non-termination probabilities in a SCFG with
at most the same size as $|G^{(2)}|$,  it follows
from Theorem \ref{1comp} that 
$p(r'(1)) \geq \frac{1}{2^{9*|G^{(2)}|}}$.
Moreover, by definition the rule $r'(1)$ 
has the form $A \stackrel{p(r'(1))}{\longrightarrow} \gamma$,
where $\gamma$ is not a single nonterminal,
and thus the corresponding transition in $M$ must be
a transition 
to an absorbing state $\gamma$ of the Markov chain.

We claim that
this then means that whenever $p(r'(2))$ or $p(r'(3))$ are sufficiently
small probabilities, relative to $p(r'(1))$, 
then we can simply remove their corresponding transitions in $M$,
yielding a new Markov chain $M'$, and
this will not substantially 
change, starting at any state, the probability 
of eventually reaching any particular absorbing state $\gamma'$.

More precisely, let $\zeta =  \frac{1}{2^{9*|G^{(2)}|}}$.
Let $\delta' > 0$ be some desired error threshold.
Consider any state $A$ of $M$ such that there are
rules of the form $r'(2)$ and $r'(3)$ associated with $A$ in $G^{(3)}$,
and thus corresponding transitions with probability $p(r'(2))$ and 
$p(r'(3))$ in $M$.
Consider the absorbing state $\gamma$ of $M$,
to which there is a transition from the state $A$
with probability $p(r'(1)) \geq \zeta$.
Consider any states $\xi_1, x_2$ of $M$ (which may be absorbing or not).
Suppose that $p(r'(2)) \leq \zeta*\delta'/2$ 
and that $p(r'(3)) \leq \zeta*\delta'/2$.
Let us define the Markov chain $M'$ by removing all 
transitions, and 
let us ask how much the probability of eventually reaching
$\xi_2$ starting from $\xi_1$ in $M$ can change
if we simply remove both of these transitions
$r'(2)$ and $r'(3)$ out of state $A$ from $M$.

Since $\gamma$ is
an absorbing state,
we have that,
{\em even if we assume that there is a transition in $M$
from $A$ right back to itself with all
of the residual probability $(1-p(r'(1)) - p(r'(2)) - p(r'(3)))$},
then the total probability that, starting from $A$, we will
ever use either of the transitions $r'(2)$ or $r'(3)$ in $M$ is 
at most $( p(r'(2)) + p(r'(3)))/p(r'(1)) \leq \zeta*\delta'/ \zeta = \delta'$.
Note that the case where all the residual probability feeds back to $A$ 
yields
the highest possible probability of ever using either transition $r'(2)$
or $r'(3)$. 
Thus, by removing transitions $r'(2)$ and $r'(3)$, for any state $\xi_2$,
we would have at most changed the probability of eventually reaching
$\xi_2$ starting at $A$ by at most $\delta'$.
Likewise, starting at any state $\xi_1$, the  
probability of eventually reaching $\xi_2$ starting in $\xi_1$ in $M$ 
is at most changed by  $\delta'$ by removing
transitions $r'(2)$ and $r'(3)$ out of $A$, because any path using these
transitions must first go through state $A$, and the probability that
it will eventually go through $r'(2)$ or $r'(3)$ is at most $\delta'$.

For any desired $\delta' > 0$,  let us compute in P-time an approximate
$G_{\delta}^{(3)} \in B_{\delta}(G^{(3)})$, 
where $ \delta = \zeta * \delta'/4$.
We can
then detect the positive ``low probability'' 
transitions of the form  $r'(2)$ and $r'(3)$, whose probability
is $\leq \zeta * \delta'/2$, and we can remove them,
yielding a new Markov chain $M'$,
without changing the resulting probability of reaching any
absorbing state $\gamma'$ by more than $\delta' > 0$.

Let $q'_{A,\gamma}$ denote the probability of reaching
absorbing state $\gamma$ starting at state $A_i$ in the
Markov chain $M'$.  We know that 
$| q'_{A,\gamma} - q^*_{A,\gamma} | \leq \delta'$.
Our aim is thus to approximate the probabilities $q'_{A,\gamma}$
of eventually reaching the absorbing state $\gamma$ starting at any
nonterminal state $A$ in $M'$, to within a desired error $\delta'' >0$.

Let $A_1, \ldots, A_{n'}$ denote the nonterminal states of $M'$ such that 
$q'_{A_i,\gamma} > 0$.
(Note that we can detect such states in
 P-time by computing $G_{\delta}^{(3)}$, because these
are determined by the underlying graph based on rules with
probability $> 2*\delta$ in $G^{(3)}$, 
and these rules can be determined by computing $G_{\delta}^{(3)}$.)

Now the substochastic matrix $P'$ associated with $M'$ and $\gamma$
is defined to be an $n' \times n'$ matrix, where $P'_{i,j}$ is the
one-step transition probability from state $A_i$ to state $A_j$
in the Markov chain $M'$.

This yields for us, 
a new system $(I- P')x = b'^{\gamma}$.
Note that as defined $P'$ may still have irrational entries,
because we have not approximated the other positive rule probabilities
which were not removed.
Note that 
the states $A_1, \ldots, A_{n'}$ of $M'$ are transient,
and thus again by standard facts (e.g., \cite{BP94}, Lemma 8.3.20)
we have that the equation $(I- P')x = b'^{\gamma}$
has a unique solution  $\hat{x^*}= (I-P')^{-1} b'^\gamma$.

Furthermore, we have just argued that 
$| x^* - \hat{x^*} | < \delta'$.
It also always holds that the entries of $b'^\gamma$ are 
either zero or $\geq 1/2^{9*|G^{(2)}|}$, and that 
there is at least one non-zero entry in $b'^\gamma$.

Note that $(I-P')^{-1} = \sum^{\infty}_{k=1} (P')^k$.
We will need an upper bound on the row sums of $(I-P')^{-1}$.
Note that $P'^{k}_{i,j}$ is the probability of being in state $A_j$
in $k$ steps after starting in state $A_i$. 

\begin{claim}
\label{claim:small-interse-norm}
Let $c > 0$ denote the smallest positive entry in $P'$, and let 
$p > 0$ denote the smallest positive entry of $b'^\gamma$.
Then
for all $i \in \{1,\ldots,n'\}$,  
$0 \leq \sum^{n'}_{j=1} (I-P')^{-1}_{i,j} \leq  \frac{n'}{p c^{n'}}$
\end{claim}
\begin{proof}
Every state among $A_1,\ldots, A_{n'}$  has, by definition, 
a positive probability of reaching $\gamma$.
Thus, since there are $n'$ states in total, and each positive
probability transition has at least probability $c > 0$,
then the probability that, starting at any of these states $A_j$ 
we reach $\gamma$ within $n'$ steps is at least $p c^{n'}$.
Thus the probability of not reaching $\gamma$ within $n'$ steps
is $(1- p c^{n'})$.  But since this is the case for any such state $A_j$,
the probability of not reaching $\gamma$ within $d n'$ steps 
starting at any state $A_j$ is at most $(1-p c^{n'})^d$.

Now note that $(P')^{dn'}_{i,j}$ is the probability of being in state $A_j$
after $dn'$ steps.  But by what we have just argued, we know that
$\sum^{n'}_{j=1} (P')^{dn'}_{i,j} \leq (1-p c^{n'})^d$.
Thus, for all $i$,  $\sum^{\infty}_{d=0} \sum^{n'}_{j=1} (P')^{dn'}_{i,j} \leq
\sum^{\infty}_{d=0} (1-p c^{n'})^d = \frac{1}{p c^{n'}}$.

Similarly, for any $r \in \{1, \ldots,n'-1\}$,
the probability of not reaching $\gamma$ within $d n' + r$ steps is 
starting at any state $A_j$ is also at most $(1-p c^{n'})^d$.
Thus $\sum^{n'}_{j=1} (P')^{dn'+r}_{i,j} \leq (1-p c^{n'})^d$.
Thus for all $i \in \{1,\ldots,n'\}$, and all $r \in \{0,\ldots,n'-1\}$, 
we have  $\sum^{\infty}_{d=0} \sum^{n'}_{j=1} (P')^{dn'+r}_{i,j} \leq
\sum^{\infty}_{d=0} (1-p c^{n'})^d = \frac{1}{p c^{n'}}$.
But note that $(I-P')^{-1}_{i,j} = (\sum^{\infty}_{k=0} P')_{i,j}
= \sum^{n'-1}_{r=0} \sum^{\infty}_{d=0} \sum^{n'}_{j=1} (P')^{dn'+r}_{i,j}
\leq n' * \frac{1}{p c^{n'}} = \frac{n'}{p c^{n'}}$.
\end{proof}

We now show that 
approximating
the entries of $P'$ and $b'^\gamma$ to within a sufficiently small
desired accuracy  $\delta'' > 0$ 
yields a new approximate linear system of equations
whose {\em unique} solution $\tilde{x^*}$ is within a desired distance 
of $\hat{x^*}$, and thus within a desired distance of $x^*$.
For this, we use a standard  {\em condition number bound} for
errors in the solution of linear systems of equations:

\begin{thm}
(see, e.g., \cite{IsaKel66}, 
Chap 2.1.2, Thm 3.\footnote{Our statement is weaker, 
but is directly derivable from the cited Theorem.})
Consider a system of linear equations, $Bx = b$,
where $B \in \real^{n \times n}$ and
$b \in \real^n$.
Suppose $B$ is
non-singular, and $b \neq 0$.   Let ${x^* = B^{-1}b}$ be the unique solution
to this linear system, and suppose $x^* \neq 0$.
Let $\norm{\cdot}$ denote any vector norm and associated matrix norm
(when applied to vectors and matrices, respectively).
Let $\cond(B) = \norm{B} \cdot \norm{B^{-1}}$ denote the
condition number of $B$.
Let  $\varepsilon , \varepsilon' > 0$,  be values such that
$\varepsilon' < 1$, and
$\varepsilon
    \cdot \cond(B) \leq \varepsilon'/4$.
Let ${\mathcal E} \in \real^{n \times n}$ and
${\theta} \in \real^n$, be such that
$\frac{\norm{{\mathcal E}}}{\norm{B}} \leq \varepsilon$,
    $\frac{\norm{\theta}}{\norm{b}} \leq \varepsilon$, and
$\norm{{\mathcal E}} < 1/\norm{B^{-1}}$.
Then the system of linear equations
 $(B+{\mathcal E})x = b+\theta$
    has a unique solution $x^*_\varepsilon$ such that:

 $$ \frac{\norm{x^*_\varepsilon - x^*}}{\norm{x^*}} \leq
\varepsilon'$$

\label{thm:lin-cond-bound}
\end{thm}

We will apply this theorem using the $l_\infty$ vector norm 
and induced matrix
norm
({\bf\em maximum absolute row sum}):
$ \norminf{x} := \max_i \abs{x_i}$ \mbox{and} $
\norminf{A} := \max_i \sum_j \abs{a_{ij}}$.

Let us define the matrices and vector in the statement of Theorem
\ref{thm:lin-cond-bound} as follows:
$B := (I-P')$ and  $b := b'^\gamma$. Note that 
$B = (I-P')$ is non-singular and $B^{-1} = \sum^\infty_{k=0} (P')^k$,
and $x^* = B^{-1}b$ is the unique solution to the linear system $Bx = b$,
and that  $x^* \neq 0$.

Let us now give bounds for, $\| B \|$, $\| B^{-1} \|$  and
$\cond(B) := \| B \| \| B^{-1} \|$.
(Note that we define $\| \cdot \| := \| \cdot \|_\infty$.)

\begin{claim}
\label{claim:bounds-on-B}
Let $p$ be the smallest non-zero probability labeling any transition in $M'$.
Then   $p \leq \| B \| \leq 2$ and $\| B^{-1} \| \leq \frac{n'}{p^{n'+1}}$.
Thus $\cond(B) \leq \frac{2n'}{p^{n'+1}}$.
\end{claim}
\begin{proof}
$B = (I-P')$ and 
$P'$ is a substochastic matrix.
Thus $\| (I -P') \| \leq 2$.  Furthermore, since every transient state
$A_j$ indexing rows and columns of $P'$ has, by definition, 
non-zero-probability of reaching the absorbing
state $\gamma$, we know that the probability of returning 
from any state $A_j$ immediately back to itself is at most $(1-p)$,
and thus  for any $j$,  $(I-P')_{j,j} \geq p$, and thus
$\| B \| \geq p$.
Next, $\| B^{-1}\| = \| \sum^\infty_{k=0} (P')^k \|$, and we 
established in Claim \ref{claim:small-interse-norm}  that
$\| \sum^\infty_{k=0} (P')^k \| \leq \frac{n'}{p^{n'+1}}$.
\end{proof}

We define $\varepsilon, \varepsilon' > 0$ as follows:
choose an arbitrary desired 
$\varepsilon' = \delta'$ so that $0 \leq \varepsilon' < 1$.
Then let $\varepsilon = \frac{\varepsilon'}{4 * \cond(B)}$.

Given the bounds on $\| B \|$, $\| B^{-1} \|$, and $\cond(B)$,
in Claim \ref{claim:bounds-on-B},
we are able to choose a suitable $\delta$ with polynomial encoding size, 
namely $\delta :=  \zeta * \varepsilon /(4 * n' * \cond(B) * \| B^{-1} \|)$, 
and compute in P-time an
approximation $G^{(3)}_{\delta} \in B_{\delta}(G^{(3)})$.
Using $G^{(3)}_{\delta}$ we can compute
an approximation $\tilde{P}'$
for the matrix $P'$  (since the entries of $P'$ are rule probabilities
in $G^{(3)}$, except for those probabilities that are too low, $\leq \delta *2$, 
which we can remove because we can detect them),
and we can also compute an approximation $\tilde{b}$ for the vector $b = b'^\gamma$,
such that, letting $\tilde{B} := (I- \tilde{P'})$, 
 if we let  ${\mathcal E} := \tilde{B} - B$,
and we let ${\mathcal \zeta} := \tilde{b} - b$,
then  
$\frac{\|\mathcal E \|}{\| B \|} \leq \varepsilon$ 
and $\frac{\|\mathcal \zeta \|}{\| b \|} \leq \varepsilon$,
and $\| {\mathcal E } \| \leq 1/\| B^{-1} \|$.

Thus all the conditions are in place to apply 
\ref{thm:lin-cond-bound}, and we have that the system  $\tilde{B} x = \tilde{b}$ has
a unique solution $\tilde{x}^* =  \tilde{B'}^{-1} \tilde{b}$, and that
$ \frac{\norm{x^*_\varepsilon - x^*}}{\norm{x^*}} \leq
\varepsilon'$.  Since we know that $0 < \| x^* \| \leq 1$, we have that 
$\norm{x^*_\varepsilon - x^*} \leq \varepsilon' = \delta'$.

We have thus established that we can approximate within a desired additive error $\delta'$
the probabilities $q^*_{A_j,\gamma}$ of eventually
reaching $\gamma$ from any nonterminal $A_j$ via linear rules.    In order to construct
$G^{(4)}_{\delta'} \in B_{\delta'}(G^{(4)})$ using this, and by adding suitable rules 
approximating $A_j \stackrel{q^*_{A_j,\gamma}}{\longrightarrow} \gamma$,
we have to make a few relatively easy technical observations.
Firstly, in our computations we have eliminated rules whose probability was ``too low'' to
effect the overall probability of reaching $\gamma$ significantly.
However, our definition of $G^{(4)}_{\delta'}$ requires that every rule 
that has positive probability in $G^{(4)}$ should have positive probability 
in $G^{(4)}_{\delta'}$.   This is easy to rectify: by the choice of $\delta$ in our
approximation $G^{(3)}_{\delta}$, we can, for all rules $A_j \rightarrow \gamma$
that should have positive probability, we can put in such rules with a small enough
positive probability $\delta/2$, so that it does not effect the overall probability
of reaching any $\gamma$ substantially.
This finally leads us to another point: we must make sure $G^{(4)}_{\delta'}$
is a proper SCFG.   This we can do again, because, by choosing $\delta$ to be suitably
small, we can make sure that we can normalize the sum of weights on the approximated rules 
coming out of each nonterminal without changing any particular rule probability substantially.
This completes the proof that we can compute $G^{(4)}_{\delta'} \in B_{\delta'}(G^{(4)})$ in
P-time.
\end{proof}

Finally, we are ready to finish the proof of both Theorems \ref{main-scfg-parsing-thm} and 
\ref{main-scfg-cnf-form-thm},
both of which are direct corollaries of the following Lemma:

\begin{lem}
\label{lem:final-lem-approx-parse}
Given any SCFG, $G$, with rational rule probabilities, 
given any rational value $\delta > 0$ in standard binary representation,
and given any natural number $N$ specified in unary representation, 
there is a polynomial time algorithm that
computes an SCFG, $G^{(5)} = G^{(4)}_{\delta'} \in B_{\delta'}(G^{(4)})$,
for a suitably chosen $\delta' = f(|G|,N,\delta) > 0$, 
where $f$ is some polynomial function, 
such that for all nonterminals $A$,
and for all strings $w \in \Sigma^*$,  
such that $|w| \leq N$, it holds that

\[   | p^A_{G^{(4)},w} - p^A_{G^{(5)},w} | \leq \delta \]

Moreover, given $G$, $\delta > 0$, and a string $w$ of
length at most $N$, we can compute in polynomial time
(in the standard Turing model of computation) a value $v_{G,w}$
such that 

\[  | p^{A}_{G^{(4)},w} - v_{G,w} | \leq \delta \]
\end{lem}

\begin{proof}
To prove this theorem, we will exploit the
standard dynamic
programming algorithm 
(a variant of Cocke-Kasami-Younger) 
for 
computing the ``inside'' probability, $p_{G,w}$, for an SCFG $G$ 
that is already in Chomsky Normal Form.

The algorithm
was originally observed as part of the inside-outside
algorithm by Baker \cite{Baker79} (see also \cite{LariYoung90}).
It works in P-time in the unit-cost arithmetic RAM model of computation,
meaning it uses a polynomial number of arithmetic $\{+,*\}$ operations,
and inductively computes the probabilities,  
$q^A_{i,j}$, that starting at the nonterminal $A$ 
the SCFG generates the string
of length $j$ starting in position $i$ of the string $w$.
In other words,  $q^A_{i,j} := P_G( A \stackrel{*}{\rightarrow} 
w_i \ldots w_{i+j-1})$.

The induction in the dynamic program is on the length, $j$, of the string.
The base case of the induction is easy:  $q^A_{i,1}$ is
the probability $p$ of the rule  $A \stackrel{p}{\rightarrow} w_i$.
If no such rule exists, then $q^A_{i,1} := 0$.

Note that we can assume the CNF grammar $G^{(4)}$ and $G^{(5)}$ do
do not have any $\epsilon$ rules.\footnote{Also, we 
can of course compute/approximate the probability that
a CNF grammar generates the empty string, $\epsilon$.
The only possible $\epsilon$ rule is $S \stackrel{p}{\rightarrow}
\epsilon$, and this can only appear if $S$ is not on
the RHS of any rule. Thus the probability of generating
$\epsilon$ from $S$ is $p$ (or $0$ if no such
rule exists), and it is $0$ for all other nonterminals.}

For the inductive step, assume we have already computed $q^A_{i,j'}$
for all nonterminals $A$, all $i$, and all $j'$ such that $1 \leq j' < j$.
Let the rules associated with $A$ whose RHSs do not
consist of just a terminal symbol be $A \stackrel{p_1}{\rightarrow} X_1 Y_1$,
$A \stackrel{p_2}{\rightarrow} X_2 Y_2$, $\ldots$,  $
A \stackrel{p_k}{\rightarrow} X_k Y_k$,   where $X_d$ and $Y_d$ are
nonterminals for all $d = 1, \ldots, k$.
(It may of course be the case that $k=0$, i.e., that there are 
no such rules associated with $A$ in this CNF grammar.)

Then it is easy to check that the probability $q^A_{i,j}$ can
be computed inductively by the following arithmetic expression:

\begin{equation}
\label{cky-dyn-equation}
  q^A_{i,j} =   \sum^k_{d = 1}  p_d  \sum^{j-1}_{m = 1} q^{X_d}_{i,m}
q^{Y_d}_{i+m,j-m}
\end{equation}

Thus by induction, we can compute $q^S_{1,n}$ which is precisely
the probability that the grammar $G$ generates the string $w$ 
starting with the start nonterminal $S$.
In this way this algorithm computes $p_{G,w}$  
for SCFGs $G$ that are already in CNF.

It is important to point out two issues with the
above algorithm.  

\begin{enumerate}
\item Firstly, even if we assume we are
given as input a SCFG $G$ which is already in CNF form, 
and where all of the rules
have {\em rational} probabilities,
the above inside algorithm,
as described, only works in P-time in the {\em unit cost arithmetic 
RAM} model of computation, because although it only requires
a polynomial number of arithmetic operations to compute $q^{S}_{1,n}$,
since we require iterated multiplications, it means
that in principle it is possible for 
the rational values $q^A_{i,j}$ that we compute
to blow up in encoding size, and in particular to require 
encoding size that is exponential in $j$, the length of the string
being parsed.

Thus, to carry out the inside algorithm in P-time in the standard Turing model
of computation,
we need to show how we can approximate the output
$q^A_{i,j}$ in P-time in the Turing model.    
We shall show that indeed
rounding the intermediate computed values to within a suitable
polynomially many bits of precision suffices to achieve this,
and thus that the approximate total probability parsing problem 
when the SCFGs are  already in CNF form can be carried out
in P-time.

\item A second problem is that 
our original grammar $G^{(4)}$ has irrational rule probabilities,
so we had to approximate $G^{(4)}$ with a suitable 
$G^{(5)} = G^{(4)}_{\delta'}$.
In this case the CKY dynamic programming method is being
applied to rule probabilities that  only approximate the 
``true'' values
of the rule probabilities.  We show that with a good enough approximation
this can not severely effect the overall probability that is computed.

\end{enumerate}

We will show that both of these issues can be addressed in
the same way:  by approximating the original rule probabilities
to within sufficient accuracy requiring only polynomial computation, and then by using 
these inductively
in the CKY dynamic programming algorithm, and rounding to within
sufficiently accuracy after each step of the induction, and iterating the induction up to
the string length value $N$ given in unary,
we will be able to compute in P-time
(in the standard Turing model of computation) an output value
that is within desired accuracy of the ``true'' output value
of this algorithm in the unit-cost RAM model on the original
(irrational) SCFG $G^{(4)}$, and thus we will 
have approximated the desired probabilities $p^{A}_{G,w}$ 
to within sufficient accuracy.

The key observation for why the above approximations
are possible is the following.
Suppose that we are inductively attempting to approximate
the value of $q^A_{i,j}$, having been given $\delta$-approximations
of all quantities on the RHS of the inductive equation:
 
\[  q^A_{i,j} =   \sum^k_{d = 1}  p_d  \sum^{j-1}_{m = 1} q^{X_d}_{i,m}
q^{Y_d}_{i+m,j-m} \]

First, observe that if we have $\delta$-approximated 
the probabilities $q^{X_d}_{i,m}$ 
and $q^{Y_d}_{i+m,j-m}$ with values $v^{X_d}_{i,m} \in [0,1]$ and 
$v^{Y_d}_{i+m,j-m} \in [0,1]$ respectively,
then since all of these values are in $[0,1]$  we have that

\begin{eqnarray*} 
| q^{X_d}_{i,m}
* q^{Y_d}_{i+m,j-m} - v^{X_d}_{i,m} * v^{Y_d}_{i+i,m} | 
& \leq & \delta \max(q^{X_d}_{i,m},  v^{X_d}_{i,m}) + 
 \delta  \max(q^{Y_d}_{i+m,j-m}, v^{Y_d}_{i+m,j-m})\\
&  \leq &  2 \delta 
\end{eqnarray*}

Therefore, 

\[ | \sum^{j-1}_{m = 1} q^{X_d}_{i,m}
q^{Y_d}_{i+m,j-m} -  \sum^{j-1}_{m = 1} v^{X_d}_{i,m}
v^{Y_d}_{i+m,j-m}| \leq j 2 \delta \]

Finally, if we have $\delta$-approximated all probabilities $p_d$ with
$v_d \in [0,1]$, in such a way that $\sum^k_{d=1} v_d \leq 1$,
then since we know that $\sum^k_{d=1} p_d  \leq 1$,
we have that:

\[ |  p_d  \sum^{j-1}_{m = 1} q^{X_d}_{i,m}
q^{Y_d}_{i+m,j-m} - v_d  \sum^{j-1}_{m = 1} v^{X_d}_{i,m}
v^{Y_d}_{i+m,j-m}|  \leq j 2 \delta + \delta j \leq 4 j \delta  \]

Thus 

\[ |  \sum^k_{d = 1}  p_d  \sum^{j-1}_{m = 1} q^{X_d}_{i,m}
q^{Y_d}_{i+m,j-m}  -   \sum^k_{d = 1}  v_d  \sum^{j-1}_{m = 1} v^{X_d}_{i,m}
v^{Y_d}_{i+m,j-m}| \leq  4 k j \delta \]

Thus, if the error accumulated at the previous iteration was $\delta$,
then the error accumulated at the next iteration is $4 k j \delta$.
Inductively, since there are $N$ iterations in total, we
see that if $m$ denotes the number of rules plus the number
of distinct RHSs in the grammar $G^{(4)}$,  
then the total error after all $N$ iterations,
assuming the base case has been computed
to within error $\delta$, is at most

\[ (4 m^2)^N \delta \]

Note that this error amounts to a loss of 
only polynomially many ``bits of precision'' in the input size.
Note also that we can, after each iteration, round the values
to within the desired polynomially many bits of precision,
only accumulating negligible extra error, and still maintain
the overall loss of only polynomially many ``bits of precision'', even when all
computations are on numbers of polynomial bit length.
\end{proof}

\section{Addendum: Quadratic convergence (with explicit
constants) for Newton's method on PPSs, and
quantitative decision problems for PPSs in unit-cost-P-time}

\label{sec:quadratic-conv-quant-decision}

In this section we
combine results stated in the STOC'12  conference version
of this paper together with
other results, in particular results stated and proved in
a subsequent paper that appeared at ICALP'12 \cite{bmdp}, 
in order to extend 
Theorem
\ref{linearconvgen} to a 
{\em quadratic convergence} result for Newton's method on PPSs (with all 
constants explicit).
Namely, given a PPS, $x=P(x)$, with LFP  ${\mathbf 0} < q^* < {\mathbf 1}$, 
if we start Newton iteration
at $x^{(0)} := {\mathbf 0}$, then 
for all $i \geq 1$, we have 
$$\| q^* - x^{(32|P| + 2 + 2i)}\|_{\infty} \leq \frac{1}{2^{2^{i}}}$$

We then use this result to show that the {\bf\em decision problem}
for the LFP $q^*$ of PPSs, which asks,
given a PPS $x=P(x)$ over $n$ variables, 
and given a rational number $r \in [0,1]$, decide
whether $q^*_i > r$  (or whether $q^*_i \geq r$)  
is decidable {\em in the unit-cost arithmetic
RAM model of computation} in polynomial time,
and thus this decision problem is itself reducible to the PosSLP problem.
We in fact show further that deciding whether $q^*_i > r$ is
P-time many-one (Karp) reducible to PosSLP.

We assume throughout this
section, w.l.o.g., that every PPS, $x=P(x)$,   is in {\em simple normal form},
and that the LFP, $q^*$ satisfies ${\mathbf 0} < q^* < {\mathbf 1}$.
We will use the following Theorem from \cite{bmdp}:

\begin{thm}  (Theorem 4.6 of \cite{bmdp}) 
\label{thm:normbounds}
If $x=P(x)$ is a PPS with LFP $q^* > 0$  then\\
(i) If $q^* < 1$ and $0 \leq y < 1$, then 
 $(I-B(\frac{1}{2}(y+q^*)))^{-1}$ exists and is non-negative, and
\begin{equation}
\label{eq:normbound-mixed}
\|(I-B(\frac{1}{2}(y+q^*)))^{-1}\|_\infty \leq  2^{10|P|} 
 \max \{2(1-y)_{\min}^{-1}, 2^{|P|}\}
\end{equation}
(ii) If $q^* = 1$ and $x = P(x)$ is strongly connected 
(i.e. every variable depends directly or indirectly  on every other) 
and $0 \leq y <  1 = q^*$, then
$(I -B(y))^{-1}$ exists and is non-negative, and 
$$\|(I -B(y))^{-1}\|_\infty \leq 2^{4|P|} \frac{1}{(1-y)_{\min}}$$
\end{thm}

We note, for clarity, that the proof of Theorem \ref{thm:normbounds}
in \cite{bmdp} exploits
results we have established in this paper, but is otherwise independent
of any results we establish in this addendum section (and thus there
is no reason to fear circular reasoning in our proofs).

\begin{cor}
\label{cor:norm-bound}
If $x=P(x)$ is a PPS with LFP $q^*$, and  $0 < q^* < 1$,
then $(I-B(q^*))^{-1}$ exists and is non-negative, and
$$\|(I-B(q^*))^{-1}\|_\infty \leq  2^{14|P| + 1}$$
\end{cor}
\begin{proof}
Applying part (i) of  Theorem \ref{thm:normbounds},  and letting $y := q^*$,
we obtain
\begin{eqnarray*}
\|(I-B(q^*))^{-1}\|_\infty & \leq &   2^{10|P|} 
 \max \{2(1-q^*)_{\min}^{-1}, 2^{|P|}\}\\
& \leq &   2^{10|P|} 
 \max \{2(2^{-4|P|})^{-1}, 2^{|P|}\}  \quad \quad \mbox{(by Theorem 
\ref{1comp})} \\
&  = &  2 \cdot 2^{14 |P|}  = 2^{14 |P| +1}.
\end{eqnarray*}
\end{proof}

\begin{lem}
\label{lem:one-point-five}
If $x=P(x)$ is a PPS with $n$ variables
in simple normal form (SNF), with LFP  $0 < q^* < 1$,
then for any $z \in \real^n$ such that $0 \leq z  \leq q^*$,  if:

$$\| q^* - z \|_\infty \leq \frac{1}{2^{28 |P|+2}}$$
then 
$$ \| q^* - {\mathcal N}(z) \|_\infty \leq \|q^* - z \|_\infty^{1.5}$$
\end{lem}
\begin{proof}
Let us first note that
\begin{equation}
\label{eq:bound-on-jac-diff}
\| \frac{B(q^*) - B(z)}{2}\|_\infty 
\leq |(q^* - z) \|_\infty
\end{equation}
This holds because $x=P(x)$ is in SNF form, and thus 
every equation $x_i = P_i(x)$ is either of the form $x_i = x_j x_k$,
or else it is a {\em linear} (affine) equation, of the form $x_i = \sum^n_{j=1} p_j x_j + p_0$.
Now, 
for every $i$ with a non-linear equation, i.e., where $P_i(x) \equiv x_j x_k$,  
the $i$'th row of
the Jacobian matrix $B(x)$,  
contains exactly two non-zero entries:
one is $x_j = \frac{\partial P_i(x)}{\partial x_k}$ and 
the other is $x_k = \frac{\partial P_i(x)}{\partial x_j}$.
Thus, if we
define the matrix  $A = \frac{B(q^*) - B(z)}{2}$,  we must have   
$\sum^n_{r=1} |A_{i,r}| =   \frac{(q^*_j - z_j)  + (q^*_k - z_k)}{2} \leq \| q^* - z \|_\infty$.
Furthermore, for every $i$ with a {\em linear} equation,  the $i$'th row of the
Jacobian matrix $B(x)$ consists of only constants that do not depend on $x$, 
and thus in that case $\sum^n_{r=1} |A_{i,r}|  =  0 \leq  \| q^* - z \|_\infty$.
Thus inequality (\ref{eq:bound-on-jac-diff}) holds.

Now, using Lemma \ref{newton},
and the equation it gives, namely:
\begin{equation}
\label{eq:newton-char}
q^* -  \mathcal{N}(z) = (I-B(z))^{-1}\frac{B(q^*) - B(z)}{2}(q^* - z)
\end{equation}
and taking norms on both sides of this equation, we have:

\begin{eqnarray*}
\| q^* - {\mathcal N}(z) \|_\infty & = & 
   \| (I-B(z))^{-1}\frac{B(q^*) - B(z)}{2}(q^* - z) \|_\infty\\
 & \leq &  \| (I-B(z))^{-1} \|_\infty \|\frac{B(q^*) - B(z)}{2}\|_\infty 
\| (q^* - z) \|_\infty\\
& \leq &   2^{14|P| + 1}  \|\frac{B(q^*) - B(z)}{2}\|_\infty 
\| (q^* - z) \|_\infty  \quad \quad  \mbox{(by Corollary \ref{cor:norm-bound})}\\
& \leq & 2^{14 |P| + 1} \| (q^* - z) \|_\infty^2   \quad \quad 
\mbox{ (by inequality (\ref{eq:bound-on-jac-diff}))} 
\end{eqnarray*}

Note that if $\|q^* - z \|_\infty = 0$, then we are trivially
done with the lemma.

Thus, assuming $0 < \| q^* - z \|_\infty \leq  \frac{1}{2^{28|P|+2}}$,
then   $\| q^* - z \|_\infty^{-0.5} \geq 2^{14|P|+1}$.
Thus, 

\begin{eqnarray*} 
\| q^* - \mathcal{N}(z) \|_\infty  & \leq  &   2^{14 |P| + 1} \| (q^* - z) \|_\infty^2  \\
& \leq  &  \| (q^* - z) \|_\infty^{-0.5} \cdot   \| (q^* - z) \|_\infty^{2}   =  \| (q^* - z) \|_\infty^{1.5}
\end{eqnarray*}
\end{proof}

\begin{thm}
\label{thm:double-exponential-convergence}
Let $x=P(x)$ be any PPS in SNF form, with
LFP $q^*$, such that $\mathbf{0} < q^* < \mathbf{1}$.
If we start Newton iteration at $x^{(0)} := \textbf{0}$,
with $x^{(k+1)} :=  \mathcal{N}_P(x^{(k)})$, 
then for any integer $i \geq 1$ the following inequality holds:
$$\| q^* - x^{(32|P| + 2 + 2i)}\|_{\infty} \leq \frac{1}{2^{(28|P|+2)2^i}}  \leq \frac{1}{2^{2^{i}}}$$
\end{thm}

\begin{proof}
By Theorem \ref{linearconvgen},  
for $k \geq 32 |P| +2 = (28 |P| + 2) + 4 |P|$,  
we have $\| q^* - x^{(k)}\|_\infty \leq 
\frac{1}{2^{(28|P|+2)}}$.
Thus, we can apply Lemma  \ref{lem:one-point-five} repeatedly,
and by induction, for $i \geq 1$, we have:

\begin{eqnarray*}
\| q^* - x^{(32 |P| + 2 + 2i)} \|_\infty & \leq & 
\| q^*- x^{(32|P|+2)}\|_\infty^{1.5^{2i}} \\
& \leq &   \frac{1}{2^{(28|P|+2)2^i}}   \quad \quad \mbox{(because 
$1.5^2 \geq 2$, 
 and $\| q^*- x^{(32|P|+2)}\|_\infty \leq \frac{1}{2^{(28|P|+2)}}$)}\\
& \leq & \frac{1}{2^{2^i}}
\end{eqnarray*}
\end{proof}

We next wish to use 
Theorem \ref{thm:double-exponential-convergence}
in order to establish that, using Newton's method with {\em exact} arithmetic, 
{\em in the unit-cost 
arithmetic RAM model of computation}, we can decide, given a rational number $r$,  whether $q^*_i \geq r$,
in time polynomial in $|P|$ and the encoding size of $r$.

To do this, we need to first establish a separation bound relating to $q^*$ and a given rational $r$.

\begin{lem}
\label{lem:sep-bounds-from-rat}
Given a PPS, $x=P(x)$, with $n$ variables, and with LFP $q^*$, such that $0 < q^* < 1$, and 
given a rational number 
$r > 0$, where $r = \frac{a}{b} < 1$ is represented 
as the ratio of positive integers $a$ and $b$ given in binary, with $a \leq b$, then for any
$k \in \{1,\ldots, n\}$, 
if  $q^*_k \neq r$,   then    $$|q^*_k - r| \geq  
    2^{-2(n+1)(\max\{|P|, \log(b)\} + 2(n+1)\log(2n+2))5^{n}}$$
\end{lem}

\begin{proof}

We shall use the following Theorem by Hansen et. al. \cite{HLMT'12}
regarding explicit separation bounds for isolated real-valued solutions 
to polynomial systems of equations:
\begin{thm}{(Theorem 23 from \cite{HLMT'12})}
  \label{th:isol-real-root-bd-full}
  Consider a polynomial system of equations
  \begin{equation}
    (\Sigma) \quad \quad
    g_1(x_1, \dots, x_n) = \cdots = g_m(x_1, \dots, x_n) = 0 \enspace,
  \label{eq:orig-system}
  \end{equation}
  with polynomials of degree at most $d$ and integer
  coefficients of magnitude at most $2^{\tau}$.

  Then, the coordinates of any {\em isolated} (in Euclidean topology) real 
  solutions of the system are
  real algebraic numbers of degree at most $(2d+1)^n$, and their
  defining polynomials have coefficients of magnitude at most 
  $2^{2n(\tau+4n\log(dm))(2d+1)^{n-1}}$.
  Also, if $\gamma_j =
  (\gamma_{j,1}, \cdots, \gamma_{j,n})$ is an isolated solution of $(\Sigma)$,
  then for any $i$, either
  \begin{equation}
    2^{-2n(\tau + 2n\log(dm))(2d+1)^{n-1}} < | \gamma_{j,i } |
    \quad \text{ or } \quad \gamma_{j,i} = 0 \enspace.
    \label{eq:isol-rr-lower-bd}
  \end{equation}
  Moreover, given coordinates of isolated solutions of two such systems, if they are not identical, they differ by at least
  \begin{equation}
    \sep(\Sigma) 
    \geq 2^{-3n(\tau + 2n\log(dm))(2d+1)^{2n-1} - \frac{1}{2}\log(n)}
    \label{eq:isol-rr-sep-bd}
    \enspace .
  \end{equation}
\end{thm}

To apply  Theorem \ref{th:isol-real-root-bd-full}, we 
need the fact that  $q^* > 0$ is
an isolated solution of the PPS.  
This follows immediately from a more general {\em unique fixed point} theorem
established in 
\cite{EY-rmc-model-checking-2012} (Theorem 18 of \cite{EY-rmc-model-checking-2012}) for 
the equations corresponding to the termination probabilities
of a general recursive Markov chains (RMCs), and it also follows 
from (variants of) older
results about multi-type branching processes (see \cite{Harris63},
Thm. II.7.2 and Corollary II.7.2.1)
PPSs correspond to the special case of MPS equations for 1-exit RMCs.

Specifically, the unique fixed point theorem of
\cite{EY-rmc-model-checking-2012} establishes that, in particular, 
if a PPS has LFP
$q^*$ with $0 < q^* < 1$, then $q^*$ is the unique solution of
$x=P(x)$ in the interior of $[0,1]^n$, i.e., in $(0,1)^n$.  Thus, it
is clearly an isolated solution.

For each $x_i$, let $d_i$ be the product of the denominators of  all coefficients of 
$P(x)_i$. Then $d_ix = d_iP(x)_i$ clearly has 
integer coefficients which are no larger than $2^{|P|}$.
Also, consider a new variable $y$, and a new equation $y = x_k - r$,
where $r = \frac{a}{b}$ is the given positive rational value.
This equation is clearly equivalent to $b y = b x_k - a$.
Suppose the PPS, $x = P(x)$, has LFP $q^* > 0$, and 
for any $k \in \{1, \ldots, n\}$, 
consider the system of $n+1$ polynomial equations,  in $n+1$ variables
(with an additional variable $y$), given by:

\begin{equation}
 d_ix_i = d_iP(x)_i \; , \mbox{for all $i \in \{1,\ldots, n\}$}; \ \  \mbox{and}
\  \ \ by = b x_k - a   \ .
\label{eq:system-to-be-applied-on-iso}
\end{equation}

Since  $0 < q^* < 1$, we know from 
the unique fixed point theorem of \cite{EY-rmc-model-checking-2012} 
that  $q^*$ is an isolated solution of $x=P(x)$. 
If $z \in \real^n$ is any solution vector for $x=P(x)$, there is a unique $w \in \real$ such that 
$x:=z$ and $y :=w$ 
forms a solution to the equations (\ref{eq:system-to-be-applied-on-iso});
namely let $w = z_k - r$.
So, letting $x := q^*$, and letting $y := q^*_k - r$, 
gives us
an isolated solution of the equations (\ref{eq:system-to-be-applied-on-iso}). 
 We can now apply Theorem \ref{th:isol-real-root-bd-full} to the system
(\ref{eq:system-to-be-applied-on-iso}).  Since $y := q^*_k - r$, 
equation (\ref{eq:isol-rr-lower-bd}) in Theorem \ref{th:isol-real-root-bd-full} 
says that
  \begin{equation*}
    2^{-2(n+1)(\max \{|P|, \log(b)\} + 2(n+1)\log(2n+2))5^{n}} < | q^*_k - r | \; , 
    \quad \text{or else } \quad q^*_k - r = 0 \enspace.
  \end{equation*}
which is just what we wanted to establish.
\end{proof}

\noindent We are now ready to establish the following:

\begin{thm}
\label{thm:deciding-in-unit-cost}
Given a PPS, $x = P(x)$,  with $n$ variables, and with LFP $0 < q^* < 1$,  and given a rational number $r = a/b \in (0,1]$, where
$a$ and $b$ are 
positive integers given in binary.   Let $g= 32 |P| + 4 + 6 n + 56 (\lceil \log(n) \rceil + \lceil \log (|P|) \rceil + 
\lceil \log( \log b)) \rceil)$.
Let $x^{(i)}$ denote the $i$'th Newton iterate starting
at $x^{(0)} := 0$,  applied to the PPS $x=P(x)$.  
Let  $m := 2 + 3n +28 (\lceil \log(n) \rceil + \lceil \log (|P|) \rceil + 
\lceil \log( \log b)) \rceil)$.
 Then for any $k \in \{1, \ldots, n\}$,  
\begin{enumerate}
\item $q^*_k > r$  if and only if    $x^{(g)}_k > r$.

\item $q^*_k < r$ if and only if  $x^{(g)}_k + 2 \cdot \frac{1}{2^{2^m}} < r$.
\end{enumerate}
\end{thm}
\begin{proof}
Let  $\gamma =     2^{-2(n+1)(\max\{|P|, \log(b)\} + 2(n+1)\log(2n+2))5^{n}}$.
Recall that Lemma \ref{lem:sep-bounds-from-rat}  tells us that $|q^*_k - r | \geq \gamma$,
for any $k$, unless $q^*_k = r$.
We know  $x^{(g)} \leq q^*$.   Furthermore, $g$ has been chosen so  
that, by Theorem \ref{thm:double-exponential-convergence},   $\| q^* - x^{(g)} \|_\infty < \frac{1}{2^{2^{m}}} < \gamma/8$.

To establish (1.), in one direction we simply note that if $x^{(g)}_k > r$, then since $q^*_k > x^{(g)}_k$,
we must have $q^*_k > r$.  In the other direction, if 
$q^*_k > r$, then  $q^*_k - r \geq \gamma$,  but we know $q^*_k - x^{(g)}_k \leq  \gamma/8$, so 
$x^{(g)}_k \geq r + \frac{7}{8} \gamma \geq r$.

To establish (2.),  in one direction since $q^*_k - x^{(g)}_k <  \frac{1}{2^{2^{m}}}$,  
we have $q^*_k <  x^{(g)}_k + 2 \cdot \frac{1}{2^{2^m}}$,  and thus if 
$x^{(g)}_k + 2 \cdot \frac{1}{2^{2^m}} < r$, then $q^*_k < r$.
In the other direction, if $q^*_k < r$,  then since $r- q^*_k \geq \gamma$, and since $q^*_k \geq x^{(g)}_k$,
and since $2 \cdot \frac{1}{2^{2^m}} \geq \gamma/4$,  we have 
$x^{(g)}_k + 2 \cdot \frac{1}{2^{2^m}} \geq \gamma/4 < r$.
This completes the proof.
\end{proof} 

\begin{cor}
Given a PPS, $x=P(x)$,  with $n$ variables, and with LFP $q^* \in [0,1]^n$, 
given a coordinate $k \in \{1,\ldots,n\}$, and given a rational number $r \in [0,1]$,
there is an algorithm that
determines which of the following cases holds: (A)  $q^*_k < r$, or (B) $q^*_k = r$,  or (C)
$q^*_k > r$.\\
The algorithm runs in time polynomial in $|P|$, the bit encoding size of the PPS,
and  $size(r)$, the binary encoding size of $r$,
{\em in the unit-cost arithmetic RAM model of computation}.

Thus, in particular, deciding whether $q^*_k \geq r$ is in ${\mathbf P^{\PosSLP}}$.
Furthermore, deciding whether $q^*_k > r$, or deciding whether $q^*_k < r$, are
both P-time many-one (Karp) reducible to PosSLP.

Thus, since the problem 
of deciding whether $q^*_k > r$, and deciding whether $q^*_k < r$, are already known to be PosSLP-hard
under many-one reductions (Theorem 5.3 of \cite{rmc}),
it follows that both these problems are P-time equivalent to PosSLP.
\end{cor}
\begin{proof}
First, we note that deciding whether $q^*_k = 0$   and
whether $q^*_k = 1$, can be
carried out in P-time (\cite{rmc}).\footnote{
Determining if $q^*_k = 0$ is easier:  the P-time algorithm given in 
\cite{rmc} for this
task does not depend on the coefficients  of $x=P(x)$,  it only depends on 
which coefficients are
non-zero.  
For determining if $q^*_k = 1$, the P-time algorithm given
in \cite{rmc} uses linear programming
to determine whether
the spectral radius of certain non-negative irreducible moment matrices is $> 1$.
Note that our running time is allowed to depend on the encoding size
$|P|$ of the PPS, and not just on the number of variables, so it is ok to use linear programming
to do this.
Alternatively, it was subsequently shown in \cite{EGK10}  that the latter problem
of determining 
whether the spectral radius of a non-negative irreducible matrix is $> 1$ can be solved
in strongly polynomial time, by solving certain systems of linear equations.}
Hence,  we can detect and remove in P-time
all variables $x_i$ such that $q^*_i \in \{0,1\}$.
Then we are left with a residual PPS, $x=P(x)$, with LFP  $q^*$ such that $0 < q^* < 1$.

Notice that each iteration of Newton's method, $x^{(j+1)} = {\mathcal N}(x^{(j)}) = 
x^{(j)} + (I - B(x^{(j)}))^{-1} (P(x^{(j)}) - x^{(j)})$, on 
a PPS, $x = P(x)$ with $n$ variables, 
can be computed by performing a $n \times n$ matrix inversion and matrix-vector multiplication
and summing of vectors.
Thus, using Cramer's rule to express the matrix inverse as the ratio of matrix determinants, 
each iteration can be computed by an
arithmetic circuit over basis $\{+,-,*,/\}$
with polynomially many gates (as a function of $n$),
given the previous iteration vector $x^{(j)}$ as input. 
Thus, it can be performed by polynomially many arithmetic operations.

Now  we apply Theorem \ref{thm:deciding-in-unit-cost}.
Since the $g$ given in the statement of Theorem \ref{thm:deciding-in-unit-cost} is polynomial in $|P|$ and $size(r)$
(in fact, even in $\log(size(r))$)
we can compute $x^{(g)}$ in polynomial time in the unit-cost arithmetic RAM model of 
computation.
Likewise, since the $m$ given in the statement of the Theorem 
is also polynomial in $|P|$ and $size(r)$,  we can use repeated squaring
to compute $\frac{1}{2^{2^m}}$ in time polynomial in $|P|$ and $size(r)$
 (i.e., with polynomially many arithmetic
operations).  We can also add two numbers at unit-cost to obtain 
$x^{(g)}_k + 2 \cdot \frac{1}{2^{2^m}}$.

In order to determine whether $q^*_k > r$,   we simply need to check whether $x^{(g)}_k > r$,
and to determine whether $q^*_k < r$ we simply need to check whether  $x^{(g)}_k + 2 \cdot \frac{1}{2^{2^m}} < r$.
Finally, note that $q^*_k = r$ holds precisely when neither $q^*_k > r$  nor $q^*_k < r$ holds.

To conclude that these problems can be decided in ${\mathbf P^{PosSLP}}$, we simply 
note that it was established by Allender et. al. in \cite{ABKM06} that every discrete decision
problem (with rational valued inputs) that can be decided in P-time in the unit-cost arithmetic
RAM model of computation can be decided in ${\mathbf P^{PosSLP}}$.

Lastly, we conclude that deciding whether $q^*_k > r$, and deciding whether $q^*_k < r$,  are
 actually P-time {\em many-one} (Karp)
reducible to PosSLP.    This holds for the following reasons.  If $r \in \{0,1\}$, we
have already pointed out that deciding both $q^*_k > r$ and $q^*_k < r$
is solvable in (strongly) polynomial time (\cite{rmc,EGK10}), thus there
is nothing to prove in this case.  

 So, suppose $r \in (0,1)$, and suppose that ${\mathbf 0} < q^* < {\mathbf 1}$.
In this case, we have established that   $q^*_k > r$ if and only if $x^{(g)}_k > r$.

As shown in \cite{ABKM06} (see also \cite{rmc}), division gates
in arithmetic circuits over $\{+,-,*,/\}$
can be removed by keeping track of numerators and denominators separately.
Thus, overall (the numerator and denominator of) the rational
coordinate $x^{(g)}_k$  of the vector $x^{(g)}$ can be computed by a polynomial-sized arithmetic circuit
which can be constructed in P-time given $x=P(x)$.
Obviously the rational number $r$ can also have its numerator and denominator represented
this way in P-time.
Consequently, to decide whether $x^{(g)}_k > r$ (likewise, whether $x^{(g)}_k < r$), 
we simply need to compare the output value of two (P-time constructible) arithmetic circuits.
But PosSLP is precisely this problem, so this yields a P-time many-one reduction from
both these problems to PosSLP.
\end{proof}


\begin{thebibliography}{10}

\bibitem{AMP99}
S.~P. Abney, D.~A. McAllester, and F.~Pereira.
\newblock Relating probabilistic grammars and automata.
\newblock In {\em Proc. of 27th Meeting of the Assocation for Computational
  Linguistics (ACL'99)}, pages 542--549, 1999.

\bibitem{ABKM06}
E.~Allender, P.~B\"{u}rgisser, J.~Kjeldgaard-Pedersen, and P.~B. Miltersen.
\newblock On the complexity of numerical analysis.
\newblock {\em SIAM J. Comput.}, 38(5):1987--2006, 2009.

\bibitem{apostol74}
T.~Apostol.
\newblock {\em Mathematical Analysis}.
\newblock Addison-Wesley, 2nd edition, 1974.

\bibitem{Baker79}
J.~K. Baker.
\newblock Trainable grammars for speach recognition.
\newblock In {\em Proc. of Spring Conference of the Acoustical Society of
  America}, pages 547--550, 1979.

\bibitem{BP94}
A.~Berman and R.~J. Plemmons.
\newblock {\em Nonnegative Matrices in the Mathematical Sciences}.
\newblock Classics in Applied Mathematics. SIAM, 1994.

\bibitem{BEK}
T.~Br{\'a}zdil, J.~Esparza, and A.~Ku\v{c}era.
\newblock Analysis and prediction of the long-run behavior of probabilistic
  sequential programs with recursion.
\newblock In {\em Proc. FOCS}, pages 521--530, 2005.

\bibitem{cameron94}
P.~Cameron.
\newblock {\em Combinatorics: topics, techniques, algorithms}.
\newblock Cambridge U. Press, 1994.

\bibitem{DEKM99}
R.~Durbin, S.~R. Eddy, A.~Krogh, and G.~Mitchison.
\newblock {\em Biological Sequence Analysis: Probabilistic models of Proteins
  and Nucleic Acids}.
\newblock Cambridge U. Press, 1999.

\bibitem{EGK10}
J.~Esparza, A.~Gaiser, and S.~Kiefer.
\newblock Computing least fixed points of probabilistic systems of polynomials.
\newblock In {\em Proc. 27th STACS}, pages 359--370, 2010.

\bibitem{lfppoly}
J.~Esparza, S.~Kiefer, and M.~Luttenberger.
\newblock Computing the least fixed point of positive polynomial systems.
\newblock {\em SIAM Journal on Computing}, 39(6):2282--2355, 2010.

\bibitem{EKM}
J.~Esparza, A.~Ku\v{c}era, and R.~Mayr.
\newblock Model checking probabilistic pushdown automata.
\newblock {\em Logical Methods in Computer Science}, 2(1):1 -- 31, 2006.

\bibitem{bmdp}
K.~Etessami, A.~Stewart, and M.~Yannakakis.
\newblock Polynomial-time algorithms for branching {M}arkov decision processes
  and probabilistic min(max) polynomial {B}ellman equations.
\newblock In {\em ICALP}, 2012.
\newblock see fulll Arxiv version, arXiv:1202.4798.


\bibitem{EY-rmc-model-checking-2012}
K.~Etessami and M.~Yannakakis.
\newblock Model checking of recursive probabilistic systems.
\newblock {\em ACM Trans. Comput. Log.}, 13(2):12, 2012.


\bibitem{EWY10}
K.~Etessami, D.~Wojtczak, and M.~Yannakakis.
\newblock Quasi-birth-death processes, tree-like {QBDs}, probabilistic
  1-counter automata, and pushdown systems.
\newblock {\em Perform. Eval.}, 67(9):837--857, 2010.

\bibitem{rmc}
K.~Etessami and M.~Yannakakis.
\newblock Recursive markov chains, stochastic grammars, and monotone systems of
  nonlinear equations.
\newblock {\em Journal of the ACM}, 56(1), 2009.

\bibitem{fixp}
K.~Etessami and M.~Yannakakis.
\newblock On the complexity of {N}ash equilibria and other fixed points.
\newblock {\em SIAM Journal on Computing}, 39(6):2531--2597, 2010.

\bibitem{rmc-model-checking}
K.~Etessami and M.~Yannakakis.
\newblock Model checking of recursive probabilistic systems.
\newblock {\em ACM Transactions on Computational Logic}, 13(2), 2011 (available
  online).

\bibitem{FaginKKRRRST00}
R.~Fagin, A.~Karlin, J.~Kleinberg, P.~Raghavan, S.~Rajagopalan, R.~Rubinfeld,
  M.~Sudan, and A.~Tomkins.
\newblock Random walks with ``back buttons''.
\newblock In {\em Proc. ACM Symp. on Theory of Computing (STOC)}, pages
  484--493, 2000.

\bibitem{HJV05}
P.~Haccou, P.~Jagers, and V.~A. Vatutin.
\newblock {\em Branching Processes: Variation, Growth, and Extinction of
  Populations}.
\newblock Cambridge U. Press, 2005.

\bibitem{HLMT'12}
K.~Arnsfelt Hansen, M.~Kouck{\'y}, N.~Lauritzen, P.~Bro Miltersen, and E.~P.
  Tsigaridas.
\newblock Exact algorithms for solving stochastic games: extended abstract.
\newblock In {\em STOC}, pages 205--214, 2011.
\newblock see full Arxiv version, arXiv:1202.3898 (2012).



\bibitem{Harris63}
T.~E. Harris.
\newblock {\em The Theory of Branching Processes}.
\newblock Springer-Verlag, 1963.

\bibitem{HU79}
J.~E. Hopcroft and J.~D. Ullman.
\newblock {\em Introduction to Automata Theory, Languages, and Computation}.
\newblock Addison-Wesley, 1979.

\bibitem{HornJohnson85}
R.~A. Horn and C.~R. Johnson.
\newblock {\em Matrix Analysis}.
\newblock Cambridge University Press, 1985.

\bibitem{IsaKel66}
E.~Isaacson and H.~B. Keller.
\newblock {\em Analysis of Numerical Methods}.
\newblock J. Wiley \& Sons, 1966.

\bibitem{KA02}
M.~Kimmel and D.~E. Axelrod.
\newblock {\em Branching processes in biology}.
\newblock Springer, 2002.

\bibitem{KolSev47}
A.~N. Kolmogorov and B.~A. Sevastyanov.
\newblock The calculation of final probabilities for branching random
  processes.
\newblock {\em Doklady}, 56:783--786, 1947.
\newblock (Russian).

\bibitem{LariYoung90}
K.~Lari and S.~J. Young.
\newblock The estimation of stochastic context-free grammars using the
  inside-outside algorithm.
\newblock {\em Computer Speech \& Language}, 4(1):35 -- 56, 1990.

\bibitem{ManSch99}
C.~Manning and H.~Sch\"{u}tze.
\newblock {\em Foundations of Statistical Natural Language Processing}.
\newblock MIT Press, 1999.

\bibitem{NedSat08}
M.-J. Nederhof and G.~Satta.
\newblock Computing partition functions of {PCFGs}.
\newblock {\em Research on Language and Computation}, 6(2):139--162, 2008.

\bibitem{NedSat08b}
M.-J. Nederhof and G.~Satta.
\newblock Probabilistic parsing.
\newblock {\em New Developments in Formal Languages and Applications},
  113:229--258, 2008.

\bibitem{PazPal08}
I.~Pazsit and L.~Pal.
\newblock {\em Nuclear Fluctuations: a treatise on the physics of branching
  processes}.
\newblock Elsevier science, 2008.

\bibitem{SBHMSUH94}
Y.~Sakakibara, M.~Brown, R~Hughey, I.S.~Mian an\~d K.~Sjolander, R.~Underwood,
  and D.~Haussler.
\newblock Stochastic context-free grammars for t{RNA} modeling.
\newblock {\em Nucleic Acids Research}, 22(23):5112--5120, 1994.

\bibitem{WojEte07}
D.~Wojtczak and K.~Etessami.
\newblock Premo: an analyzer for probabilistic recursive models.
\newblock In {\em Proc. 13th Int. Conf. on Tools and Algorithms for the
  Construction and Analysis of Systems (TACAS)}, pages 66--71, 2007.

\end{thebibliography}
\end{document}